\documentclass[twocolumn,english,aps,reprint, superscriptaddress,showpacs,showkeys]{revtex4-2}
\usepackage{amsmath,amssymb,bbm,mathrsfs,bm,braket,color,graphicx,comment,todonotes, tikz, booktabs, quantikz, tabularx}
\usepackage[colorlinks,citecolor=blue,urlcolor=blue,linkcolor = blue]{hyperref}
\usepackage{soul}
\usepackage[mathscr]{euscript}
\usepackage{hyperref}
\usepackage{tikz}
\usepackage{xcolor}
\usepackage{multirow}
\usetikzlibrary{arrows.meta,positioning}

\usetikzlibrary{quantikz2}
\def\thickhline{\noalign{\hrule height.8pt}}
\newcommandx*{\labelbox}[6]{\mbox{\footnotesize$\begin{array}{c}
    \flowlabel{#1}{#2}{#3}\\\thickhline
    \flowlabel{#4}{#5}{#6}
\end{array}$}}
\newcommandx*{\mathbox}[2]{\mbox{\footnotesize$\begin{array}{c}
    #1\\\thickhline
    #2
\end{array}$}}

\usepackage{amsthm}

\newcommand{\revprod}[2]{%
  \mathop{\prod}\limits_{#1}^{%
    \substack{
      \longleftarrow\\[-0.1ex]
      #2
    }
  }%
}

\newtheorem{theorem}{Theorem}
\newtheorem{lemma}[theorem]{Lemma}
\newtheorem{corollary}[theorem]{Corollary}
\newtheorem{definition}[theorem]{Definition}

\usepackage[ruled,vlined,linesnumbered]{algorithm2e}

\SetKwInput{KwIn}{Input}
\SetKwInput{KwOut}{Output}
\SetKw{KwRet}{Return}
\SetKwFunction{Nonzero}{nonzero}
\SetAlgoNoEnd

\begin{document}

\title{Fermion lattices can be simulated by same-size qubit lattices\\ with $\mathcal{O}(1)$ interaction overhead}

\author{Gregor Aigner}
\thanks{
These authors contributed equally to this work.\\
Corresponding authors:\\
gregor.aigner@uibk.ac.at, berend.klaver@uibk.ac.at
}
\affiliation{Institute for Theoretical Physics, University of Innsbruck, A-6020 Innsbruck, Austria}
\affiliation{Parity Quantum Computing GmbH, A-6020 Innsbruck, Austria}

\author{Berend Klaver}
\thanks{
These authors contributed equally to this work.\\
Corresponding authors:\\
gregor.aigner@uibk.ac.at, berend.klaver@uibk.ac.at
}
\affiliation{Institute for Theoretical Physics, University of Innsbruck, A-6020 Innsbruck, Austria}
\affiliation{Parity Quantum Computing GmbH, A-6020 Innsbruck, Austria}

\author{Martin Lanthaler}
\affiliation{Institute for Theoretical Physics, University of Innsbruck, A-6020 Innsbruck, Austria}
\affiliation{Parity Quantum Computing GmbH, A-6020 Innsbruck, Austria}

\author{Wolfgang Lechner}
\affiliation{Institute for Theoretical Physics, University of Innsbruck, A-6020 Innsbruck, Austria}
\affiliation{Parity Quantum Computing GmbH, A-6020 Innsbruck, Austria}

\begin{abstract}
   Local interactions among electrons underlie many complex properties of correlated materials. 
   While the Jordan-Wigner transformation can preserve this locality along one spatial dimension, interactions along the remaining dimensions typically incur substantial overhead.
   We show how to simulate all geometrically local interactions on an $N$-site two-dimensional fermion lattice with no asymptotic overhead in the number of interactions and no space overhead. 
   The primary overhead of our method is circuit depth, which on a qubit lattice matches that of fermionic swap networks, scaling as $\mathcal{O}(\sqrt{N})$, but reduces to $\mathcal{O}(\log N)$ on reconfigurable qubit arrays and to $\mathcal{O}(1)$ in lattice-surgery-based surface-code architectures. 
   This is enabled by dynamically reorienting the Jordan-Wigner transformation to switch the lattice dimension along which locality is preserved.
   Furthermore, we study fermion routing, as required for the simulation of non-local interactions. When using qubit lattices, we reach resource scaling that asymptotically matches that of qubit routing, whilst on fully connected qubit devices, a depth scaling arbitrarily close to $\mathcal{O}(\log N)$ is reached. This allows the fermionic fast Fourier transform to be implemented on qubit lattices with asymptotically optimal resource scaling under these locality constraints.
   Notably, all of our constructions naturally extend to $d$-dimensional lattices. 
   Beyond scaling improvements, we show explicit examples of our method, including Fermi-Hubbard-model simulations of the square-, Lieb- and kagome lattice and the fermionic fast Fourier transform.
\end{abstract}

\maketitle

\begin{figure*}
    \centering
    \includegraphics[width = \linewidth]{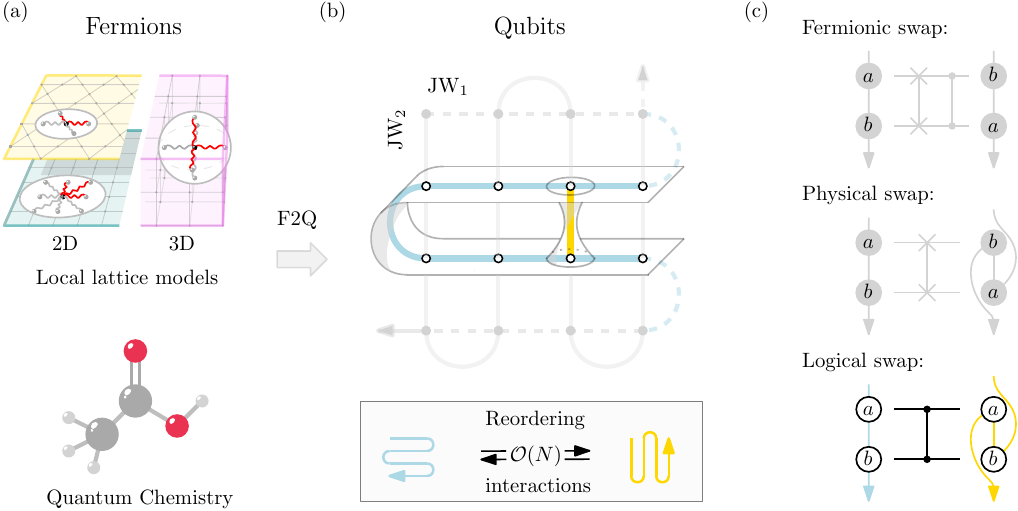}
    \caption{
    Overview of the main concepts. 
    (a) Simulating fermionic systems, such as lattice models or quantum-chemical systems, on qubit-based quantum hardware requires F2Q mappings.
    (b) The prototypical JW mapping imposes an ordering, represented here by arrows. 
    To overcome the corresponding one-dimensional connectivity constraint and better exploit typical two-dimensional hardware connectivity, we employ dynamical reordering, i.e., switching between different JW orderings, e.g., $\mathrm{JW}_1$ and $\mathrm{JW}_2$.
    (c) Reorder primitives: 
    A fermionic swap gate swaps both their JW order and their physical positions of two qubits. 
    In contrast, a physical qubit swap changes only the physical positions, while preserving their JW ordering. 
    A CZ gate, on the other hand, changes only the JW ordering, 
    leaving the physical positions unchanged. 
    Since this ``logical'' swap resembles the methods developed in this work, it is highlighted in color.}
    \label{fig:1}
\end{figure*}

In strongly correlated electronic systems, complex macroscopic phenomena such as pseudogaps, strange metals, and high-temperature superconductivity often emerge from simple local interactions on two-dimensional lattices~\cite{arovasHubbardModel2022a}.
Characterizing the interplay of these interactions and understanding how they give rise to such complex macroscopic phenomena is a long-standing objective in the field~\cite{qinHubbardModelComputational2022}, yet their simulation remains notoriously difficult for classical methods.
Quantum-computer-based simulations offer a natural route towards building such understanding~\cite{feynmanSimulatingPhysicsComputers1982}. 
However, as contemporary quantum hardware is primarily qubit-based, they typically require fermion-to-qubit (F2Q) mappings~\cite{bravyiFermionicQuantumComputation2002, verstraeteMappingLocalHamiltonians2005a, derbyCompactFermionQubit2021}.
The prototypical example is the Jordan-Wigner (JW) encoding, which represents fermionic operators as qubit operators, at the expense of imposing a rigid one-dimensional connectivity constraint. 
Consequently, while the JW encoding is well-suited for one-dimensional fermionic systems, simulating two- or higher-dimensional lattices results in non-local interactions that incur significant overhead.

To mitigate this overhead, some methods dynamically adjust the encoding to simplify target interactions. 
Commonly used methods employ fermionic swap networks (FSN), as they require only local interactions by reordering fermionic modes through local swap operations~\cite{kivlichanQuantumSimulationElectronic2018}.
Yet, this locality restricts the speed of encoding adjustments, causing FSNs to induce an $\mathcal{O}(\sqrt{N})$ overhead in depth and gate count for simulations of two-dimensional lattices.
Recent results have shown that arbitrary modifications of the JW encoding can be implemented in polylogarithmic depth.
However, these methods rely either on reconfigurable qubit arrays with $\mathcal{O}(N)$ ancillas~\cite{maskaraFastSimulationFermions2025}, or on fully connected qubit systems~\cite{constantinidesLowdepthFermionRouting2025}.
This stands in stark contrast to the simple local interaction structure of typical fermion lattices.

In this work, we introduce dynamic fermionic encoding primitives that enable efficient implementations on qubit lattices, and we demonstrate their usefulness in several applications. 
We show how to mitigate the one-dimensional connectivity restriction, inherent to the JW encoding, by switching to a complementary second encoding, as illustrated in Fig.~\ref{fig:1}. 
Effectively, this enables fermionic modes that would be distant under a single JW encoding to interact using a second physical dimension of a qubit lattice.
We illustrate in Fig.~\ref{fig:2} and prove in the appendix, using the parity flow formalism \cite{klaverParityFlowFormalism2025}, that transitions between such encodings can be performed efficiently.

As an important use-case, we utilize efficient switching between two complementary JW encodings, to simulate all fermionic interactions of a two-dimensional lattice on a qubit lattice with only $\mathcal{O}(1)$ interaction overhead.
The corresponding depth scaling follows the light-cone limit of the qubit lattice, $\mathcal{O}(\sqrt{N})$, and thus matches the depth of FSNs. 
Furthermore, we demonstrate that this simulation can be performed via lattice surgery based surface codes in $O(1)$ depth.
On reconfigurable arrays of qubits, we reach $O(\log N)$ depth scaling, and therefore, match the respective light-cone limits of all these settings.
%
Although we focus on two-dimensional lattices in the main text, we extend to arbitrary dimensions in Appendix~\ref{app:dD}. 
There, we show that both the constant qubit interaction overhead and the light-cone scaling in depth $\mathcal{O}(N^{1/d})$ generalize to $d$-dimensional fermion lattices simulated on $d$-dimensional qubit lattices.

Beyond simulating fermion lattices, our work lowers the resource cost for routing of fermions. 
This enables the fermionic fast Fourier transform (\textsf{FFFT}), an important primitive for state preparation~\cite{verstraeteQuantumCircuitsStrongly2009, maskaraFastSimulationFermions2025} and material simulation algorithms~\cite{babbushLowDepthQuantumSimulation2018}, to be implemented on qubit lattices with the same asymptotic resource cost as on fermion lattices.
Furthermore, we extend our fermion routing primitives for fully-connected qubit systems and lower the associated asymptotic scaling of resources compared to the state of the art.

In the following, we briefly recap the JW transformation and dynamical encodings schemes in Sec.~\ref{sec:background}. 
We then introduce our methods in Sec.~\ref{sec:2D_boustrophedon} and showcase a variety of applications in Sec.~\ref{sec:applications}.
These applications include explicit Trotter circuits, implemented on qubit lattices, for simulating Fermi-Hubbard models on square lattices with nearest-neighbor (NN) and next-nearest-neighbor (NNN) interactions, as well as on Lieb and kagome lattices (Sec.~\ref{sec:fermihubbard}).
Section~\ref{sec:ffft} discusses how our methods enable more efficient implementations of the \textsf{FFFT} on qubit architectures with local connectivity.
Finally, in Sec.~\ref{sec:molecularsim}, we examine how our fermion-routing techniques can be utilized in electronic structure simulations of molecular systems.

\section{Dynamical F2Q maps}
\label{sec:background}

\begin{figure*}
    \centering
    \includegraphics[width = \linewidth]{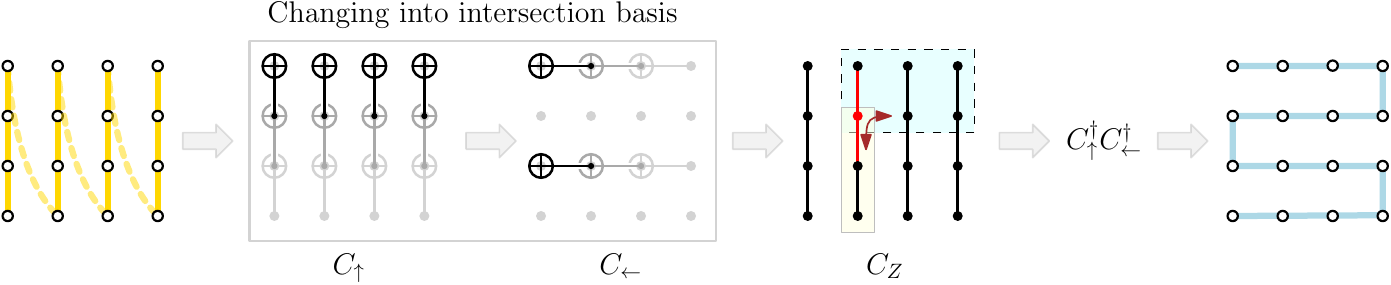}
    \caption{
    Example of JW reordering in two dimensions from a $Z$-pattern (left) to a $S$-pattern (right).
    The first step consists of a basis change accomplished by CNOT ladders along each column, denoted by $C_{\uparrow}$, followed by CNOT ladders along every second row, denoted by $C_{\leftarrow}$. 
    The shading indicates the gate sequence, progressing from grey to black. 
    At a high level, this basis change accumulates information first along columns and then along every second row. 
    Although the row-wise accumulation partially mixes with the column-wise accumulation, the row information remains encoded in the information-difference between qubits separated by two rows.
    At the intersections where the propagated information crosses, controlled-Z operations allow column-associated information to be exchanged with row-associated information. This enables the local switching of row- and column-originating information at each lattice site, facilitating transitions between column-first and row-first JW paths.
    After undoing the initial basis transformation, and an application of single-qubit Z gates to all qubits in every second row starting from the second row (not depicted),
    the canonical ordering follows a row-wise $S$-pattern. 
    }
    \label{fig:2}
\end{figure*}

The JW encoding is the canonical example of a F2Q mapping. 
It is highly flexible, as a specific encoding is defined by assigning an unique rank $j \in \{0, \ldots N-1\}$ to each qubit, allowing for $N!$ possible canonical orderings.
The translation between fermionic operations and qubit operations is given by

\begin{equation*}
\begin{aligned}
c_j^\dag & \equiv  \frac{1}{2} (X_j-iY_j) Z_0 \cdots Z_{j-1}\\  c_j & \equiv  \frac{1}{2} (X_j+iY_j) Z_0 \cdots Z_{j-1}.    
\end{aligned}
\end{equation*}
For instance, under JW, the hopping terms

\begin{equation*}
c_i^\dag c_j + c_j^\dag c_i = \frac{1}{2}(X_iX_j+Y_iY_j)Z_{i+1}\cdots Z_{j-1}
\end{equation*}
map to $k$-local qubit operators with
\begin{equation*}
    k=|i-j|+1,
\end{equation*}
spanning all qubits between ranks $i$ and $j$ along the canonical ordering. 
Thus, a natural picture of a particular JW ordering is as a one-dimensional path that visits each qubit exactly once, with any mapped hopping between modes $i$ and $j$ involving all qubits lying between them along the path.
In a two-dimensional lattice, for instance, two fermionic modes may be geometrically close yet highly separated along the JW path. 
This discrepancy results in long strings of $Z$ operators that significantly increase simulation complexity.
To mitigate this, dynamic encoding primitives allow for adjustments of the canonical ordering so that target interactions become local. 
In the simplest case, these adjustments are constructed from fermionic swaps (FSWAPs), which exchange the canonical ranks and positions (or state) of two qubits
\begin{equation*}
c_i \leftrightarrow c_{i+1}, \qquad
c_i^\dagger \leftrightarrow c_{i+1}^\dagger.
\end{equation*}
By composing these gates into fermionic swap networks (FSNs), one can systematically reconfigure the JW encoding so that each desired fermionic interaction is brought onto neighboring qubits along the JW path, replacing long JW strings by a sequence of fermionic swaps.

For a fully connected problem with two-body interactions, FSNs can be constructed with optimal parallelizability, leading to asymptotically optimal scaling by ensuring every pair of fermionic modes becomes adjacent exactly once \cite{kivlichanQuantumSimulationElectronic2018}.
However, for geometrically local fermionic systems, FSNs still incur substantial overhead. For example, simulating two-dimensional fermion lattices using standard JW-based swap-network constructions requires $\mathcal{O}(N^{\frac{3}{2}})$ local interactions and $\mathcal{O}(\sqrt{N})$ depth, even on two-dimensional qubit lattices.
This represents an $\mathcal{O}(\sqrt{N})$ overhead relative to fermionic quantum processors, which, lacking the need for an F2Q mapping, can achieve the same simulation with $\mathcal{O}(N)$ local interactions at $\mathcal{O}(1)$ depth. 
Thus, the overhead scaling of FSNs reflects the general limitation imposed by the JW mapping, which restricts the system to an effectively one-dimensional connectivity and thereby prevents full exploitation of the underlying qubit-lattice connectivity.

\begin{figure*}
    \centering
    \includegraphics[width = .95\linewidth]{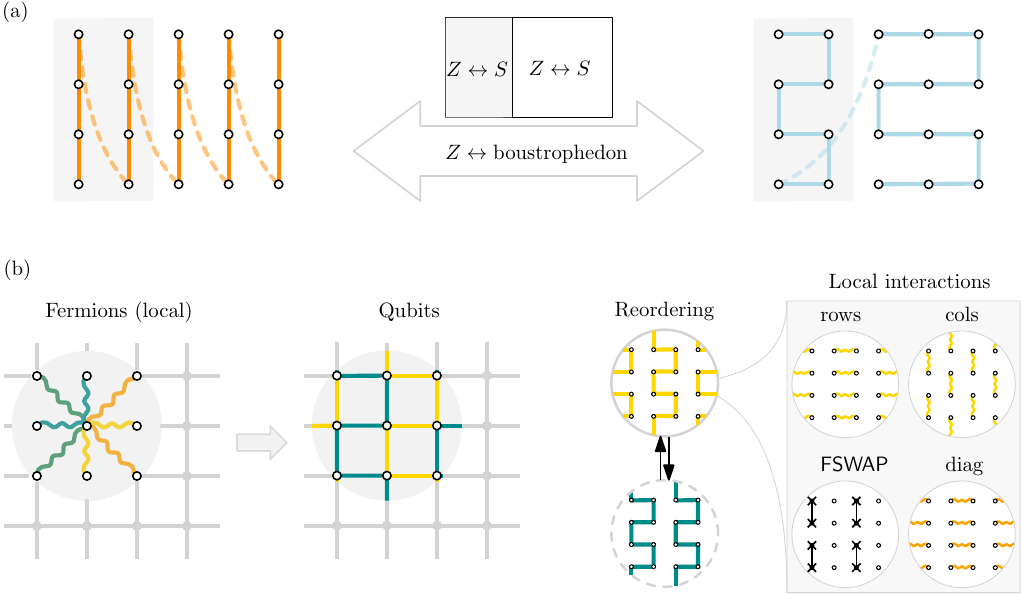}
    \caption{
    Switching between two-dimensional boustrophedon encodings.
    (a)
    A global $Z$-pattern is transferred into a boustrophedon pattern by partitioning the columns of the grid. 
    Applying the $Z$-to-$S$ transformation circuit depicted in Fig.~\ref{fig:2} to these parts individually yields a series of interconnected $S$-patterns, a configuration we refer to as a boustrophedon pattern due to its alternating serpentine structure. 
    An example is shown where a partition is introduced between columns two and three, counting from the left.
    (b)
    As detailed in Appendix~\ref{app:locality_conservation}, two such complementary patterns are sufficient to cover all nearest-neighbor fermionic interactions on the lattice.
    In general this allows to map geometrically local fermionic interactions to local qubit gates.
    }
    \label{fig:3}
\end{figure*}

\section{Complementary encodings}
\label{sec:2D_boustrophedon}
To overcome the one-dimensional connectivity constraint imposed by the JW mapping, we introduce a complementary second JW encoding that is designed to bridge the direction orthogonal to the first JW path [see Fig.~\ref{fig:1}(b)]. 
For instance, the second encoding can transform a column-first ordering into a row-first ordering, providing an alternative dimensional hierarchy.
Alternating between these two encodings allows one to exploit the connectivity along the second spatial dimension, thereby avoiding the large overhead that would arise working within a single encoding.
A measurement-based implementation of such a permutation is also introduced in Ref.~\cite{maskaraFastSimulationFermions2025}.

One particularly relevant example is depicted in Fig.~\ref{fig:2}, where a $Z$-type encoding is complemented by an $S$-type encoding. 
Because the $Z$ encoding ensures that qubits are column-wise adjacent with respect to the JW ordering, it allows for the implementation of nearest-neighbor (NN) column interactions using simple two-body operators. 
Similarly, the $S$-encoding facilitates the implementation of row-wise interactions. 
This $Z$-to-$S$ transformation provides the most straightforward example of reconfiguring the dimensional hierarchy of the JW encoding, the details of which are provided in Appendix~\ref{app:2d_circuits}.
Taken together, these dual encodings realize an effective two-dimensional connectivity, thereby preserving the locality of two-dimensional fermionic interactions and ensuring that all lattice-adjacent fermionic modes can stay adjacent in the qubit encoding.

\subsection*{Two-dimensional lattices}

While there exist multiple sets of complementary encodings that can compensate for one another’s connectivity restrictions, they do not all incur the same resource costs for a given task. 
Since switching between encodings is generally expensive, the goal is to identify a minimal set of complementary mappings that collectively cover all target interactions and are easy to transform into one another. 
To this end we introduce the family of \textit{boustrophedon} patterns~\footnote{The term \textit{boustrophedon} originates from Greek and literally means ``ox-turning'', referring to the alternating left-right pattern of an ox plowing a field. In computational geometry, boustrophedon decomposition is a standard method for partitioning and traversing planar space in a back-and-forth sweep pattern.} 
of the JW path, which allow for systematic coverage of lattice interactions with minimal overhead.
More precisely, a boustrophedon encoding is defined by a colum-wise partition of the qubit lattice.
Within each subset, the JW path follows a snake-like traversal, such that the concatenation of these local paths covers the entire qubit lattice. 
This construction allows for flexible reordering strategies; several examples of such decompositions are illustrated in Fig.~\ref{fig:3}.

The transformation between the $Z$-type encoding to an arbitrary boustrophedon pattern is depicted in Fig.~\ref{fig:3}(a) and is executed by applying the $Z$-to-$S$ transformation to each subset of qubits. 
This approach also enables transformations between any two boustrophedon patterns, as one can first revert to the $Z$-pattern and subsequently transition to the second boustrophedon.
Consequently, after accounting for straightforward CNOT cancellations, any boustrophedon-to-boustrophedon transformation can be realized with a two-qubit gate count bounded by
\begin{equation*}
    C < 6N+\mathcal{O}(\sqrt{N})
\end{equation*}
and depth bounded by
\begin{equation*}
    T < 6\sqrt{N}+\mathcal{O}(1).
\end{equation*}
Here, the depth exhibits the same asymptotic scaling as the system’s light cone, which bounds the propagation speed of the CNOT ladders shown in Fig.~\ref{fig:2} along each lattice axis. 
In lattice-surgery-based surface-code architectures or hardware with higher connectivity, these CNOT ladders can be accelerated dramatically, as detailed for the lattice-surgery setting in Appendix~\ref{app:latticeSurgery}. 
For the fully connected qubit setting, an efficient decomposition of CNOT ladders with $\mathcal{O}(\log N)$ depth is provided in Ref.~\cite{gluckReversibleComputation17th2025}. 
Notably, these logarithmic-depth ladders only necessitate non-local connectivity between qubits within the same row or column. 
This requirement aligns well with realistic neutral-atom hardware, where non-local row- and column-wise connectivity is natively supported via acousto-optical deflectors~\cite{xuConstantoverheadFaulttolerantQuantum2024}.

One important finding of this work is that the combination of two complementary boustrophedon encodings preserves the locality of the fermionic interactions.
At a high level, this means that if the range of fermionic interactions is bounded, these interactions can be simulated using two boustrophedon encodings on qubits with an interaction range of comparable size. 
A formal discussion of this locality preservation is provided in Theorem~\ref{thm:LocalImplementation}.
This local property has an important consequence: increasing the interaction range of the simulated fermion lattice contributes only a constant additive overhead to the total depth. 
This is in contrast to FSNs, for which increasing the interaction range increases the prefactor of the $\mathcal{O}(\sqrt{N})$ depth scaling. 
Further details on this are given in Corollary~\ref{cor:AddingInteractions}.

We generalize our methods to arbitrary spatial dimensions in Appendix~\ref{app:dD}. 
This extension enables, for example, the simulation of three-dimensional fermion lattices on fully connected qubit-based hardware with circuit depth $\mathcal{O}(\log N)$. 

\subsection*{Routing of fermions}

\begin{figure}
    \centering
    \includegraphics[width = \columnwidth]
    {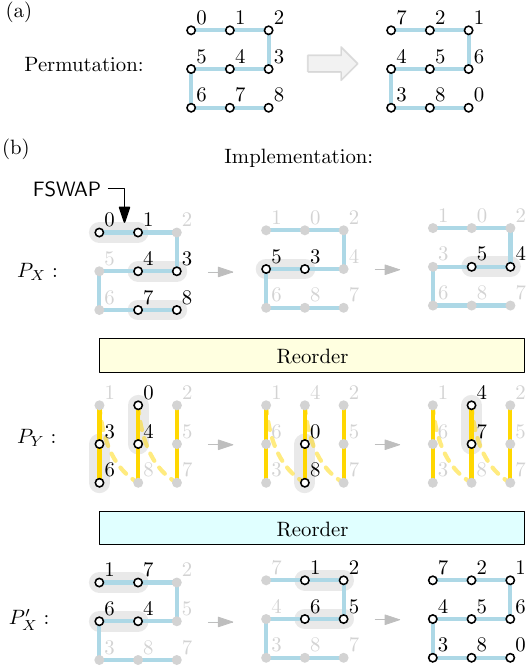}
    \caption{
    Fermion routing.  (a) An example routing problem. 
    (b) Starting from an $S$-pattern ordering, the first row-wise stage $P_X$ is executed using fermionic swaps. Next we switch to a $Z$-pattern, enabling the column-wise stage $P_Y$ to be realized with fermionic swaps. Finally, the original $S$-pattern is restored, which is appropriate for the concluding row-wise stage $P_X'$. Each permutation stage, as well as the reordering, can be implemented in $O(\sqrt{N})$ depth, reaching an overall optimal asymptotic scaling.}
    \label{fig:ArbitraryPermutation}
\end{figure}

In order to facilitate geometrically non-local interactions, we combine our dynamical encoding primitives with a standard algorithm for off-line permutation routing~\cite{annexsteinUnifiedApproachOffline1990}, which is known to achieve optimal asymptotic depth scaling for general data movement on a $d$-dimensional  lattice.
In this approach, one realizes an arbitrary two-dimensional permutation $P$ by decomposing it into a sequence of row-wise and column-wise permutations. 
The decomposition can be written as
\begin{equation*}
    P = P_X \circ P_Y \circ P_X',
\end{equation*}
where $P_X$ and $P_X'$ denote permutations within rows, and $P_Y$ denotes a permutation along columns. 

Our key observation is that dynamic switching between row-first and column-first patterns makes this classical approach directly compatible with row- and column-wise FSNs. 
As visualized in Fig.~\ref{fig:ArbitraryPermutation}, each pattern aligns the effective one-dimensional connectivity with a different set of local directions on the lattice. 
By alternating between these patterns at the appropriate points in the sequence $P_{X}\rightarrow P_{Y}\rightarrow P_{X'}$, we can implement each sub-routing using FSNs.
Crucially, the encoding switches can be performed with sufficiently low cost that they do not alter the asymptotic resource cost of the underlying lattice-permutation algorithm. 

In Appendix~\ref{app:fermion_routing}, we provide additional details and extend the construction to lattices of arbitrary dimension. 
In particular, we use that routing on a $d$-dimensional lattice can be performed analogously to the two-dimensional case without asymptotic scaling introduced by the F2Q mapping, yielding an overall depth of $\mathcal{O}(N^{1/d})$. 
Furthermore, we study fully connected qubit systems, where for each choice of $a \in \{1,2,3,\ldots\}$ we construct a corresponding routing algorithm with depth $\mathcal{O}\big(\log^{(a+1)/a} N\big)$. 
This improves upon the asymptotic scaling of Ref.~\cite{constantinidesLowdepthFermionRouting2025} and shows that, even without ancilla qubits, logarithmic asymptotic depth scaling can be approached arbitrarily closely.
However, in the more practically relevant setting, treating a fully connected qubit system as a $d$-dimensional lattice, with $d$ chosen flexibly, makes it possible to trade lower constant factors against asymptotic scaling advantages for specific problem sizes. 

\section{Example applications}
\label{sec:applications}

In this section, we consider several widely studied fermionic problems and corresponding algorithms. 
We show that incorporating our dynamic encoding primitives reduces resource requirements, even at scales relevant to current devices.

\begin{figure*}
    \centering
    \includegraphics[width = \linewidth]{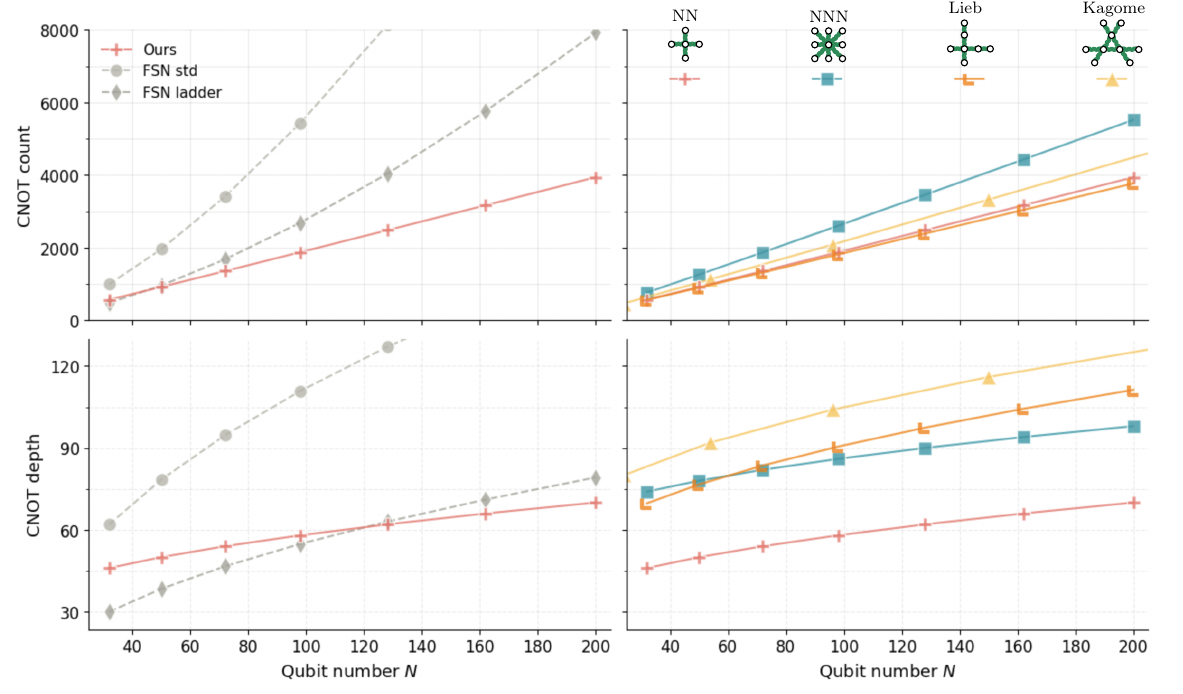}
    \caption{
    Resource counts of Fermi-Hubbard simulations.
    The plots show CNOT counts (upper panels) and depths (lower panels) for a second-order Trotter step simulating various Hubbard models on a two-dimensional square lattice of qubits.
    Models: We study nearest-neighbor (NN) and next-nearest-neighbor (NNN) square lattices, as well as Lieb and kagome geometries.
    For NN and NNN models, we utilize an $L\times (2L)$ qubit lattice (where the factor of $2$ accounts for spin and $L^2$ is the number of unit cells) for sizes $L=4$ to $10$.
    For Lieb and kagome lattices, we employ a $3L\times 2L$ qubit lattice (where the factor of $3$ reflects the three-site unit cell) for sizes $L=2$ to $5$.
    Performance Comparison (Left): We compare our implementation on the NN square lattice against the standard fermionic swap (FSWAP) network \cite{kivlichanQuantumSimulationElectronic2018} and an optimized FSWAP version utilizing a ladder-like connectivity graph, as employed in recent experimental works \cite{hemeryMeasuringLoschmidtAmplitude2024, alamProgrammableDigitalQuantum2025a}.
    Complex Geometries (Right): We present the resource counts required to implement the more complex interaction graphs of the NNN, Lieb, and kagome models using our dynamical encoding framework.
    }
    \label{fig:HubbardComp}
\end{figure*}

\subsection{2D Fermi-Hubbard models}
\label{sec:fermihubbard}
The Fermi-Hubbard model stands as a foundational framework in quantum many-body physics. Despite its conceptual simplicity, the intricate competition between kinetic energy and strong electronic correlations gives rise to a remarkably rich phase diagram. This includes Mott insulators, pseudogap and strange-metal regimes,  antiferromagnetism and exotic quantum spin liquids and, most prominently, unconventional superconductivity.
Since the discovery of high-temperature cuprate superconductors, a central question has been whether the two-dimensional Hubbard model captures the essential physics required to describe these materials.
Growing evidence suggests that the ground state of the Hubbard model in its simplest form, featuring a square lattice with only NN hopping, is non-superconducting in the parameter regime most relevant to cuprates~\cite{qinHubbardModelComputational2022, simonscollaborationonthemany-electronproblemAbsenceSuperconductivityPure2020}. 
However, including NNN hopping has been shown to exert a strong influence on the stability of the superconducting phase, making this model a more promising candidate to study the basic mechanisms of curate high-temperature superconductivity~\cite{jiangSuperconductivityDopedHubbard2019, simonscollaborationonthemany-electronproblemAbsenceSuperconductivityPure2020, tasakiHubbardModelIntroduction1998}. 
A microscopically more accurate description of the cuprates uses a three-band model to capture the essential physics of copper and oxygen electron orbitals. This so-called Emery model lives on the more complicated Lieb lattice, a square lattice with three sites per unit-cell, making classical simulations even more challenging~\cite{zhaoEnhancedSuperconductingCorrelations2025, cuiGroundstatePhaseDiagram2020, maiPairingCorrelationsCuprates2021}. 
Beyond cuprate physics, the Hubbard model exhibits a rich variety of properties that depend critically on the underlying lattice geometry; for instance, the frustrated nature of the kagome lattice gives rise to a wide range of correlated phenomena~\cite{huRichNatureVan2022, liaoGaplessSpinLiquidGround2017}.

Despite the extensive research of the Hubbard model triggered by these rich phenomena, many fundamental questions remain unresolved, primarily due to the inherent computational complexity for classical methods. In the following, we demonstrate how the methods developed in this paper allow for efficient implementations of the Hubbard model on many different lattice geometries onto qubit lattices.

The general Hubbard Hamiltonian under consideration is given by
\begin{gather*}
H=-t\sum_{\langle i,j\rangle,\sigma}\left(c^\dagger_{i\sigma}c_{j\sigma}+\mathrm{h.c.}\right)
-t'\sum_{\langle\!\langle i,j\rangle\!\rangle,\sigma}\left(c^\dagger_{i\sigma}c_{j\sigma}+\mathrm{h.c.}\right)\\
+U\sum_i n_{i\uparrow}n_{i\downarrow},
\end{gather*}
where the parameters $t$ and $t'$ denote the hopping amplitudes and $U$ represents the on-site Coloumb repulsion.
The summations $\langle i,j\rangle$ and $\langle\!\langle i,j\rangle\!\rangle$ refers to the NN and NNN pairs of the underlying lattice, respectively. 
We primarily consider examples with exclusively NN interactions, but additionally study square lattices with NNN hopping $t'\neq 0$.

By utilizing two complementary boustrophedon encodings, all local hopping terms can be implemented such that each encoding captures one half of the interactions through local qubit gates. 
This approach necessitates only a single switch between the two encodings to simulate the full Hamiltonian. 
Our approach is illustrated explicitly for the NNN square lattice in Fig.~\ref{fig:3}(b).

To assess the competitiveness of our approach for near-term quantum devices, Fig.~\ref{fig:HubbardComp} compares the resource cost of implementing a single second-order Trotter step for the NN square lattice using our method with that of two state-of-the-art FSN-based approaches.
We find that our scaling advantage already yields lower gate counts for $5\times5$ lattices, which correspond to a total of $50$ qubits considering both spin species. 
Furthermore, while the boustrophedon strategy and FSNs exhibit the same asymptotic depth scaling, $\mathcal{O}(\sqrt{N})$, the significant smaller prefactor in our approach allows it to outperform FSNs for lattices larger than $7\times7$.

Beyond the standard square lattice, our method maintains these competitive resource costs even as the underlying fermionic connectivity increases in complexity, as illustrated in Fig.~\ref{fig:HubbardComp}.
This resilience is a direct consequence of the additive $\mathcal{O}(1)$ overhead in both gate count and depth required to accommodate longer-range interactions.
This property allows our approach to scale well to more sophisticated models, such as the square lattice with NNN interactions, the Lieb lattice and the kagome lattice.
Further details to these examples are provided in Appendix~\ref{app:lattice_examples}.

{
\renewcommand{\arraystretch}{1.3}
\begin{table}[t]
\centering
\begin{tabular}{l c c}
\toprule
Lattice & Gate Count & Depth \\
\midrule
Square NN & $21N$ & $4\sqrt{N/2}$\\
Square NNN & $30N$ & $4\sqrt{N/2}$\\
Lieb & $19\frac{2}{3}N$ & $12\sqrt{N/6}$\\
Kagome & $24N$ & $12\sqrt{N/6}$\\
\bottomrule
\end{tabular}
\caption{
Leading-order scaling of circuit gate counts and depths for a single second-order Trotter step across various fermionic lattice models. We consider spinful $L\times L$ realizations of each geometry implemented on a two-dimensional qubit lattice.
}
\label{tab:gate-counts}
\end{table}
}

Crucially, the total depth for all presented examples is dominated by the depth of the encoding switch. Consequently, our approach benefits directly from architectures that facilitate efficient boustrophedon switches. 
A particularly relevant case is the surface code~\cite{kitaev2003fault, fowler2012surface}, arguably the most prominent error-correcting code, which has recently demonstrated logical error probabilities beyond the break-even point~\cite{google2025quantum}.
A promising method for implementing universal fault-tolerant operations in these systems is lattice surgery~\cite{horsman2012surface, fowler2018low, litinskiGameSurfaceCodes2019}. 
In this setting, the depth for simulating fermion lattices reduces to $\mathcal{O}(1)$. 
This reduction in depth is possible because the CNOT ladders (the primary bottleneck of our encoding switch) can be reordered with trivial overhead using feed-forward operations within the lattice-surgery framework. 
We provide a comprehensive discussion of this approach in Appendix~\ref{app:latticeSurgery}.

In many contemporary cost models for fault-tolerant quantum circuits, the overhead of Clifford operations required for dynamic encoding changes, whether via FSNs or our proposed method, is often neglected. 
It remains unclear, however, whether this assumption will remain justified as the field progresses. 
The relative cost of non-Clifford operations continues to shift as magic-state synthesis becomes more efficient~\cite{gidneyMagicStateCultivation2024} and new techniques substantially reduce the T-count required for Hubbard model simulations~\cite{kivlichanImprovedFaultTolerantQuantum2020a}.
This shift is particularly relevant given that a single Trotter step using FSNs requires $\mathcal{O}(N)$  T gates but $\mathcal{O}(N^{3/2})$ CNOTs. 
In this regime, reducing the CNOT complexity to $\mathcal{O}(N)$  could have a non-negligible impact on the total resource requirements.

\subsection{Fermionic fast Fourier transform}
\label{sec:ffft}

\begin{figure}
    \centering
    \includegraphics[width = \columnwidth]{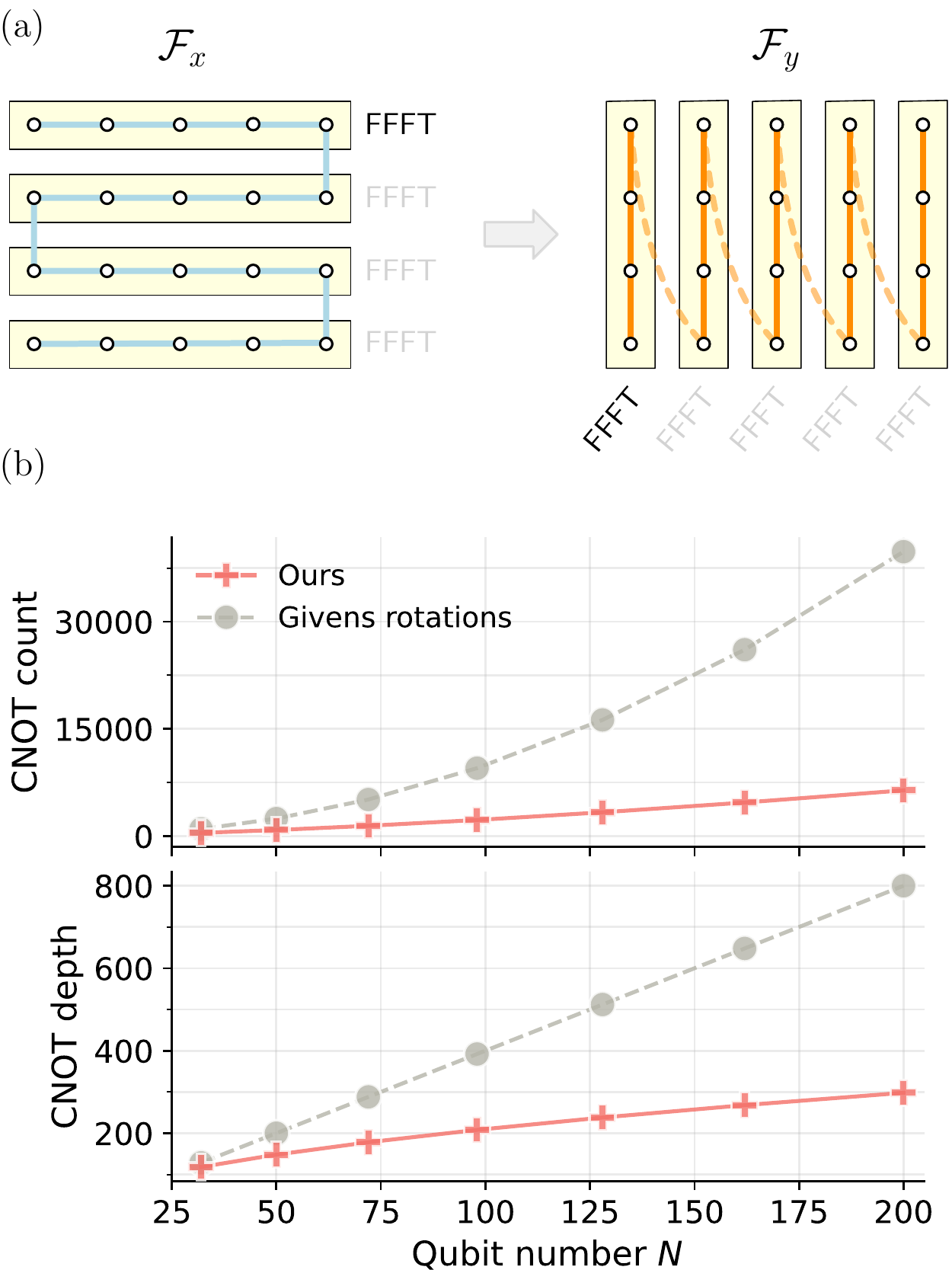}
    \caption{
    Two-dimensional \textsf{FFFT} construction via one-dimensional transforms.
    (a) Conceptual framework. 
    The 2D \textsf{FFFT} is decomposed into successive one-dimensional transforms, $\mathcal{F}=\mathcal{F}_x\mathcal{F}_y$.
    These act along the $x$ and $y$ dimensions and are efficiently parallelized using $S$ and $Z$ patterns, respectively. For spinful systems, the spin-up and spin-down species are interleaved along the columns; consequently, the 1D \textsf{FFFT} is adapted according to the construction detailed in Appendix~\ref{app:lattice_examples}.
    (b) Numerical performance. 
    Benchmarking the spinful 2D \textsf{FFFT} within our framework against the standard approach based on Givens rotations from Ref.~\cite{kivlichanQuantumSimulationElectronic2018}. 
    }
    \label{fig:FFFT}
\end{figure}

As a further application, we investigate algorithms built upon the \textsf{FFFT}, which finds widespread use across digital quantum simulation. 
For instance, translationally invariant free-fermion states can be prepared in constant depth in momentum space and subsequently transformed to real space via the \textsf{FFFT}~\cite{verstraeteQuantumCircuitsStrongly2009}.
This is a crucial step in preparing states with long-range structure, including chiral superconductors, critical states, and Fermi surfaces~\cite{maskaraFastSimulationFermions2025}. 
Furthermore, the \textsf{FFFT} serves as a core subroutine in state-of-the-art Hamiltonian simulation algorithms for periodic materials~\cite{babbushLowDepthQuantumSimulation2018}.

Although the $d$-dimensional \textsf{FFFT} admits an implementation with depth $\mathcal{O}(\log^2 N)$ on fully connected architectures~\cite{constantinidesLowdepthFermionRouting2025}, devices with limited-connectivity necessarily incur additional depth overhead. 
For reference, a $d$-dimensional fermionic quantum computer that is limited to local interactions can implement a $d$-dimensional \textsf{FFFT} in $\mathcal{O}(N^{\frac{1}{d}})$ depth. 
This scaling is inherited from the one-dimensional \textsf{FFFT}s that compose the higher-dimensional transformation, as illustrated in Fig.~\ref{fig:FFFT}. 
Consequently, the central challenge lies in implementing the $d$-dimensional \textsf{FFFT} on a qubit lattice without introducing overhead beyond the fundamental constraints of local connectivity.

By leveraging our dynamical encoding primitives, we can effectively implement the intermediate reorientations required to apply successive one-dimensional \textsf{FFFT}s while maintaining optimal resource scaling. 
This allows the entire $d$-dimensional \textsf{FFFT} to be executed with $\mathcal{O}(N^\frac{1}{d})$ depth on a corresponding qubit lattice.
The conceptual framework for this algorithm for two dimensions is illustrated in Fig.~\ref{fig:FFFT}(a), while panel (b) provides a comparison to the state-of-the-art method for restricted connectivity devices.
Detailed derivations for the higher-dimensional construction are provided in Appendix~\ref{app:d_dimensional_boustrophedon}.

{
\renewcommand{\arraystretch}{1.3}
\begin{table}[t]
\centering
\begin{tabular}{l @{\hspace{2em}} c c}
\toprule
Method & Gate-Count & Depth \\
\midrule
Ours & $\mathcal{O}(N^{3/2})$ & $\mathcal{O}(\sqrt{N})$\\
Constantinides et al. \cite{constantinidesLowdepthFermionRouting2025} & $\mathcal{O}(N^{3/2}\log^2 N)$ & $\mathcal{O}(\sqrt{N}\log^2 N)$\\
Givens rotations \cite{kivlichanQuantumSimulationElectronic2018} & $\mathcal{O}(N^2)$ & $\mathcal{O}(N)$\\
\bottomrule
\end{tabular}
\caption{
Asymptotic resource scaling for the two-dimensional \textsf{FFFT} on an ancilla-free two-dimensional qubit lattice. We report the leading-order gate counts and circuit depths required for the implementation of the transform across different algorithmic methods.
}
\label{tab:FFFT}
\end{table}
}

\subsection{Molecular simulation}
\label{sec:molecularsim}

As a final application, we focus on molecular electronic structure problems.
These represent a primary use case for quantum computing, as ground-state electronic properties underpin much of modern chemistry and materials science. 
In a second-quantized basis, the electronic Hamiltonian takes the form
\begin{equation}
H = \sum_{i,j} h_{ij}\, c_i^\dagger c_j
+ \frac12 \sum_{i,j,k,l} h_{ijkl}\, c_i^\dagger c_j^\dagger c_k c_l,
\end{equation} 
where the two-electron tensor $h_{ijkl}$ contains $\mathcal{O}(N^4)$ nonzero terms.
A central task in this setting is the estimation of the ground-state energy.
However, near-term approaches based on traditional phase estimation are hindered by the high gate counts and depths required to handle the large number of Hamiltonian terms. 
In contrast, sample-based methods leverage classical post-processing to reduce the demands on quantum hardware, making them a promising candidate for practical chemical simulations on near-term devices.

A prominent sample-based algorithm is the variational quantum eigensolver (VQE). 
In particular, ADAPT-VQE~\cite{grimsleyAdaptiveVariationalAlgorithm2019} constructs its ansatz iteratively by selecting single-term evolutions from an operator pool through a classical optimization procedure. 
In this context, the unitary coupled-cluster ansatz truncated to single and double excitations (UCCSD) provides a particularly natural operator pool:
\[
\mathcal{P}
=
\left\{
c_i^\dagger c_j - c_j^\dagger c_i
\right\}
\cup
\left\{
c_i^\dagger c_j^\dagger c_k c_l
-
c_l^\dagger c_k^\dagger c_j c_i
\right\}.
\]
Operationally, instead of compiling the full $\mathcal{O}(N^4)$-term Hamiltonian into a single deep circuit, the algorithm repeatedly executes circuits that contain only a comparatively sparse subset of non-local interactions selected by the optimizer.

Another recently proposed method is the SqDRIFT protocol~\cite{piccinelliQuantumChemistryProvable2025}. 
In this approach, a long product-formula circuit is replaced by a sequence of stochastically chosen single-term evolutions whose expectation value reproduces the target time evolution. 
Consequently, each circuit implements only a sparse set of non-local interactions selected at random, while the intended dynamics are recovered through classical post-processing.

When the term selection is sparse, the primary compilation burden shifts from implementing the entire Hamiltonian simultaneously to efficiently realizing various spatially non-local fermionic interactions.
Effectively, this transforms the dominant bottleneck into a fermion routing problem.
Our framework provides a solution to this bottleneck by allowing routing to be executed with depth governed by the intrinsic connectivity of a $d$-dimensional qubit lattice. 
By eliminating the overhead typically induced by traditional F2Q mappings, our approach ensures that the physical hardware constraints, rather than the encoding scheme, limits the performance. 

\section{Discussion and outlook}

Fermion-to-qubit encodings constitute an intensely studied area and have seen substantial progress along several directions.
A particularly important line of work addresses the intrinsic one-dimensional connectivity restriction of the Jordan-Wigner transformation by embedding the fermion lattice into a constrained subspace of a larger qubit lattice \cite{bravyiFermionicQuantumComputation2002, verstraeteMappingLocalHamiltonians2005a, derbyCompactFermionQubit2021}.
These locality-preserving encodings retain the geometry of the original problem in the sense that local operators on the fermion lattice are mapped to local operators on the qubit lattice.
This local-to-local property is their main strength, enabling direct simulation of fermion lattices with asymptotically optimal gate-count and depth scaling.

However, this is not without costs.
These encodings require $\mathcal{O}(N)$ additional qubits, which can impose a prohibitive overhead, especially in fault-tolerant architectures where qubit count is often more costly than Clifford-gate depth due to exponential error suppression~\cite{constantinidesLowdepthFermionRouting2025}.
Furthermore, they require the preparation of an entangled initial state, which, on a qubit lattice, must be achieved either through a deep unitary circuit or via feed-forward measurements.

A complementary recent development is the use of dynamical encodings, which enable arbitrary fermionic simulations with only polylogarithmic depth overhead.
On reconfigurable qubit arrays, this can be achieved with an $\mathcal{O}(\log N)$ depth overhead \cite{maskaraFastSimulationFermions2025}.
However, this construction relies on $\mathcal{O}(N)$ ancillas and feed-forward measurements, which, for the simulation of fermion lattices, leads to trade-offs similar to those encountered in locality-preserving encodings.
A second dynamical-encoding approach applies to fully connected qubits and allows arbitrary fermionic simulations to be performed without ancillas, at the cost of an $\mathcal{O}(\log^2 N)$ depth overhead \cite{constantinidesLowdepthFermionRouting2025}.
When transferred to geometrically local qubit architectures, this fully connected construction implies an upper bound of $\mathcal{O}(N^{\frac{1}{d}}\log^2 N)$ depth overhead on a $d$-dimensional qubit lattice.
For simulating two-dimensional fermion lattices, however, this depth remains larger than that achieved by FSNs, indicating that the advantages of the fully connected setting in Ref.~\cite{constantinidesLowdepthFermionRouting2025} do not directly translate into optimal performance on local architectures.

We have presented methods that combine the key advantages of locality-preserving and dynamical encodings, mapping local operators on the fermion lattice to local operations on the qubit lattice while incurring primarily depth overhead.
Together with their small prefactors, this local structure suggests that our methods are not merely of theoretical interest, but may be directly relevant for near- and long-term quantum-computing architectures.
This prospect is especially compelling in light of two hardware platforms currently at the forefront of quantum computing: superconducting qubits and neutral atoms.

Superconducting qubits are embedded in metallic layers, which prevents their physical movement and typically yield a lattice-like connectivity that is naturally well suited to our two-dimensional fermionic encoding primitives. 
At the same time, this platform benefits from extremely fast, high-fidelity operations and has already demonstrated Hubbard model simulations where exact solutions are classically intractable \cite{alamProgrammableDigitalQuantum2025a}. 

By contrast, neutral-atom platforms trap atoms in optical tweezers and thereby allow for physical movement of the qubits. In principle, this mobility enables all-to-all interactions. However, realistic devices employ specific protocols for qubit movement, allowing some types of non-local interactions to be implemented more efficiently than others. One common class of native qubit movements, for example in architectures based on acousto-optical deflectors \cite{xuConstantoverheadFaulttolerantQuantum2024}, consist of row and column permutations. These enable direct use of the intrinsic product structure of our methods, making it especially promising to trade a small number of native qubit permutations for lower depth.

\section*{Acknowledgements}
The authors thank Christophe Goeller, Barry Mant, Anette Messinger, Roberto Munoz, Michael Schuler and Riccardo Valencia for insightful discussions. 
Also, they thank Frank Verstraete for his suggestion to consider lattices with periodic boundaries. 
ChatGPT 5.5 was used to improve the language and mathematical notation of our manuscript.
This research was funded in whole, or in part, by the Austrian Science Fund (FWF) SFB BeyondC Project No. F7108-N38 (DOI: 10.55776/F71). 
This project was supported by a FFG Funding (Project No. FO99918691) as part of the international Eureka cooperation, and by the Austrian Research Promotion Agency (FFG Project No. FO999937388, FFG Basisprogramm). For the
purpose of open access, the author has applied a CC BY public copyright license to any Author Accepted manuscript version arising from this submission.


%

\appendix

\section{Dynamic fermion-to-qubit mappings}
\label{app:clifford_gates}

In this section, we provide a detailed description of fermionic computations and dynamic encodings, as these form the basis for all subsequent results.
We also give a short summary of the parity flow formalism \cite{klaverParityFlowFormalism2025}, which we use to find efficient dynamic encoding switches.

\subsection*{Preliminaries}
We use some of the notation and results introduced in Ref.~\cite{maskaraFastSimulationFermions2025}, which we review in this subsection for completeness. 
We also introduce several adaptations to better align this formalism with our framework.
We begin by defining fermionic computations and then discuss how they can be implemented on a qubit-based quantum computer.

A fermionic computation consists of $N$ modes together with a set of basic interactions. 
Common examples include hopping interactions $e^{-i\theta(c_i^\dagger c_j + c_j^\dagger c_i)}$
and density-density interactions $e^{-i\theta n_i n_j}$.
These interactions are expressed in terms of fermionic ladder operators $c_i$ and $c_j^\dagger$, alongside number operators $n_i=c_i^\dagger c_i$.
The fermionic ladder operators satisfy the canonical anticommutation relations:
\[
\{c_i,c_j\}=\{c_i^\dagger,c_j^\dagger\}=0,
\qquad
\{c_i^\dagger,c_j\}=\delta_{i,j}.
\] 
The gate count of a fermionic computation is the number of basic interactions required for its implementation. 
Furthermore, requiring that each mode participates in at most one basic interaction at any given time, yields a natural notion of computational depth.

On the qubit side, we primarily consider computations on $N=L^d$ qubits indexed by a $d$-dimensional grid graph $H = (G,E)$ with vertices $G=\{0,\dots,L-1\}^d$. 
This indexing represents a hardware layout that is well suited to our circuit constructions. 
For simplicity, we focus on lattices with equal side lengths across all dimensions; nonetheless, all subsequent results extend naturally to asymmetric lattices.
Our basic gate set consists of arbitrary single-qubit Pauli rotations $R_P(\theta)=e^{-i\theta P}$ for $P\in\{X,Y,Z\}$, together with CNOT gates. 
The application of CNOT gates is restricted to pairs qubits connected by an edges in the graph $G$. 
Our primary setting considers nearest-neighbor edges
\begin{equation*}
    E=\{(x_1,x_2)\in G\times G : \|x_2-x_1\|_1 = 1\},
\end{equation*}
which reflects a realistic model of qubit connectivity in many hardware architectures.
To perform a fermionic computation on a qubit-based quantum computer, we must relate Pauli rotations to fermionic ladder operators. 
We achieve this using the JW encoding, which is uniquely specified by a choice of canonical ordering.

\begin{definition}[canonical ordering]
A JW transformation is specified by a bijection
\[
m : \{0,\ldots,N-1\} \to \{0,\ldots,N-1\},
\]
which maps each qubit index to its canonical rank. 
We refer to this bijection as the canonical ordering. 
Since any ordering of the $N$ qubit indices is valid, there are $N!$ possible choices. 
When the qubits are indexed by a $d$-dimensional grid graph, we write
\[
m : \{0,\ldots,L-1\}^d \to \{0,\ldots,N-1\},
\]
with $N=L^d$.
\end{definition}

The corresponding F2Q map for the JW encoding can be defined in terms of Majorana operators:
\begin{equation*}
\chi_{2j}
  = X_j \prod_{j' \in L(j)} Z_{j'},
\qquad
\chi_{2j+1}
  = Y_j \prod_{j' \in L(j)} Z_{j'},
\end{equation*}
where
\[
 L(j)=\left\{j' : m(j')<m(j),\; 1\leq j'\leq N\right\}.
\]
This gives rise to the fermionic ladder operators 
\begin{equation*}
c_j = \frac{1}{2}\bigl(\chi_{2j} + i\,\chi_{2j+1}\bigr),
\qquad
c_j^\dagger = \frac{1}{2}\bigl(\chi_{2j} - i\,\chi_{2j+1}\bigr).
\end{equation*}
While choosing among different JW orderings offers considerable freedom in encoding elementary fermionic interactions, this flexibility remains fundamentally limited by the one-dimensional nature of the encoding. 
In particular, hopping interactions between modes separated by distance $k$ in the canonical ordering give rise to $k$-body qubit interactions. 
This overhead serves as the main motivation behind fermionic permutations, which allow to dynamically change the canonical order via Clifford gates.

The simplest example of such a fermionic permutation is illustrated in Fig.~\ref{fig:1}(c), representing a CZ gate between qubit $j$ and qubit $k$, where $m(j)=m(k)+1$. 
The CZ gate swaps these qubits in the canonical ordering, yielding $m'(j)=m(j)-1$ and $m'(k)=m(k)+1$. 
This can be seen from the Heisenberg-picture action of the Clifford unitary on either of the two Majorana operators associated with qubit $j$:
\begin{equation*}
\begin{aligned}
    \mathrm{CZ}_{j,k}(\chi_{2j})\mathrm{CZ}_{j,k} & = \mathrm{CZ}_{j,k}\left(X_j \prod_{j' \in L(j)} Z_{j'}\right)\mathrm{CZ}_{j,k}\\
    & = X_jZ_k\prod_{j' \in L(j)} Z_{j'},
\end{aligned}
\end{equation*}
which corresponds to a JW encoding specified by the permuted ordering $m'$:
\begin{equation*}
    \mathrm{CZ}_{j,k}(\chi_{2j})\mathrm{CZ}_{j,k} = \chi'_{2j} = X_j\prod_{j' \in  L'(j)} Z_{j'}.
\end{equation*}
More generally, changing the canonical rank of qubit $j$ from $m(j)=a$ to $m'(j)=b$ can be implemented by applying CZ gates between qubit $j$ and all qubits whose ranks lie between $a$ and $b$ in the original ordering $m$. 
More explicitly, applying $\prod_{j' \in J} \mathrm{CZ}_{j,j'}$ with
\begin{equation}\label{eq:PermLine}
\begin{aligned}
    J =
    \begin{cases}
        \{ m^{-1}(r) : a+1 \leq r \leq b \}, & a < b, \\
        \{ m^{-1}(r) : b\leq r \leq a-1 \}, & b < a .
    \end{cases}
\end{aligned}
\end{equation}
shifts the rank of qubit $j$ to $b$. 
In additon, the rank of the other qubits $j' \in J$ change according to $m'(j')=m(j')-1$ if $a<b$, and $m'(j')=m(j')+1$ if $a>b$.

We now present a procedure for switching between arbitrary canonical orderings, which is similar in spirit to Ref.~\cite[Lemma~3]{maskaraFastSimulationFermions2025}.
\begin{lemma}[Arbitrary Ordering Switch]
\label{lem:FermPerm}
To transform the canonical ordering $m$ into $m'$, one can apply a $\mathrm{CZ}$ gate to every qubit pair $(i,j)$ whose relative order is reversed in $m'$ with respect to $m$. 
This corresponds to all pairs satisfying
\[
m(i)<m(j) \quad \text{and} \quad m'(i)>m'(j).
\]
\end{lemma}
\begin{proof}
Consider any two orderings $m$ and $m'$ and apply the following sequence:
\begin{enumerate}
    \item 
    Change the canonical rank of qubit $j_{0}=m'^{-1}(0)$ from $m(j_{0})$ to $0$ by applying the \textrm{CZ} gates defined in Eq.~\eqref{eq:PermLine}, which reduce to 
    \[
    \prod_{j' \in J} \mathrm{CZ}_{j_0,j'},
    \qquad
    J = \{k: m(k)<m(j_{0})\}.
    \]
    This results in an intermediate ordering where qubit $j_0$ holds its target rank $0$, while the relative order of all other qubits $j\neq j_0$ remain unchanged. 
    Consequently, no further CZ gates involving qubit $j_{0}$ are required in the subsequent steps.
    Note that if $m(j_0) = 0$, then $J$ is empty and the identity is applied.
    \item 
    Change the rank of qubit $j_{1}=m'^{-1}(1)$ from its current value $m(j_{1})$, to $1$. 
    To achieve this, apply CZ gates as defined in Eq.~\eqref{eq:PermLine}, excluding qubit $j_0$:
    \[
    \prod_{j' \in J} \mathrm{CZ}_{j_1,j'},
    \quad
    J = \{k: m(k)<m(j_{1})\}\setminus \{j_0\}.
    \]
    Analogous to step $1$, this gives an intermediate ordering where $j_0$ and $j_1$ at their target ranks. 
    Consequently, no further CZ gates need to act on $j_0$ or $j_1$.
    \item Continue this strategy for the remaining $N-2$ qubits. 
    For each qubit $j_i$, apply the gates 
    \begin{gather*}
    \prod_{j' \in J} \mathrm{CZ}_{j_i,j'}\\
    \begin{aligned}
    J &= \{k: m(k)<m(j_{i})\}\setminus \{l : m'(l)<m'(j_{i})\}\\
    &=\{k: m(k)<m(j_{i})\}\cap\{l : m'(l)>m'(j_{i})\}.
    \end{aligned}
    \end{gather*}
\end{enumerate}
Thus, the specified sequence transforms the canonical ordering $m$ into $m'$, with the applied CZ gates satisfying the condition stated in Lemma~\ref{lem:FermPerm}.
\end{proof}


In the following, we frequently study circuits composed of CZ gates. 
To simplify these calculations, we represent them in a canonical form.
For this we consider a circuit
\begin{equation}
\label{eq:CZcircuit}
C=\prod_{(i,j) \in S}\mathrm{CZ}_{i,j},  
\end{equation}
where $S \subseteq \{(i,j): i,j\in \{0,1,\ldots, n\} \}$ is some set consisting of index pairs.
    
\begin{definition}[CZ circuits as polynomials]
\label{def:CZpolynomial}
Given a circuit $C$ consisting of CZ gates, we define the corresponding $\mathbb{F}_2$ circuit-polynomial as
\[
T_C(x_0,...,x_n) = \sum_{(i,j) \in S} x_ix_j.
\]
\end{definition}

We proof that any circuit consisting of CZ gates can be written in terms of such a polynomial.

\begin{lemma}\label{lem:polynom_repr_of_CZcircs}
Let $C$ be a circuit of CZ gates defined via the set $S$. 
Then,
\begin{equation*}
C=
\exp\left(-i\pi T_C\right)
\sum_{x \in \mathbb{F}_2^n}  |x\rangle\langle x|,
\end{equation*}
where $T_C$ is the polynomial defined in Def.~\ref{def:CZpolynomial}.
\end{lemma}

\begin{proof}
Starting from Eq.~\eqref{eq:CZcircuit}, we utilize that $\mathrm{CZ}$ gates 
\begin{equation*}
    \mathrm{CZ}_{i,j}=
\exp\left[-i\frac{\pi}{4}(I-Z_i)(I-Z_j)\right]
\label{eq:czgate_representation}
\end{equation*}
are diagonal in the computational basis and write
\[
C=
\prod_{(i,j) \in S}
\exp\left(-i\pi x_ix_j\right)
\sum_{x \in \{0,1\}^n}  
|x\rangle\langle x|.
\]
Furthermore, treating the bits $x_i$ as elements of $\mathbb{F}_2$, the expression can take the form
\[
C=
\exp\left(-i\pi\sum_{(i,j) \in S} x_ix_j\right)
\sum_{x \in \mathbb{F}_2^n}  |x\rangle\langle x|.
\]
\end{proof}

Obviously, two $\mathrm{CZ}$ circuits are equivalent when their corresponding polynomials are identical (over $\mathbb{F}_2$).

Unlike products of CZ gates, where the application order does not matter because they commute, generic cases include non-commuting terms.
There we make the following definition.
\begin{definition}[Circuit-ordered product]
For non-commuting operator products, we adopt a circuit-inspired convention where factors are ordered by ascending index, with smaller indices acting first:
\[
\revprod{i=0}{n-1}
U_{i,i+1}
\coloneqq
U_{n-1,n}
\cdots
U_{1,2}
U_{0,1}.
\]
\end{definition}

\subsection*{Fermionic permutations in the parity flow formalism}
From the method described in the previous section, we find that changing the rank of a qubit from $m(i)$ to $m'(i)$ requires $|m'(i) - m(i)|$ non-local CZ gates. 
A naive approach for an arbitrary permutation is to apply all required CZ gates directly via Lemma~\ref{lem:FermPerm}. 
However, as illustrated by the Z-to-S pattern transformation, executing multiple rank updates in parallel via structured CNOT ladders and CZ gates yields far greater efficiency than this naive approach. 
Nevertheless, finding an optimal combination of local CNOT and CZ gates for more complex permutations remains non-trivial.

To address this, the parity flow formalism \cite{klaverParityFlowFormalism2025} provides a convenient framework for describing the joint action of CNOT and CZ gates. 
While this formalism is quite general, for our purposes it suffices to track information propagation through circuits consisting solely of CNOT gates.

In this special case, the information is encoded in $Z$-labels associated with each qubit, which update efficiently upon applying a CNOT gate.
To see how this works, consider the Heisenberg-picture conjugation identity for a CNOT gate acting on $Z$ operators:
\[
\begin{aligned}
    \mathrm{CNOT}_{i,j}^\dagger Z_i \, \mathrm{CNOT}_{i,j} &= Z_i, \\
    \mathrm{CNOT}_{i,j}^\dagger Z_j \, \mathrm{CNOT}_{i,j} &= Z_i Z_j.
\end{aligned}
\]
These relations describe how the CNOT gate transforms the Pauli $Z$ operators of control qubit $i$ and target qubit $j$. 
Specifically, the second identity shows that measuring $Z_j$ after the gate is equivalent to measuring the joint operator $Z_i Z_j$ beforehand.

In the parity flow formalism, this information content is abstracted into a set of labels $\ell(i) \subseteq \{1, \dots, n\}$ that record which logical operators are stored in physical qubit $i$.
For instance, starting with initial labels $\ell(i) = \{i\}$ and $\ell(j) = \{j\}$, applying $\mathrm{CNOT}_{i,j}$ updates the label to $\tilde\ell(i) = \{i\}$ and $\tilde\ell(j) = \{i, j\}$. 
More generally, if $\ell(i)$ and $\ell(j)$ denote the labels at time $t$, applying $\mathrm{CNOT}_{i,j}$ updates them at time $t+1$ according to
\begin{equation}
    \tilde\ell(i) = \ell(i), \quad \tilde\ell(j) = \ell(i) \Delta \ell(j),
    \label{eq:update_labels}
\end{equation}
where $\Delta$ denotes the symmetric difference between sets.

This formalism allows to efficiently track the information content, leading to following key result. 

\begin{lemma}\label{lem:FlowCZ}
Let $U_{\mathrm{CNOT}}$ denote a Clifford circuit consisting solely of $\mathrm{CNOT}$ gates. 
Furthermore, let $\ell(i)$ and $\ell(j)$ be the labels obtained by tracking the action of $U_{\mathrm{CNOT}}$ within the parity flow formalism, starting from the trivial initial labeling $\ell(k) = \{k\}$ for all $k$.
Then, conjugating a single $\mathrm{CZ}_{i,j}$ gate by $U_{\mathrm{CNOT}}$ is equivalent to applying $\mathrm{CZ}_{i',j'}$ gates across all qubit pairs $(i', j')$ with $i' \in \ell(i)$ and $j' \in \ell(j)$. 
Formally,
\[
U_{\mathrm{CNOT}}^{\dagger} \, \mathrm{CZ}_{i,j} \, U_{\mathrm{CNOT}} = \prod_{i' \in \ell(i)} \prod_{j' \in \ell(j)} \mathrm{CZ}_{i',j'}.
\]
\end{lemma}

\begin{proof}
A $\mathrm{CZ}_{i,j}$ gate conjugated by CNOT gates corresponds to
\[
\begin{aligned}
    &U_{\text{CNOT}}^{\dagger} \left[\exp\left(-i\frac{\pi}{4}\left(I-Z_{i}\right)\left(I-Z_{j}\right)\right) \right] U_{\text{CNOT}}\\
    &=\exp\left[-i\frac{\pi}{4}\left(I-\prod_{i' \in \ell(i)}Z_{i'}\right)\left(I-\prod_{j' \in \ell(j)}Z_{j'}\right)\right].
\end{aligned}
\]
With Definition~\ref{def:CZpolynomial}, this diagonal gate is equivalent to the $\mathbb{F}_2$ polynomial
\[
T=
\left(\sum_{i' \in \ell(i)} x_{i'}\right)
\left(\sum_{j' \in \ell(j)} x_{j'}\right).
\]
By distributivity over $\mathbb{F}_2$, we have
\[
\left(\sum_{i' \in \ell(i)} x_{i'}\right)
\left(\sum_{j' \in \ell(j)} x_{j'}\right)
=
\sum_{i' \in \ell(i),\, j' \in \ell(j)}
x_{i'}x_{j'}.
\]
On the other hand, the polynomial associated with the circuit 
\[
C=\prod_{i' \in \ell(i)} \prod_{j' \in \ell(j)} \mathrm{CZ}_{i',j'}
\]
is identical to this
\[
T_C 
= 
\sum_{i' \in \ell(i),\, j' \in \ell(j)}
x_{i'}x_{j'}
=
T.
\]
\end{proof}

Because Lemma~\ref{lem:FlowCZ} enables the decomposition of complex, non-local CZ gates into geometrically local CNOT and CZ gates, it serves as our primary tool for efficiently implementing switches between canonical orderings.

\section{2D code switching}
\label{app:2d_circuits}

In this section we prove the correctness of the circuit presented in Fig.~\ref{fig:2} for switching between $Z$- and $S$-pattern encodings. 
Subsequently, we show that applying this transformation to a subset of the qubit lattice facilitates the switching between boustrophedon encodings, as illustrated in Fig.~\ref{fig:3}.

As noted above, the proof formally builds on Lemma~\ref{lem:FlowCZ} by using a structured set of CNOT gates tailored to the 2D layout, enabling the exchange of information between columns and rows.

\subsection*{Switching between S and Z ordering}

Consider a 2D lattice of size $L\times L$, with rows and columns indexed by $r$ and $c$ (counting from $0$).
The canonical ordering corresponding to the $Z$ pattern depicted in Fig.~\ref{fig:2} (left) is described as
\begin{equation}
   \label{eq:2D_Z_pattern}
    m_{Z}(r,c)=Lc+(L-r).
\end{equation}
For the $S$-pattern [see Fig.~\ref{fig:2} (right)] we differentiate between even and odd rows
\begin{equation}\label{eq:2D_S_pattern}
m_{S}(r, c) = 
\begin{cases}
    Lr + c, & \text{if } r \text{ is even}, \\
    Lr + (L - c), & \text{if } r \text{ is odd}.
\end{cases}
\end{equation}
The key difference between these two patterns is that the $Z$-pattern realizes a column-first ordering, whereas the $S$-pattern realizes a row-first ordering. 
To switch between orderings, we employ CNOT ladders along columns and rows as depicted in Fig.~\ref{fig:2} (middle). 
These CNOT ladders accumulate parity information first along columns and then along rows. 
Utilizing only every second row enables us to locally distinguish column-based information from row-based information. 
By applying Lemma~\ref{lem:FlowCZ}, we simplify the required CZ interactions between rows and columns to a single CZ gate on neighboring rows.

To formally prove the $Z$-to-$S$ transformation, we first establish an intermediate result in Lemma~\ref{lem:colRowBasis} that, when applied in the correct sequence, allows one to swap a single column with a row.
From this, we prove in Theorem~\ref{thm:C2D} the correctness of the circuit that simultaneously swaps all rows and columns. 
As noted in the main text, this circuit is built from CNOT ladders and CZ gates. 
Formally, these CNOT ladders are defined as parallel CNOT ladders along each column,
\begin{equation*}
    C_{\uparrow}=\prod_{c}\revprod{r=0}{L-2} \mathrm{CNOT}_{(r+1, c), (r, c)},
\end{equation*}
and parallel CNOT ladders along each even row,
\begin{equation*}
    C_{\leftarrow}=\prod_{r_{\mathrm{even}}}\revprod{c=0}{L-2} \mathrm{CNOT}_{(r_{\mathrm{even}}, c+1), (r_{\mathrm{even}}, c)}.
\end{equation*}
The key insight is that applying column-wise followed by row-wise CNOT ladders accumulates labels in a way that makes executing a single row-to-column switch exceptionally simple using very few CZ gates. 
Specifically, via Lemma~\ref{lem:FlowCZ}, these CNOT circuits map all CZ gates required for the transformation down to just two CZ gates.

This motivates the following definition.
\begin{definition}[Intersection basis]
The basis transformation (from computational basis) to the intersection basis is defined by 
\[
U_{\rm ibasis} :=C_{\leftarrow} C_{\uparrow}.
\]
\end{definition}

To transform from the $Z$-pattern to the $S$-pattern (or vice versa), row information must be exchanged with column information. 
Let $(r,c)$ denote a qubit in the lattice serving as a pivot point for this information exchange. 
For each odd $r$ (counting from $0$), we define two sets, as illustrated in Fig.~\ref{fig:crossing_point}: the column set
\begin{equation}
\label{eq:column_set}
C_{r,c} = \{(r',c) : r' \geq r\},
\end{equation}
and the row set
\begin{equation}
\label{eq:row_set}
R_{r,c} = \{(r-1,c') : c' \geq c\} \cup \{(r,c') : c' \geq c\}.
\end{equation}

Note that unlike the column set, which comprises a single column, the row set spans two rows. 
This structural difference arises because the JW path aligns along rows in an $S$-pattern, whereas it aligns along columns in a $Z$-pattern. 
For instance, to transition from an $S$-pattern to a $Z$-pattern, the qubit directly above the pivot point $(r, c)$, located at $(r-1, c)$ as shown in Fig.~\ref{fig:crossing_point}, must be brought adjacent to the pivot in the JW-ordering.
This is achieved using CZ gates from the pivot to all intermediate qubits along the $S$-path leading to $(r-1, c)$, a path that unavoidably spans two rows.

As shown in the following lemma, all pairwise CZ gates between column set $C_{r,c}$ and the row set $R_{r,c}$ can be readily implemented using our CNOT chains $U_{\mathrm{ibasis}}$ to work in the intersection basis. 
To avoid edge cases, we temporarily assume that $L$ is odd.

\begin{lemma}
\label{lem:colRowBasis}
Let $C_{r,c}$ and $R_{r,c}$ be sets of qubits arranged along columns and rows as defined in Fig.~\ref{fig:crossing_point}.
The simultaneous implementation of all CZ gates between these sets 
\begin{equation}
    U_{r,c}=\prod_{i' \in C_{r,c}} \prod_{j' \in R_{r,c}}
\mathrm{CZ}_{i',j'}
\end{equation}
can be implemented, in the intersection basis, using two local $\mathrm{CZ}$ operations, i.e.,  
\[
U_{\rm ibasis}U_{r,c}U_{\rm ibasis}^{\dagger}=\mathrm{CZ}_{(r,c), (r+1,c)}\mathrm{CZ}_{(r,c), (r-1,c)}.
\]
\end{lemma}
\begin{proof}
We first compute the parity labels $\ell(r,c)$ induced by the CNOT gates.
For this let $\ell_0(r,c) = (r,c)$ denote the initial labels.
After application of $C_{\uparrow}$ these labels can be updated via Eq.~\ref{eq:update_labels} to
\begin{equation*}
    \ell_1(r,c) = \{(r',c):r' \geq r\}.
\end{equation*}
Finally, after applying $C_{\leftarrow}$ they result into 
\begin{equation*}
\label{eq:labels2D}
\ell(r,c)= \left\{(r',c'):
\left[
\begin{array}{ll}
r' \geq r, c' \geq c & \quad r \text{ even}, \\[1em]
r' \geq r, c' = c & \quad r \text{ odd},
\end{array}
\right.
\right\}.
\end{equation*}

Let $r$ be odd. We rewrite $U_{r,c}$ as 
\begin{gather*}
\left(
U_{\rm ibasis}^{\dagger} 
\mathrm{CZ}_{(r,c), (r+1,c)}
U_{\rm ibasis}
\right)
\left(
U_{\rm ibasis}^{\dagger} 
\mathrm{CZ}_{(r,c), (r-1,c)}  
U_{\rm ibasis} 
\right)
\end{gather*}
and make use of Lemma~\ref{lem:FlowCZ}, resulting into
\begin{equation*}
\left(\prod_{j \in \ell(r,c)} \prod_{k \in \ell(r+1,c)}\mathrm{CZ}_{j,k}\right)\left(\prod_{j \in \ell(r,c)} \prod_{k \in \ell(r-1,c)}\mathrm{CZ}_{j,k}\right).
\end{equation*}
This expression can be simplified further. 
To see this, consider a label $m \in \ell(r-1,c) \cap \ell(r+1,c)$. 
Since $m$ belongs to both label sets, the previous representation contains two identical CZ gates between the pivot and $m$. 
These contributions cancel out, resulting in
\begin{equation}\label{eq:CZpair}
U_{r,c}=\prod_{j \in \ell(r,c)} \prod_{k \in \ell(r+1,c)\triangle\ell(r-1,c)}\mathrm{CZ}_{j,k}.
\end{equation}
From this, we can identify $C_{r,c} = \ell(r,c)$ and
\begin{equation}\label{eq:Delta2D}
\begin{gathered}
\ell(r+1,c)\triangle\ell(r-1,c)=\\[1em]
\{(r-1,c'): c' \geq c\} 
\cup
\{(r,c'): c' \geq c\}
=R_{r,c}.
\end{gathered}
\end{equation}
This concludes the proof.
\end{proof}
With these results in place, we can finally prove the main result of this section. 
To this end, we define the circuit $C_{\mathrm{2D}}$, as illustrated in Fig.~\ref{fig:2}, as
\begin{equation*}
    C_{\mathrm{2D}} := U_{\mathrm{ibasis}}^{\dagger} C_Z U_{\mathrm{ibasis}} = C_{\uparrow}^{\dagger} C_{\leftarrow}^{\dagger} C_Z C_{\leftarrow} C_{\uparrow},
\end{equation*}
where $C_Z$ represents a product of CZ gates along each column:
\begin{equation*}
    C_Z := \prod_{c} \prod_{r} \mathrm{CZ}_{(r,c), (r+1,c)}.
\end{equation*}
While Lemma~\ref{lem:colRowBasis} considers only a single column and two rows, yielding just two CZ gates, $C_Z$ generalizes this structure by applying column-wise CZ gates across all rows and columns simultaneously.

\begin{theorem}($Z\leftrightarrow S$ transformation)
\label{thm:C2D}
The circuit $C_{\mathrm{2D}}$ implements the transformation between the orderings $m_Z$ and $m_S$ in either direction, up to single-qubit $Z$ gates on odd rows.
\end{theorem}

\begin{proof}
According to Lemma~\ref{lem:FermPerm}, switching between the orderings $m_{Z}$ and $m_{S}$ requires applying a CZ gate between every qubit pair $(r, c)$ and $(r', c')$ satisfying
\begin{equation}
\label{eq:2D_condition}
\begin{aligned}
    (m_{Z}(r, c) > m_{Z}(r', c')) &\text{ and } (m_{S}(r, c) < m_{S}(r', c')).
\end{aligned}
\end{equation}
For a given $(r,c)$ the corresponding set which satisfies this is illustrated in Fig.~\ref{fig:2D_conditions}(a).
This condition is formalized by defining the target sets $\Lambda_{r,c}$, as illustrated in Fig.~\ref{fig:2D_conditions}(b).
Geometrically, $\Lambda_{r,c}$ contains qubits in the upper-right sector relative to $(r,c)$.
Due to the snake-like pattern of the S ordering, the exact domain depends on row parity: for even $r$, $\Lambda_{r,c}$ excludes qubits in the same row, whereas for odd $r$, same-row qubits are included (excluding self-interaction).
Formally, if CZ gates are applied between each qubit $(r,c)$ and all qubits in its target set $\Lambda_{r,c}$ as
\begin{equation}
\label{eq:CZlambda}
\begin{gathered}
C_{\Lambda}=
\prod_{(r,c)}
\prod_{(r',c') \in \Lambda_{r,c}}
\mathrm{CZ}_{(r,c),(r',c')},\\\\[0.5em]
\Lambda_{r,c}
=
\left\{
(r',c') :
\begin{array}{l}
(r',c')\neq (r,c),\\[0.3em]
c' \geq c,\\[0.3em]
r' < r \ (r\ \mathrm{even}),\\[0.3em]
r' \leq r \ (r\ \mathrm{odd})
\end{array}
\right\},
\end{gathered}
\end{equation}
the required change of canonical ordering is successfully implemented.
Following Definition~\ref{def:CZpolynomial}, this is equivalent to the $\mathbb{F}_2$ polynomial
\begin{equation}\label{eq:logicSwap}
T_{\Lambda}=
\sum_{(r,c)}
\sum_{(r',c') \in \Lambda_{r,c}}
x_{(r,c)}x_{(r',c')}.
\end{equation}

We now compare this polynomial with the one induced by $C_{\mathrm{2D}}$. 
Again assuming odd $L$, we pair the CZ gates appearing in $C_{\mathrm{2D}}$ into sets of two CZ gates between column-wise nearest-neighbor qubits, matching their structure in Lemma~\ref{lem:colRowBasis}
\footnote{~ 
For even $L$, the same pairing construction applies to all gates, except for an additional set involving the boundary qubits in row $r = L - 1$.
For simplicity we only show the proof for odd $L$, however this can be easily adapted to the even case if we additionally include these boundary terms. 
}.
For this we decompose,
\begin{equation*}
C_{\mathrm{2D}}=\prod_{r \text{ odd}}\prod_c U_{r,c},
\end{equation*}
via
\[
U_{r,c}=C_{\uparrow}^{\dagger} C_{\leftarrow}^{\dagger} \left ( \mathrm{CZ}_{(r,c), (r+1,c)}\mathrm{CZ}_{(r,c), (r-1,c)} \right ) C_{\leftarrow} C_{\uparrow}.
\]
According to Lemma~\ref{lem:colRowBasis} we have
\begin{equation*}
    C_{\mathrm{2D}}=\prod_{r \text{ odd}}\prod_c \prod_{i' \in C_{r,c}} \prod_{j' \in R_{r,c}}
\mathrm{CZ}_{i',j'}
\end{equation*}
involving all column sets $C_{r,c}$ and row sets for all possible pivots $(r,c)$ with $r$ odd.
Making use of Definition~\ref{def:CZpolynomial}, this circuit is equivalent to the polynomial
\begin{gather*}
T_{\mathrm{2D}}=
\sum_{\substack{(r,c)\\ r \text{ odd}}}
\left(\sum_{i' \in C_{r,c}} x_{i'}\right)
\left(\sum_{j' \in R_{r,c}} x_{j'}\right)=\\
\sum_{\substack{(r,c)\\ r \text{ odd}}}
\left(\sum_{\substack{r' \geq r}} x_{(r',c)}\right)
\left(\sum_{\substack{c' \geq c}} x_{(r-1,c')} + x_{(r,c')}\right).
\end{gather*}
Due to the distributively of summation one can rearrange all these terms to bring them into the same form as  $T_{\Lambda}$, where an outer sum runs over all $(r,c)$.
To see that the row-column set construction covers the entire target set $\Lambda_{r,c}$ for an arbitrary qubit $(r,c)$, observe that $(r,c)$ belongs to the column set $C_{r',c}$ for every pivot $(r',c)$ positioned above it ($r' \le r$).
Each such pivot induces pairwise interactions between $(r,c)$ and the qubits in its corresponding row set $R_{r',c}$. 
Accumulating these interactions over all valid pivots above $(r,c)$ gives
\begin{equation*}
\bigcup_{r' {\rm odd}\le r} R_{r',c} = 
\begin{cases}
    \Lambda_{r,c}, & r\text{ even},\\
    \Lambda_{r,c} \cup (r,c), &r\text{ odd},
\end{cases}
\end{equation*}
which covers the full target set $\Lambda_{r,c}$ of required interactions with $(r,c)$.
This insight allows to rewrite the polynomial as
\[
T_{\mathrm{2D}} =
\sum_{\substack{(r,c)\\ r \text{ odd}}}
\sum_{\substack{r' \leq r\\ c'\geq c}}
x_{(r,c)}x_{(r',c')}
+
\sum_{\substack{(r,c)\\ r \text{ even}}}
\sum_{\substack{r' < r\\ c'\geq c}}
x_{(r,c)}x_{(r',c')}.
\]
This mostly reproduced Eq.~\eqref{eq:logicSwap}, except for additional diagonal terms
\[
T_{\mathrm{2D}}=
T_{\Lambda}+
\sum_{\substack{(r,c)\\ r \text{ odd}}} x_{(r,c)}x_{(r,c)}.
\]
These diagonal terms can be eliminated by applying single-qubit $Z$ gates to all qubits in odd rows.
Consequently, we find that
\begin{equation*}
    C_{\Lambda}=C_\mathrm{2D}\prod_{\substack{(r,c)\\ r \text{ odd}}} Z_{(r,c)}.
\end{equation*}
\end{proof}

As a result of Thm.~\ref{thm:C2D}, the $Z\leftrightarrow S$ transformation requires $\mathcal{O}(N)$ Clifford gates and a depth which is dominated by the constituent CNOT ladders.
Hence, for nearest-neighbor connectivity, the depth scales as $\mathcal{O}(L)$, whereas for fully connected qubits it scales as $\mathcal{O}(\log L)$.

\begin{figure}
    \centering
    \includegraphics[width=0.65\columnwidth]{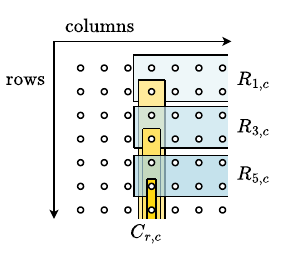}
    \caption{The column set $C_{r,c}$ in yellow and the row set $R_{r,c}$ in blue, as defined in Eq.~\eqref{eq:column_set} and Eq.~\eqref{eq:row_set} respectively.}
    \label{fig:crossing_point}
\end{figure}

\begin{figure}
    \centering
    \includegraphics[width=0.5\columnwidth]{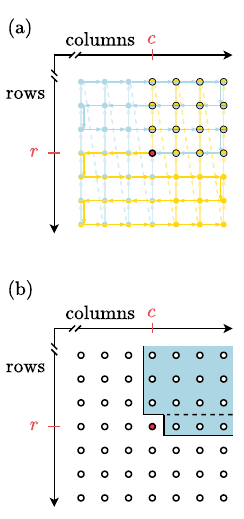}
    \caption{
    (a) The orderings $m_{S}$ and $m_{Z}$, drawn in blue for all rows and columns $(r',c')$ where $m(r',c')<m(r,c)$ for either ordering and in yellow where $m(r',c') > m(r,c)$.
    The region that satisfies the condition in Eq.~\eqref{eq:2D_condition} is given by the qubits on the top right, where the blue part of $m_{S}$ intersects with the yellow part of $m_{Z}$.
    (b) When switching between an $S$-pattern and a $Z$ pattern, CZ gates may be applied between any qubit $(r,c)$ and all qubits in the blue-marked region, as expressed in Eq.~\eqref{eq:CZlambda}. The dotted line indicates the distinction between even and odd rows.
    }
    \label{fig:2D_conditions}
\end{figure}

\subsection*{Switching between boustrophedon encodings}

Building on the result established in the previous section for switching between $S$- and $Z$-patterns, we can demonstrate that a generalized version of $C_{\mathrm{2D}}$ (shown in Fig.~\ref{fig:3}) allows us to switch between any two boustrophedon encodings.

We formally define such an encoding in the following.
\begin{definition}[Boustrophedon encoding]
A boustrophedon encoding $m_{\mathcal{C}}$ is specified by an interval partition
\[
    \mathcal{C} = \{C_1,\dots,C_K\}
\]
of the column index set $\{0, \ldots, L-1\}$. 
That is, the sets $C_k$ are pairwise disjoint, satisfy
\[
    \bigcup_{k=1}^K C_k =\{0, \ldots, L-1\},
\]
and each $C_k$ is an interval of consecutive column indices
\[
    C_k = \{a_k, a_k+1, \dots, a_k+w_k-1\}.
\]
We refer to $w_k = |C_k|$ as the width of partition $C_{k}$.
Let $j(c)$ denote the partition index of column $c$, i.e., $c \in C_{j(c)}$.
The corresponding canonical ordering is formally defined by
\begin{equation}
\label{eq:meander_order}
m_{\mathcal{C}}(r,c) = m^{(j(c))}_S\left(r, c-\sum_{k=1}^{j(c)-1}|C_{k}|\right) +\sum_{i=1}^{j(c)-1}L\,w_{i}, 
\end{equation}
where $m^{(j(c))}_S$ denotes a $S$ pattern of width $w_{j(c)}$.
Conceptually,  $m_{\mathcal{C}}$ is composed of a sequence of $S$-patterns $m^{(j)}_S$ in each subgrid $\{0, \ldots, L-1\} \times C_j$.
\end{definition}

\begin{corollary}
\label{cor:bous2D}
Switching between any two boustrophedon encodings $m_{\mathcal{C}}(r,c)$ and $m_{\mathcal{C}'}(r,c)$, can be achieved using $\mathcal{O}(N)$ nearest-neighbour Clifford gates, and the same asymptotic depth scaling as the $C_{2D}$ transformation.
\end{corollary}
\begin{proof}
We can switch between any two boustrophedon encodings via the following sequence:
\begin{enumerate}
    \item \textbf{Transform $m_{\mathcal{C}}$ to $m_Z$:} Apply $C_{\mathrm{2D}}$ to each subgrid $\{0, \ldots, L-1\} \times C_j$ for all $C_j \in \mathcal{C}$. 
    Note that for any two qubits $(r_{i}, c_{i})$ and $(r_{j}, c_{j})$ in distinct subgrids with $i<j$ the relative order remains invariant:
    \begin{gather*}
    m_{\mathcal{C}}(r_{i}, c_{i}) < m_{\mathcal{C}}(r_{j}, c_{j})\\
    \text{and}\\
    m_{Z}(r_{i}, c_{i}) < m_{Z}(r_{j}, c_{j}).
    \end{gather*}
    Consequently, converting the local $S$-pattern to a $Z$-pattern within each subgrid yields a global $Z$-pattern $m_Z$ across the entire grid.

    \item \textbf{Transform $m_Z$ to $m_{\mathcal{C}'}$:} Apply the inverse transformation $C_{\mathrm{2D}}^{\dagger}$ to each subgrid $\{0, \ldots, L-1\} \times C'_j$ for $C'_j \in \mathcal{C}'$.
\end{enumerate}
This sequence requires $\mathcal{O}(N)$ total gates, with circuit depth dominated asymptotically by the depth of $C_\mathrm{2D}$.
\end{proof}

\section{Constant depth encoding switching via lattice surgery}
\label{app:latticeSurgery}

\begin{figure*}
    \centering
    \includegraphics[width=0.9\textwidth]{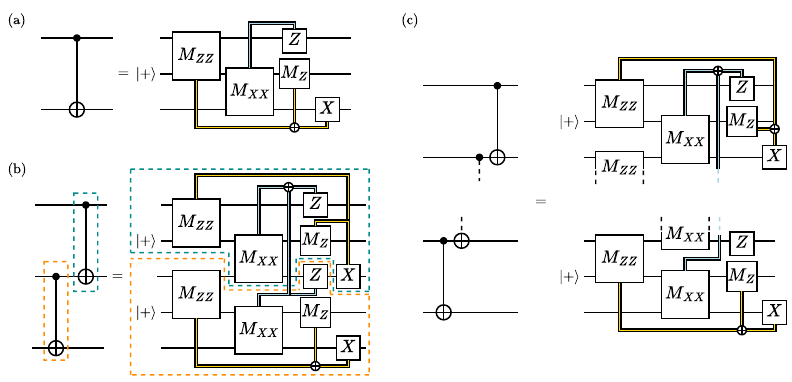}
    \caption{
    CNOT gates via lattice surgery \cite{campbellRoadsTowards2017}. (a) Implementation of a single CNOT using two-qubit Pauli measurements. (b) Two sequential CNOT gates implemented using two-qubit Pauli measurements, where the $Z$ correction of the first CNOT is commuted through the $XX$ measurement of the second CNOT. (c) An Hadamard-transformed CNOT ladder implemented using two-qubit Pauli measurements in $\mathcal{O}(1)$ depth.}
    \label{fig:LatticeSurgery}
\end{figure*}

This section demonstrates how the proposed encoding switches can be implemented in constant gate depth using lattice surgery on surface codes.
Lattice surgery works by measuring local stabilizers between boundaries of different surfaces codes, to either merge or split the boundaries of two different surface codes. 

Specifically, measuring $X$ stabilizers between two rough boundaries merges these boundaries and induces a joint Pauli $XX$ measurement of the two logical qubits, while measuring the $Z$ stabilizers along two smooth boundaries induces a joint $ZZ$ measurement.
COnversely, a surface code can be split by measuring a line of data qubits in the $X$ basis (creating a smooth boundary) or in the $Z$ basis (yielding a rough boundary).

We can use the joint two-qubit Pauli measurements resulting from surface code lattice surgery to implement a CNOT gate between the two logical qubits encoded in the surface code patches \cite{campbellRoadsTowards2017} [see Fig.~\ref{fig:LatticeSurgery}(a)]:
\begin{enumerate}
    \item Introduce an auxiliary patch (qubit) initialized in the $\ket{+}$ state placed inbetween the logical control and target of the CNOT. 
    \item Merge and split the smooth boundary of the control qubit and the auxiliary qubit (inducing a joint $ZZ$ measurement).
    \item Merge and split the rough boundary of the target qubit and the auxiliary qubit (inducing a joint $XX$ measurement).
    \item Measure the auxiliary patch (qubit) in the $Z$ basis.
    \item Apply a $Z$ gate to the control qubit if the $XX$ measurement outcome is $-1$.
    \item Apply a $X$ gate to the target qubit if the product of the $ZZ$ and $Z$ measurement is $-1$.
\end{enumerate}

This procedure allows us to prove the following result.

\begin{corollary}[Constant-depth switching]
Using lattice surgery in surface code lattice surgery, the circuit $C_{2D}$ (and transformations between arbitrary boustrophedon encodings) can be implemented with $\mathcal{O}(N)$ gates and $\mathcal{O}(1)$ depth.
\end{corollary}

\begin{proof}
    The depth of the $C_{\mathrm{2D}}$ circuit
    is determined by the depth of its CNOT ladders.
    Here, using the prescribed CNOT implementation via lattice surgery, we show how reduce this bottleneck by implementing a CNOT ladder with constant depth.
    For this, consider a reverse CNOT ladder, as illustrated on the left side of Fig.~\ref{fig:LatticeSurgery}(c).
    Substituting each CNOT by the lattice surgery procedure [depicted as quantum circuit in Fig.~\ref{fig:LatticeSurgery}(a)],
    we obseve that for each CNOT in the sequence, the corresponding $XX$-measurement outcome is flipped by the $Z$-gate correction applied to the control qubit of the preceding CNOT.
    Equivalently, one can perform the $XX$ measurement prior to applying the corrective $Z$ gate and conditionally flip the measurement result whenever the $Z$ correction is triggered, as illustrated for two CNOTs in Fig.~\ref{fig:LatticeSurgery}(b) and generalized Fig.~\ref{fig:LatticeSurgery}(c).
    Finally, this reversed ladder can be transformed into a standard CNOT ladder by applying a single round of Hadamard gates to all data qubit at both the start and the end of the circuit.
    Consequently, the resulting quantum circuit requires only $3$ rounds of logical measurements and $4$ rounds of single-qubit Clifford gates, yielding a constant circuit depth for implementing a CNOT ladder.
\end{proof}

\section{Locality conservation in subgrid coverage}
\label{app:higher_dimensional_conservation}

So far, we have introduced different JW encodings and demonstrated how to efficiently switch between them. 
Under these encodings, qubits that are adjacent on the grid are also neighbors in either the $S$- or $Z$-pattern, rendering local fermionic interactions nearest-neighbor. 
In this and the following section, we generalize this concept beyond nearest-neighbor interactions on the grid by using multiple boustrophedon encodings. 
This set of boustrophedon encodings is constructed such that for any pair of qubits within a fixed distance on the grid, there exists a boustrophedon encoding in which they remain close, independently of the grid size.

To define these boustrophedon encodings, we partition the grid into subgrids, which, in the simplest case, form vertical stripes. 
Because the JW path snakes through these vertical stripes, qubits within the same stripe remain close in the canonical ordering. 
To also capture pairs that are geometrically close but lie across adjacent stripes, we introduce a second partition shifted by half the stripe width along the columns. 
Furthermore, for reasons that will become clear later, whereas standard boustrophedon encodings partition only along columns, we additionally partition along rows. 
Consequently, qubits co-located within the same vertical and horizontal subgrid are guaranteed to remain close in the canonical ordering, independently of system size. 

Let $V = \{0,\dots,L-1\}^d$ denote a $d$-dimensional qubit grid, and let $\delta$ specify an interaction cutoff radius on $V$ such that geometrically local sites satisfy
\begin{equation*}
    \|\mathbf x - \mathbf y\|_\infty \le \delta.
\end{equation*}
We first define an unshifted partition of the full grid into patches of width $2\delta$ along each dimension. 
For example, in 2D, this yields a grid of $2\delta \times 2\delta$ blocks like a checkerboard (with boundary patches potentially truncated to fit $V$). 
We then construct shifted partitions by optionally translating the subgrids by $\delta$ along each dimension, represented by binary shift indicators $(b_1, \dots, b_d) \in \{0, 1\}^d$, where $b_i = 1$ corresponds to a shift of $\delta$ along dimension $i$.
Formally we make following definition.

\begin{definition}[Subgrid coverage]\label{def:subgrid}
Consider the constant-size one-dimensional intervals 
\[
I_j := \{k : \max\{0,j\} \le k \le \min\{L-1,j+2\delta-1\}\}.
\]
parameterized by start index $j$ and width $2\delta$, such that the upper boundary of the interval is $j+2\delta-1$.
These give rise to $d$-dimensional subgrids via the Cartesian product
\begin{equation*}
    I_{i_1,\ldots, i_d} = I_{i_1} \times I_{i_2} \times \ldots I_{i_d}.
\end{equation*}
We use them to define $2^d$ shifted partitions of $V$
\begin{equation*}
\mathcal{G}_{b_{1}, \ldots, b_{d}}
=
\left\{
I_{2\delta n_1-\delta b_1,\ldots,2\delta n_d-\delta b_d}
:
(n_1,\ldots,n_d)\in \mathbb{Z}_{\ge 0}^d
\right\},
\end{equation*}
with the shift of dimension $k$ indicated by $b_{k}\in\{0,1\}$. 
For any partition $\mathcal{G}$ of $V$ it holds that 
\[
\forall I \neq I' \in \mathcal{G},\quad I \cap I' = \varnothing,
\qquad
\bigcup_{I \in \mathcal{G}} I = V.
\]

\end{definition}
Figure~\ref{fig:locality_proof} illustrates a 2D partition in panel (a) (alongside the associated boustrophedon), while panel (b) shows the shifted versions of a single patch across each dimension, covering the cases $(b_1, b_2) \in \{(0,0), (0,1), (1,0), (1,1)\}$.
The following lemma establishes that this set of partitions is sufficient to guarantee that any two qubits within distance $\delta$ belong to at least one patch. 
Consequently, under the associated boustrophedon encoding, these qubits are guaranteed to remain close in the corresponding canonical ordering.
While these results focus strictly on grid partitioning, App.~\ref{app:locality_conservation} leverages them to prove proximity under the canonical ordering, ultimately deriving explicit bounds on the number of required interactions.

\begin{lemma}\label{lem:localPartition}
Consider pairs of geometrically local points $\|\mathbf x - \mathbf y\|_\infty \le \delta$ and the $2^d$ partitions $\mathcal{G}_b$.
At least one such partition includes a subgrid $I$ with $\mathbf x, \mathbf y \in I$.
\end{lemma}
\begin{proof}
We start with the one-dimensional grid $\{0,\dots,L-1\}$ and the two partitions
\begin{equation*}
\begin{aligned}
    \mathcal{G}_0 &= \{I_0, I_{2\delta}, I_{4\delta}, \ldots\}\\
    \mathcal{G}_1 &= \{I_{-\delta}, I_{\delta}, I_{3\delta}, \ldots\}.
\end{aligned}
\end{equation*}
Consider any two points with $|b-a|\le \delta$. Lets assume that they are contained in different subgrids of $\mathcal{G}_0$, meaning they must be in neighboring subgrids $a \in I_{i}$ and $b \in I_{i+2\delta}$. It follows
\[
i+\delta \le a < i+2\delta \le b < i+3\delta,
\]
and therefore $a,b \in I_{i+\delta}$ with $I_{i+\delta} \in \mathcal{G}_1$. 
This completes the one-dimensional setting.

As a direct generalization, consider the $d$-dimensional grid $\{0,\dots,L-1\}^d$ with partitions $\mathcal{G}_b$ and two points $\mathbf x$ and $\mathbf y$ with $\|\mathbf x- \mathbf y\|_\infty \le \delta$. 
Then, in each coordinate direction $k$, we have a one-dimensional interval $I_{\ell_k}$ such that
\[
{x_k, y_k} \in I_{\ell_k}.
\]
The Cartesian product of these intervals corresponds to a subgrid that contains both $\mathbf x$ and $\mathbf y$, which itself is part of a partition $\mathcal{G}_b$.
The special case of two dimensions is illustrated by an example in Fig.~\ref{fig:locality_proof}.
\end{proof}

\section{Asymptotic-optimality proof for 2D local lattice models}
\label{app:locality_conservation}

\begin{figure}
    \centering
    \includegraphics[width = 0.65\columnwidth]{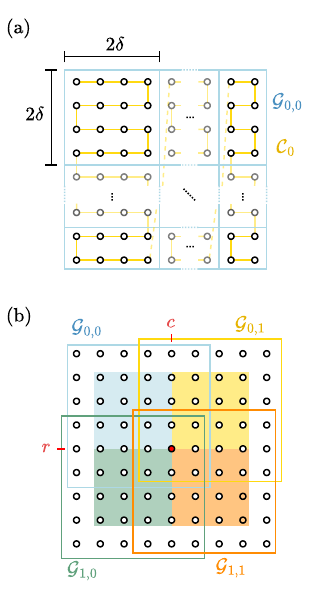}
    \caption{
    (a) The partitions $\mathcal{G}_{0,0}$ of the grid in blue and the boustrophedon encoding $m_{\mathcal{C}_{0}}$ in yellow corresponding to a Hamiltonian path through each subgrid.
    (b) Four overlapping subgrids of different partitions. The qubit $(r,c)$ is highlighted in red and any other qubit within the maximal interaction distance always falls in one of the four subgrids as indicated by the colored areas. 
    }
    \label{fig:locality_proof}
\end{figure}

Here, we show that all geometrically local fermionic interactions on a two-dimensional square lattice $G_{\mathrm{2D}} = (V_{\mathrm{2D}}, E_{\mathrm{2D}})$ with $V_{\mathrm{2D}} = \{0,\dots,L-1\}^2$ can be implemented on a corresponding qubit lattice using $\mathcal{O}(N)$ nearest-neighbor qubit interactions in $\mathcal{O}(\sqrt{N})$ depth.
The corresponding definition of nearest-neighbor edges is
\begin{equation*}
    E_{2D}=\{(\mathbf x, \mathbf y)\subset V_{2D}\times V_{2D} : \|\mathbf x - \mathbf y\|_1 = 1\}.
\end{equation*}
In this two-dimensional context, we label the indices with $r$ and $c$.
We consider a subgrid \(I_{r,c}\subset G_{\mathrm{2D}}\), as given in Def.~\ref{def:subgrid}, along with the induced edge set
\begin{equation*}
    E_{r,c}=\{(\mathbf x, \mathbf y) \in E_{\mathrm{2D}} : \mathbf x, \mathbf y \in I_{r,c}\}.
\end{equation*}
To proof the claimed statement, we first show how to implement all pairwise interactions between the modes contained within $I_{r,c}$.
For this we make following definition.

\begin{definition}[Hamiltonian path]
Let $m$ be a canonical ordering and let $I_{r,c}$ be a subgrid containing $|I_{r,c}|$ points with induced edge set $E_{r,c}$. 
Define
\[
    m_{\min} = \min_{\mathbf x \in I_{r,c}} m(\mathbf x).
\]
to be the minimal canonical rank among the points contained in the subgrid.
Furthermore, let
\[
m(I_{r,c}) = \{m(\mathbf x) : \mathbf x \in I_{r,c}\}
\]
denote the set of canonical ranks for all points in the subgrid.
We say that $m$ follows a Hamiltonian path through $I_{r,c}$ if
\[
    m(I_{r,c}) = \{m_{\min}, m_{\min}+1, \dots, m_{\min}+|I_{r,c}|-1\},
\]
and
\[
    \{m^{-1}(i), m^{-1}(i+1)\} \in E_{r,c}
\]
for all $i \in \{m_{\min},\dots,m_{\min}+|I_{r,c}|-2\}$.
\end{definition}

A canonical ordering that follows a Hamiltonian path through a subgrid $I_{r,c}$ of the qubit hardware enables local application of FSNs on that subgrid. 
Because each subgrid has constant size, all pairwise fermionic interactions within such a subgrid can be implemented with a constant number of gates in constant depth. 
By Lemma~\ref{lem:localPartition}, the lattice $G_\mathrm{2D}$ admits four shifted partitions into constant-size sublattices that collectively contain all geometrically local interactions, as showcase in Fig.~\ref{fig:locality_proof}(b). 
It therefore remains to construct an encoding whose canonical ordering restricts to a Hamiltonian path on every subgrid within such a partition.

\begin{theorem}[Two-dimensional locality preservation]
\label{thm:LocalImplementation}
All pairwise fermionic interactions between geometrically local modes satisfying $\|\mathbf x - \mathbf y\|_\infty \le \delta$ can be implemented using $\mathcal{O}(N)$ gates, sharing the same asymptotic depth scaling as the $C_{2D}$ transformation. 
Consequently, this implementation requires $\mathcal{O}(N)$ gates and achieves a depth of $\mathcal{O}(\sqrt{N})$ on a two-dimensional square qubit lattice, $\mathcal{O}(\log N)$ given all-to-all connectivity, or $\mathcal{O}(1)$ using lattice-surgery-based surface codes.
\end{theorem}
\begin{proof}
According to Definition~\ref{def:subgrid} we define four distinct partitions of the 2D grid: 
\begin{equation*}
\begin{aligned}
    \mathcal{G}_{0,0}&=\{ I_{r,c} : (r,c) \in \{0,2\delta,\ldots\}\times \{0,2\delta,\ldots\} \},\\
    \mathcal{G}_{0,1}&=\{ I_{r,c} : (r,c) \in \{0,2\delta,\ldots\}\times \{-\delta, \delta, \ldots\} \},\\
    \mathcal{G}_{1,0}&=\{ I_{r,c} : (r,c) \in \{-\delta, \delta, \ldots\}\times \{0,2\delta,\ldots\} \},\\
    \mathcal{G}_{1,1}&=\{ I_{r,c} : (r,c) \in \{-\delta, \delta, \ldots\}\times \{-\delta, \delta, \ldots\}\}.
\end{aligned}
\end{equation*}
Aligned with these partitions we  introduce the two boustrophedon encodings, $m_{\mathcal{C}_{0}}$ and $m_{\mathcal{C}_{1}}$, defined by column-partitions
\begin{equation*}
\begin{aligned}    
    \mathcal{C}_0 &= \{I_0, I_{2\delta}, I_{4\delta}, \ldots\},\\
    \mathcal{C}_1 &= \{I_{-\delta}, I_{\delta}, I_{3\delta}, \ldots\}.
\end{aligned}
\end{equation*}
It is straightforward to see that  $m_{\mathcal{C}_0}$ corresponds to a Hamiltonian path through each subgrid in $\mathcal{G}_{0,0}$ [see Fig.~\ref{fig:locality_proof}(a)] and $\mathcal{G}_{1,0}$.
Similar, $m_{\mathcal{C}_1}$ corresponds to a Hamiltonian path through each subgrid in $\mathcal{G}_{0,1}$ and $\mathcal{G}_{1,1}$. 
This allows us to implement all pairwise interactions between geometrically local modes with the sequence:
\begin{enumerate}
    \item Initialize in the $m_{\mathcal{C}_{0}}$ boustrophedon encoding.
    \item Apply an FSN within each subgrid in $\mathcal{G}_{0,0}$ in parallel.
    \item Apply an FSN within each subgrid in $\mathcal{G}_{1,0}$ in parallel.
    \item Switch from boustrophedon encoding $m_{\mathcal{C}_{0}}$ to $m_{\mathcal{C}_{1}}$ as described in Corollary~\ref{cor:bous2D}. 
    \item Apply an FSN within each subgrid in $\mathcal{G}_{0,1}$ in parallel.
    \item Apply an FSN within each subgrid in $\mathcal{G}_{1,1}$ in parallel.
\end{enumerate}
All steps except the encoding switches run in constant depth. 
Thus, the asymptotic gate-count and depth scaling is determined by the cost of switching between the boustrophedon encodings.
\end{proof}

Although both boustrophedon encodings and purely FSN-based approaches achieve circuit depth of $\mathcal O(\sqrt{N})$ for implementing local interaction graphs on a two-dimensional nearest-neighbor qubit lattice, boustrophedon encodings offer an key practical advantage.
For FSN-based constructions, increasing the maximum interaction distance by a constant amount,
\[
    \|x_i-x_j\|_\infty \leq \delta+\mathcal O(1),
\]
generally changes the leading-order prefactor of the circuit depth, i.e.,
\[
    C\sqrt{N}+o(\sqrt{N})\longrightarrow C'\sqrt{N}+o(\sqrt{N}).
\]
In contrast, under boustrophedon encodings, constant increases in the interaction range are absorbed into $\mathcal O(1)$ terms, leaving the leading-order contribution unchanged. 
This property is formalized in the following result.
\begin{corollary}\label{cor:AddingInteractions}
Adding fermionic interactions between modes satisfying 
\[
\|\mathbf x - \mathbf y\|_\infty \le \delta + \mathcal{O}(1)
\]
into the protocol of Thm.~\ref{thm:LocalImplementation} increases the circuit depth by at most an additive $\mathcal{O}(1)$ term. 
Consequently, this leaves the leading prefactor of the $\mathcal{O}(\sqrt{N})$ depth scaling unchanged.
\end{corollary}
\begin{proof}
Additional geometrically local fermionic interactions that preserve the maximum interaction bound,
\[
    \|\mathbf x-\mathbf y\|_\infty \le \delta,
\]
can be incorporated directly into the existing sequence of FSNs. 
It remains to consider the effect of expanding the maximum interaction distance to $\delta+\mathcal{O}(1)$, using boustrophedon patterns with larger width.
In general, the only non-constant-depth step is the encoding switch.
However, by Corollary~\ref{cor:bous2D}, the width $w$ of the $S$-patterns only contributes an additive $\mathcal{O}(w)$ factor to the circuit depth. 
Increasing this width by $\mathcal{O}(1)$, as required by the larger interaction distance, changes the total depth by at most an additive $\mathcal{O}(1)$ term.
\end{proof}

\section{Changing the dimensional hierarchy}\label{app:dimHier}

In this section we generalize our encoding switches to higher dimensions. 
To this end, we define a dimensional hierarchy and introduce Clifford circuits that facilitate changing this hierarchy. 
This allows us to switch between $d$-dimensional boustrophedon encodings, which preserve locality of higher-dimensional fermion lattices while enabling fermion routing in higher dimensional lattices.

\subsubsection*{Two-dimensions}
We begin by revisiting the two-dimensional $S$-pattern as introduced in Eq.~\eqref{eq:2D_S_pattern}.
\begin{definition}[Reversal function]
For $n \geq 1$, define the index-dependent permutation
\[
\rho_i^n \colon \{0,\ldots,n-1\} \to \{0,\ldots,n-1\}
\]
by
\[
\rho_i^n(x) :=
\begin{cases}
x, & i \equiv 0 \pmod{2},\\
n-1-x, & i \equiv 1 \pmod{2}.
\end{cases}
\]
\end{definition}
Using this reversal function, with the convention that 
\[
\rho_i(x) :=\rho^L_i(x),
\]
the two-dimensional snake-like $S$-pattern can be written as
\begin{equation*}
    m_S^{\mathrm{2D}}(r, c) =  \rho_r (c) + Lr.
\end{equation*}
Importantly, the $m_S(r,c)$ encoding imposes a dimensional hierarchy between the row and column indices. 
As the canonical rank $j$ increases, the inverse mapping
\[
\left(m_S^{\mathrm{2D}}\right)^{-1}(j) = (r,c)
\]
first traverses all $L$ column indices $c$ associated with a fixed row $r$. 
Only after these $L$ values have been exhausted does the row index advance. 
Thus, $c$ varies along consecutive ranks, whereas $r$ changes once every $L$ ranks.

We can define an $S$-pattern where the rows and columns are swapped in hierarchy as
\begin{equation*}
    m_{S}^{\mathrm{2D}}(c,r) = \rho_{c}(r)+Lc.
\end{equation*}
This $S$-pattern is almost equal to the Z pattern (cf. Eq.~\ref{eq:2D_Z_pattern})
\begin{equation*}
    m_{Z}^{\mathrm{2D}}(r,c) = \rho_{1}(r)+Lc,
\end{equation*}
up to an order reversal of each even column.
\begin{lemma}
\label{lem:C1D}
We consider $n$ qubits, which are indexed by
$\{0,\ldots,n-1\}$, and a corresponding ordering
$m(i)=i$.
The circuit
\begin{gather*}
C_\mathrm{1D}=\\
\left(\revprod{i=0}{n-2}\mathrm{CNOT}_{i,i+1}\right)^\dagger
\left(\prod_{i=0}^{n-2}\mathrm{CZ}_{i,i+1}\right) 
\left(\revprod{i=0}{n-2}\mathrm{CNOT}_{i,i+1}\right),
\end{gather*}
reverses this ordering to
\[
    m'(i)=n-1-i,
\]
up to single-qubit $Z$ gates applied to all qubits $\mathcal{I}_n$, where
\[
\mathcal{I}_n
=
\left\{ i \in \{0,\dots,n-2\} : (n-i) \bmod 2 = 1 \right\}.
\]
\end{lemma}
\begin{proof}
A valid CZ circuit for such an encoding switch can be obtained using Lemma~\ref{lem:FermPerm}, which is
\[
C_{\Lambda}=\prod_{i=0}^{n-2}\prod_{j\leq i}\mathrm{CZ}_{i+1,j}.
\]
Following Definition~\ref{def:CZpolynomial}, the corresponding $\mathbb{F}_2$ polynomial is
\[
T_{\Lambda}=\sum_{i=0}^{n-2}\sum_{j\leq i}x_{i+1}x_{j}.
\]
We will now show that this is equal to $C_{\mathrm{1D}}$, up to single-body $Z$ corrections.
For this, we note that the parity labels induced by the CNOT ladder are
\[
\ell(i) = \{j: j \leq i\}.
\]
According to Lemma~\ref{lem:FlowCZ},
\begin{gather*}
\left(\revprod{j=0}{n-2}\mathrm{CNOT}_{j,j+1}\right)^\dagger
\mathrm{CZ}_{i,i+1}
\left(\revprod{j=0}{n-2}\mathrm{CNOT}_{j,j+1}\right)\\
=
\prod_{j\in \ell(i)}
\prod_{k\in \ell(i+1)}
\mathrm{CZ}_{j,k},
\end{gather*}
which corresponds to the $\mathbb{F}_2$ polynomial
\[
T_{i}
=
\left(\sum_{j\leq i} x_j \right)\left(\sum_{j\leq i+1} x_j \right)
=
\sum_{j\leq i} x_j + \sum_{j\leq i} x_{i+1}x_{j}
,
\]
for a fixed index $i$.
Therefore,
\[
T_{\mathrm{1D}}=
\sum_{i=0}^{n-2}T_i = T_{\Lambda} + \sum_{i=0}^{n-2} \left((n-i)\bmod 2\right)x_{i},
\]
where the additional diagonal terms can be canceled via the stated single-body $Z$ gates on qubits in $\mathcal{I}_n$.
\end{proof}
This result provides a straightforward way to reverse the ordering within the full set of qubits or any subset thereof. 
Utilizing $C_{\mathrm{2D}}$ and subsequently reversing the order within each even column allows to switch between $m_{S}^{\mathrm{2D}}(r,c)$ and $m_{S}^{\mathrm{2D}}(c,r)$, thereby changeing the dimensional hierarchy.
At the same time, these additional reversals preserve the asymptotic resource scaling of $C_\mathrm{2D}$.

\subsubsection*{Three-dimensions}
We extend the $S$-pattern to three dimensions by considering the cubic index set
\[
V_{\rm 3D} = \{0,\ldots,L-1\}^3.
\]
In addition to the row and column indices $r$ and $c$, we introduce $p$ as the plane index.

The three-dimensional $S$-pattern can be viewed as a stack of two-dimensional $S$-patterns along the third dimension as visualized in Fig.~\ref{fig:3D_switch}(b). 
As $p$ increases, every second plane is traversed in reverse order. 
This construction mirrors the two-dimensional case, in which consecutive rows are placed next to one another and every second row is traversed in reverse order.

\begin{definition}[Three-dimensional $S$-pattern]
\label{def:S3D}
We define the three-dimensional $S$-pattern as follows
\begin{equation*}
\begin{aligned}
m^{\mathrm{3D}}_S(p,r,c) 
&=  
\rho_p^{L^2}(m_S^{\mathrm{2D}}(r,c))+ L^2 p
\\
&= 
\rho_{p+r}(c) + L\rho_p(r)  + L^2p.
\end{aligned}
\end{equation*} 
\end{definition}
Similar to the 2D case, we demonstrate how to implement efficient permutations within this three-dimensional hierarchy. 

\begin{lemma}
\label{lem:3D-lower-dimension-exchange}
Consider the three-dimensional $S$-pattern ordering
$m_S^{\mathrm{3D}}(p,r,c)$.
The two lowest dimensions can be exchanged
\[
m_S^{\mathrm{3D}}(p,r,c)
\longmapsto
m_S^{\mathrm{3D}}(p,c,r),
\]
using $\mathcal{O}(N)$ gates. 
On a three-dimensional nearest-neighbor-connected qubit lattice, this transformation can be implemented with depth $\mathcal O\left(N^{1/3}\right)$.
On an all-to-all connected architecture, it can be implemented with depth $\mathcal{O}(\log N)$.
\end{lemma}
\begin{proof}
By definition,
\[
m_S^{\mathrm{3D}}(p,r,c)
=
\rho_p^{L^2}\!\left(m_S^{\mathrm{2D}}(r,c)\right)
+
L^2p.
\]
Hence, for each fixed plane index $p$, the ordering within the plane is either the two-dimensional $S$-pattern or its reversal:
\[
m_S^{\mathrm{3D}}(p,r,c)
=
\begin{cases}
m_S^{\mathrm{2D}}(r,c)+L^2p,
    & p\equiv 0,\\[2mm]
L^2-1-m_S^{\mathrm{2D}}(r,c)+L^2p,
    & p\equiv 1.
\end{cases}
\]

It therefore suffices to exchange the row and column coordinates independently within each plane. 
On even-indexed planes, the required transformation is
\[
m_S^{\mathrm{2D}}(r,c)
\longmapsto
m_S^{\mathrm{2D}}(c,r),
\]
which is precisely the previously established two-dimensional dimension-exchange transformation.

On odd-indexed planes, we instead require
\[
L^2-1-m_S^{\mathrm{2D}}(r,c)
\longmapsto
L^2-1-m_S^{\mathrm{2D}}(c,r).
\]
We distinguish two cases according to the parity of $L$.
Suppose first that $L$ is odd. 
Since $L-1$ is even, the reflected row index $L-1-r$ has the same parity as $r$. 
Consequently, the reversed two-dimensional ordering can be expressed as
\[
L^2-1-m_S^{\mathrm{2D}}(r,c)
=
m_S^{\mathrm{2D}}(L-1-r,L-1-c).
\]
Define reflections of coordinates by
\begin{align*}
\mathfrak r_{1,2}(x,y)&\coloneqq(L-1-x,L-1-y),\\[0.4em]
\left(
m_S^{\mathrm{2D}} 
\circ \mathfrak r_{1,2}
\right)
(r,c)
&=
m_S^{\mathrm{2D}}(L-1-r,L-1-c).
\end{align*}
We apply a circuit in this reflected coordinate system by relabeling each of its qubit indices according to $\mathfrak r_{1,2}$. 
Equivalently, the reflected circuit is obtained by conjugation,
\[
\mathcal C^{\langle 1,2\rangle}
=
\mathfrak r_{1,2}^{-1}\circ\mathcal C\circ\mathfrak r_{1,2}.
\]
Thus, coordinate lines parallel to each axis are traversed in the opposite physical direction, while the circuit structure itself is unchanged.
Consequently, applying the usual two-dimensional exchange construction in the reflected coordinates yields
\begin{gather*}
\left(
m_S^{\mathrm{2D}}
\circ \mathfrak r_{1,2}
\right)
(r,c) 
\longmapsto 
\left(
m_S^{\mathrm{2D}}
\circ \mathfrak r_{1,2}
\right)
(c,r) 
\end{gather*}
which is the desired transformation on odd-indexed planes.

For even $L$, reversal of the canonical ordering within a plane can be expressed as a reflection of the row coordinate. 
Indeed,
\[
L^2-1-m_S^{\mathrm{2D}}(r,c)
=
m_S^{\mathrm{2D}}(L-1-r,c).
\]
Consequently, applying the usual two-dimensional exchange construction after reflecting the first coordinate 
\begin{align*}
\mathfrak r_{1}(x,y)&\coloneqq(L-1-x,y),\\[0.4em]
\left(
m_S^{\mathrm{2D}}
\circ \mathfrak r_{1}
\right)
(r,c) 
&=
m_S^{\mathrm{2D}}(L-1-r,c),\\[0.4em]
\mathcal C^{\langle 1\rangle}
&=
\mathfrak r_1^{-1}\circ\mathcal C\circ\mathfrak r_1,
\end{align*}
yields
\[
\left(
m_S^{\mathrm{2D}}
\circ \mathfrak r_{1}
\right)
(r,c) 
\longmapsto
\left(
m_S^{\mathrm{2D}} 
\circ \mathfrak r_{1}
\right)
(c,r)
.
\]
Note
\begin{align*}
\left(
m_S^{\mathrm{2D}}
\circ \mathfrak r_{1}
\right)
(c,r) 
&=
m_S^{\mathrm{2D}}(L-1-c, r),\\
&=
L^2-1-m_S^{\mathrm{2D}}(c, r), 
\end{align*}
which is the desired ordering on odd-indexed planes.
\end{proof}

Switching the hierarchy of rows and planes, for example
\[
m^{\mathrm{3D}}_S(p,r,c) \to m^{\mathrm{3D}}_S(r,p,c),
\]
requires additional considerations.
Fundamentally, we are looking at the same permutation problem, but formulated for the column ordering
\begin{align*}
m^{\mathrm{3D}}_S(p,r,c)&=\rho_{p+r}(c) + L\rho_p(r)  + L^2p \\
& =\rho_{p+r}(c)+Lm_S^{\mathrm{2D}}(p,r),
\end{align*}
that transforms to
\[
m^{\mathrm{3D}}_S(r,p,c)=\rho_{p+r}(c)+Lm_S^{\mathrm{2D}}(r,p).
\]
That means, 
for each coordinate of row $r$ and plane $p$ there is a line of qubits with $c\in \{0, \ldots, L\}$ for which we do not have to change the ordering of qubits inside this line, but have to change which line is before which other line in the ordering.
%
To efficiently perform such column-wise swaps, we introduce an additional tool.

\begin{lemma}
\label{lem:columnSwap}
We consider $n$ qubits, indexed by $[n]=\{0,\ldots,n-1\}$, with a canonical ordering $m$ and a partition into two disjoint sets of qubit indices
\[
[n]= A\sqcup B,
\]
where the JW-path first visits all the qubits in $A$ and afterwards all the qubits in $B$, i.e, 
\begin{align*}
m(A) &= \{0,\ldots,|A|-1\},\\
m(B) &= \{|A|,\ldots,n-1\}.
\end{align*}
For each $C \in \{A,B\}$, choose an arbitrary enumeration
\[
c : \{1,\ldots,|C|\} \rightarrow C, \quad c \text{ is bijective}
\]
with a terminal qubit denoted by
\[
x_C = c(|C|).
\]
Writing \(c_i = c(i)\), define an associated CNOT ladder by
\[
L_C
=
\revprod{i=1}{|C|-1}\mathrm{CNOT}_{c_i, c_{i+1}}.
\]
Then the circuit
\[
C_{A\leftrightarrow B}=\left(L_A L_B\right)^\dagger\mathrm{CZ}_{x_A,x_B}(L_AL_B)
\]
swaps the order of the two sets
\[
m'(i)=
\begin{cases}
m(i)+|B|, & i\in A,\\[2mm]
m(i)-|A|, & i\in B.
\end{cases}
\]
\end{lemma}
\begin{proof}
A valid CZ circuit for the encoding switch of interest can be constructed using Lemma~\ref{lem:FermPerm}, which is
\begin{equation*}
    C_{\Lambda} = \prod_{i \in A}\prod_{j \in B}\mathrm{CZ}_{i,j}.
\end{equation*}
The key property of the CNOT-ladder construction is that each terminal qubit stores the parity of the qubits in its corresponding set. 
Thus, the terminal qubits $x_A$ and $x_B$ accumulate the labels of all $A$ and $B$, respectively:
\[
\ell(x_A)=A,
\qquad
\ell(x_B)=B.
\]
Then, according to Lemma~\ref{lem:FlowCZ},
\[
C_{A\leftrightarrow B}=C_{\Lambda}.
\]
\end{proof}
Lemma~\ref{lem:columnSwap} allows us to use CNOT ladders to transform permutations of sets of qubits into permutations of individual qubits.
We can now use this property to simplify the swapping of hierarchies.
\begin{definition}[Compressing a dimension]
\label{def:compressDimension}
We consider a cubic qubit lattice
\[
[L]\coloneqq\{0,\ldots,L-1\},
\qquad
V= [L]^3,
\]
and coordinate lines parallel to the column-axis, defined for each $(r,p)\in[L]^2$ by
\[
F^c_{r,p}
\coloneqq
\{(p,r,c):c\in[L]\}
=
\{p\}\times\{r\}\times[L].
\]
Analogous constructions for coordinate lines parallel to the $r$- and $p$-axes will be used throughout without further specification.
Let $m$ be a canonical ordering such that every $c$-coordinate line occupies a contiguous block of $L$ indices. 
Thus, for each $(p,r)\in[L]^2$, there exists a unique $n_{p,r}\in\{0,1,\ldots, L^2-1\}$ such that
\[
m\left(F^c_{p,r}\right)
=
\{n_{p,r}L,n_{p,r}L+1,\ldots,n_{p,r}L+L-1\}.
\]
This is equivalent to saying that the canonical ordering $m$ allows for a well defined order of $c$-coordinate lines.
We refer to this as the representative ordering
\[
m_c(p,r)\coloneqq n_{p,r}.
\]
If $m$ allows for such a representative ordering, it is called $c$-compressible.
Furthermore, we refer to performing the CNOT ladders
\[
C_{c}\coloneqq\prod_r \prod_p \left(\revprod{c=0}{L-2} \mathrm{CNOT}_{(p,r,c), (p,r,c+1)} \right),
\]
as compressing the dimension $c$.
After compressing the $c$-dimension, the terminal qubits of the $c$-coordinate lines form the representative plane
\[
P_{c}
\coloneqq
\{(p,r,L-1):p,r\in[L]\}
=
[L]\times[L]\times\{L-1\}.
\]
\end{definition}
By Lemma~\ref{lem:columnSwap}, the order of two lines that are adjacent in the representative ordering,
\[
\bigl|m_c(i,j)-m_c(i',j')\bigr|=1,
\]
can be exchanged by applying a single
\[
\mathrm{CZ}_{(i,j,L-1),(i',j',L-1)}
\]
gate in the representative plane.
Consequently, switches of the representative ordering $m_c$, implemented by $\mathrm{CZ}$ gates between qubits on $P_{c}$, can be treated equivalently as two-dimensional ordering switches of a standard canonical ordering.

Notable, changing the representative ordering only shifts each $c$-coordinate line as a whole and therefore leaves the canonical ordering within each line unchanged. More precisely,
\[
m'(i,j,k)
=
m(i,j,k)
+
L\bigl(m'_c(i,j)-m_c(i,j)\bigr).
\]
This discussion concludes our toolkit for three-dimensional hierarchy exchanges, enabling us to prove the following main result.
\begin{theorem}[Three-dimensional hierarchy switching]
\label{thm:3DimTransposition}
Consider an $m^{\mathrm{3D}}_S(p,r,c)$ ordering as defined in Def.~\ref{def:S3D}. 
It is possible to perform any permutation of the dimensional hierarchy using $\mathcal{O}(N)$ gates. 
On a three-dimensional qubit lattice, these transformations can be implemented with depth $\mathcal{O}(N^{\frac{1}{3}})$, while an all-to-all connected architecture permits depth $\mathcal{O}(\log N)$.
\end{theorem}
\begin{proof}
Every permutation of the dimensional hierarchy can be decomposed into a sequence of adjacent exchanges, namely exchanges of the two lowest dimensions and of the two highest dimensions.
Exchanging the lowest two dimensions is already discussed by Lemma~\ref{lem:3D-lower-dimension-exchange}.

For the missing permutation
\[
m^{\mathrm{3D}}_S(p,r,c)\to m^{\mathrm{3D}}_S(r,p,c),
\]
we use that 
\[
m^{\mathrm{3D}}_S(p,r,c)= \rho_{p+r}(c) + L\rho_p(r)+L^2p,
\]
is column-compressible, as introduced in Definition~\ref{def:compressDimension}. 
The corresponding representative ordering is
\[
    m_{c}(p,r)=\rho_p(r)+L p,
\]
which is equivalent to a two-dimensional $S$-pattern on the representative plane $P_{c}$.
Here, we make use of circuits for two-dimensional lattices, on the representative plane $P_{c}$, to change the representatives ordering to
\[
    m'_{c}(p,r)=\rho_r(p)+L r,
\]
which corresponds to the canonical ordering
\begin{align*}
m'(p,r,c)
&=m(p,r,c)+L\left(m'_c(p,r)-m_c(p,r)\right),\\
&=\rho_{p+r}(c)+L\rho_r(p)+L^2 r,\\
&=m^{\mathrm{3D}}_S(r,p,c).
\end{align*}
Therefore, we have performed the necessary switch in hierarchy, using only the promised resource scaling.
\end{proof}

\begin{figure*}
    \centering
    \includegraphics{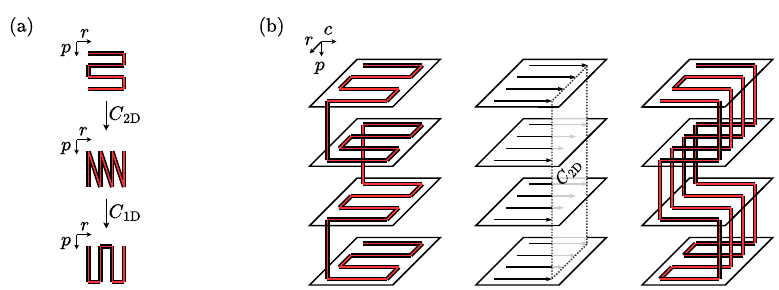}
    \caption{3D setup. 
    (a) The effect of applying $C_{\mathrm{2D}}$ in the compressed column basis and including $C_{\mathrm{2D}}$ to obtain a 2D transposition in the row-plane basis.
    (b) Switching between 3D snake pattern example for a $4 \times 4 \times 4$ grid. Left: the 3D snake pattern $m^{\mathrm{3D}}_{S}(p, r, c)$. 
    Middle: the direction of the CNOT ladders to compress the columns indicated by the arrows and the surface where $C_{\mathrm{2D}}$ up to $C_{\mathrm{1D}}$ is applied. 
    Right: the resulting 3D snake encoding $m^{\mathrm{3D}}_{S}(r, p, c)$ such that $r$ and $p$ switched in hierarchy. 
    }
    \label{fig:3D_switch}
\end{figure*}

\subsubsection*{$d$-dimensions}

We can generalize the concept of a dimensional hierarchy to $d$-dimensions.
\begin{definition}[Dimensional hierarchy]
Let
\[
D\coloneqq\{1,\ldots,d\}
\]
be the set of dimensions. 
A dimensional hierarchy is given as a permutation $\pi$ of $D$
\[
    \boldsymbol \sigma = (\sigma_1,\dots,\sigma_d) =  (\pi(1),\pi(2),\dots,\pi(d)),
\]
which is defining an ordering.
We refer to $\sigma_1$ as the lowest dimension in the hierarchy and $\sigma_d$ as the highest.
\end{definition}

Through extending the recursion started in Definition~\ref{def:S3D}, we can define an $S$-pattern for a $d$-dimensional qubit grid, while considering an arbitrary dimensional hierarchy.
\begin{definition}[$d$-dimensional $S$-pattern]
\label{def:dD_S_pattern}
Given a dimensional hierarchy $\sigma$, we define the corresponding $d$-dimensional $S$-pattern encoding as a bijection
\[
    m_S^\sigma : \{0,\ldots,L-1\}^d \to \{0,\dots,N-1\}.
\]
For an index
\[
\mathbf v \in \{0,\ldots,L-1\}^d
\]
and a coordinate ordering $\boldsymbol \sigma$, define
\[
\mathbf v_{\boldsymbol{\sigma}}
:=
(v_{\sigma_1}, v_{\sigma_2}, \ldots),
\quad
\boldsymbol{\sigma}_{<i}
:=
(\sigma_1,\ldots,\sigma_{i-1}).
\]
The canonical ordering is then defined recursively by
\[
m_{S}^{d}(\mathbf v_{\boldsymbol{\sigma}})
=
\rho_{v_{\sigma_d}}^{L^{d-1}}
\left(
m_{S}^{d-1}
\bigl(\mathbf v_{\boldsymbol{\sigma}_{<d}}\bigr)
\right)
+
L^{d-1}v_{\sigma_d},
\]
with base case
\[
m_{S}^{2}(v_{\sigma_1},v_{\sigma_2})
=
m_S^{\mathrm{2D}}(v_{\sigma_2},v_{\sigma_1}).
\]
This corresponds to
\[
m_{S}^{d}(\mathbf v_{\boldsymbol \sigma})
=
\sum_{i=1}^d L^{i-1} 
\rho_{\varepsilon_i^d(\mathbf v_{\boldsymbol{\sigma}})}
\bigl(v_{\sigma_i}\bigr),
\]
where we have used 
\[
\varepsilon_i^d(\mathbf v_{\boldsymbol{\sigma}})
\coloneqq
\left(
\sum_{k=i+1}^{d} v_{\sigma_k}
\right)
\bmod 2,
\quad
\varepsilon_{d+1}^d(\mathbf v_{\boldsymbol{\sigma}})\coloneqq 0.
\]
\end{definition}

Note that in the preceding definition we reversed our standard argument order for $m$.
This choice simplifies the recursive definition of $m_S^d(\mathbf v_{\mathbf \sigma})$, as it is more natural to truncate tuple entries from the right than from the left.

Equivalently to the three-dimensional case, a higher dimensional $S$-pattern allows efficient switches of its lowest two dimensions in hierarchy.
\begin{lemma}
\label{lem:d-dim-lower-dimension-exchange}
Consider the $d$-dimensional $S$-pattern ordering
$m_S^{d}(v_1, \ldots, v_d)$.
The two lowest dimensions can be exchanged,
\[
m_S^{d}(v_1, v_2, v_3, \ldots)
\longmapsto
m_S^{d}(v_2, v_1, v_3, \ldots)
\]
using $\mathcal{O}(N)$ gates. 
On a $d$-dimensional nearest-neighbor-connected qubit lattice, this transformation can be implemented with depth $\mathcal O\left(N^{1/d}\right)$.
On an all-to-all connected architecture, it can be implemented with depth $\mathcal{O}\bigl(\log \left( N^{1/d} \right)\bigr)$.
\end{lemma}
\begin{proof}
We begin by separating the contributions of the first two coordinates:
\begin{align*}
m_S^{d}(\mathbf v)
&=
\rho_{\varepsilon_1^d(\mathbf v)}^{L}(v_1)
+
L\rho_{\varepsilon_2^d(\mathbf v)}^{L}(v_2)
+
\sum_{i=3}^{d}
L^{i-1}
\rho_{\varepsilon_i^d(\mathbf v)}(v_i)\\
&=
\rho_{\varepsilon_2^d(\mathbf v)}^{L^2}
\left(
\rho_{v_2}^{L}(v_1)+Lv_2
\right)
+
\sum_{i=3}^{d}
L^{i-1}
\rho_{\varepsilon_i^d(\mathbf v)}(v_i)\\
&=
\rho_{\varepsilon_2^d(\mathbf v)}^{L^2}
\left(
m_S^{\mathrm{2D}}(v_2,v_1)
\right)
+
\sum_{i=3}^{d}
L^{i-1}
\rho_{\varepsilon_i^d(\mathbf v)}(v_i).
\label{eq:dD-ordering-plane-decomposition}
\end{align*}

For fixed coordinates $v_3,\ldots,v_d$, define the corresponding plane offset by
\[
B(v_3,\ldots,v_d)
\coloneqq
\sum_{i=3}^{d}
L^{i-1}
\rho_{\varepsilon_i^d(\mathbf v)}^{L}(v_i).
\]
Since this quantity is independent of $v_1$ and $v_2$,
the ordering within each $(v_1,v_2)$-coordinate plane is either the two-dimensional $S$-pattern or its reversal:
\begin{gather*}
m_S^{d}(\mathbf v)
=\\
\begin{cases}
m_S^{\mathrm{2D}}(v_2,v_1)
+
B(v_3,\ldots,v_d),
&
\varepsilon_2^d(\mathbf v)=0,
\\[2mm]
L^2-1-m_S^{\mathrm{2D}}(v_2,v_1)
+
B(v_3,\ldots,v_d),
&
\varepsilon_2^d(\mathbf v)=1.
\end{cases}
\end{gather*}
It therefore suffices to exchange the row and column coordinates independently within each $(v_1,v_2)$-coordinate plane, which can be done analogously to Lemma~\ref{lem:3D-lower-dimension-exchange}.
\end{proof}
We also present a straightforward generalization of the compression scheme to higher-dimensional lattices.
\begin{definition}[d-dimensional compressing]
\label{def:dcompression}
We consider a $d$-dimensional qubit lattice
\[
[L]\coloneqq\{0,\ldots,L-1\},
\qquad
V\coloneqq[L]^d.
\]
For \(a\in\{1,\ldots,d\}\) and
\[
\mathbf u=(u_1,\ldots,u_{d-1})\in[L]^{d-1},
\]
define the embedding
\[
\iota_{\mathbf u}^{a}\colon [L]\to V,
\qquad
\iota_{\mathbf u}^{a}(t)
\coloneqq
(u_1,\ldots,u_{a-1},t,u_a,\ldots,u_{d-1}).
\]
The coordinate line parallel to the \(a\)-th axis and specified by
\(\mathbf u\) is then
\[
F_{\mathbf u}^{a}
\coloneqq
\left\{
\iota_{\mathbf u}^{a}(t):t\in[L]
\right\}.
\]

Let $m\colon V\to[L^d]$ be a canonical ordering such that every $a$-coordinate line occupies a contiguous block of $L$ indices.
Thus, for every $\mathbf u\in[L]^{d-1}$, there exists a unique
$n_{\mathbf u}\in\{0,1,\ldots,L^{d-1}-1\}$ such that
\[
m\bigl(F_{\mathbf u}^{a}\bigr)
=
\{n_{\mathbf u}L,\ldots,n_{\mathbf u}L+L-1\}.
\]
We define the representative ordering of the $a$-coordinate lines by
\[
m_a(\mathbf u)\coloneqq n_{\mathbf u}.
\]
For each $\mathbf u\in[L]^{d-1}$, define the CNOT ladder
along $F_{\mathbf u}^{a}$ by
\[
\mathcal L_{\mathbf u}^{a}
\coloneqq
\revprod{i=0}{L-2}\mathrm{CNOT}_{\iota_{\mathbf u}^a(i),\iota_{\mathbf u}^a(i+1)}.
\]
We refer to the simultaneous application of these ladders along all $a$-coordinate lines as compression along dimension $a$, and denote the corresponding circuit by
\[
C_a
\coloneqq
\prod_{\mathbf u\in[L]^{d-1}}
\mathcal L_{\mathbf u}^{a}.
\]
We refer to the set of terminal qubits of the CNOT ladders,
\[
V_a
\coloneqq
[L]^{a-1}\times\{L-1\}\times[L]^{d-a},
\]
as the representative lattice.
Since the $a$-th coordinate is fixed at $L-1$, $V_a$ is naturally isomorphic to the $(d-1)$-dimensional lattice $[L]^{d-1}$.
\end{definition}

Successive dimension compressions allow arbitrary modifications of the dimensional hierarchy. 
For the applications considered here, however, it is sufficient to move a selected dimension either to the top or to the bottom of the hierarchy. 
Moving a dimension to the top is required for boustrophedon encodings, whereas moving it to the
bottom is needed for fermion routing.
A key detail for what follows is that, on a fully connected device, any dimension can be moved to the lowest level of the hierarchy with depth $\mathcal{O}(\log N)$. 
Since this depth depends only on the system size $N$, rather than on the number of dimensions $d$, it enables
efficient fermion-routing schemes even when the dimensionality grows with the system size.

\begin{theorem}[$d$-dimensional hierarchy switching]
\label{thm:dDimTransposition}
The ordering $m_S^d(v_1,\ldots,v_d)$ can be transformed into either $m_S^d(v_2,\ldots,v_d,v_1)$ or $m_S^d(v_d,v_1,\ldots,v_{d-1})$ using $\mathcal{O}(N)$ gates. On a $d$-dimensional qubit lattice, these transformations can be implemented with depth
$\mathcal{O}\left(dN^{1/d}\right)$, while an all-to-all connected architecture permits depth $\mathcal{O}(\log N)$.
\end{theorem}
\begin{proof}
The ordering
\[
m_S^{d}(v_1,\ldots, v_d)
=
\sum_{i=1}^d
L^{i-1}
\rho_{\varepsilon_i^d(\mathbf v)}\bigl(v_i\bigr)
\]
is compressible along the first axis. Its representative ordering is
\[
m_1(v_2,\ldots,v_d)
=
\sum_{i=2}^d
L^{i-2}
\rho_{\varepsilon_i^d(\mathbf v)}\bigl(v_i\bigr),
\]
which has the same form as the ordering $m_S^{d-1}(v_2, \ldots,v_d)$.
It therefore has the same $S$-pattern structure as the original ordering, now on the representative $(d-1)$-dimensional lattice.
Iterating this, after applying $a-1$ successive compressions on the shrinking representative lattices, we arrive at the representative ordering
\begin{align*}
m_{1,\ldots,a-1}(v_a,\ldots,v_d)
&=
\sum_{i=a}^d
L^{i-a}
\rho_{\varepsilon_i^d(\mathbf v)}\bigl(v_i\bigr),\\
&=
m_S^{d-(a-1)}(v_a,\ldots,v_d).   
\end{align*}
According to Lemma~\ref{lem:d-dim-lower-dimension-exchange}, the ordering $m_S^{d-(a-1)}(v_a,\ldots,v_d)$ allows dimensions $a$ and $a+1$ to be exchanged efficiently.
Consequently, any adjacent transposition in the dimensional hierarchy can be implemented efficiently.

Furthermore, the transformation
\[
m_S^d(v_1,\ldots,v_d) \longmapsto m_S^d(v_2,\ldots,v_d,v_1)
\]
can be implemented by successively moving dimension $1$ upward through the hierarchy and compressing each exchanged dimension afterward. 
Conversely, to transform
\[
m_S^d(v_1,\ldots,v_d)
\longmapsto
m_S^d(v_d,v_1,\ldots,v_{d-1}),
\] 
we first compress dimensions until dimension $d$ can be moved downward through the hierarchy. We then undo the compressions one at a time, exchanging dimension $d$ with each newly uncompressed dimension.

Hence, on a $d$-dimensional qubit lattice, these transformations can likewise be implemented with depth
$\mathcal{O}\left(dN^{1/d}\right)$.
On an all-to-all connected qubit architecture, each of the
$\mathcal{O}(d)$ successive compression and hierarchy exchange steps can be implemented with depth $\mathcal{O}\left(\log\left(N^{1/d}\right)\right)$.
Since
\[
d\log\left(N^{1/d}\right)=\log N,
\]
the resulting total depth is $\mathcal{O}(\log N)$.
\end{proof}

\section{$d$-Dimensional local lattice models}\label{app:dD}
In this section, we extend our 2D boustrophedon encodings first to $3$ dimensions and subsequently to $d$-dimensions, employing the hierarchy-changing tools introduced in Appendix~\ref{app:dimHier}.
Analogous to the 2D case in Appendix~\ref{app:locality_conservation}, we then demonstrate how to preserve locality in these higher-dimensional settings.
Specifically, we show how to implement geometrically local fermionic interactions using $\mathcal{O}(N)$ nearest-neighbour qubit gates in $\mathcal{O}(N^{\frac{1}{d}})$ depth on a $d$-dimensional qubit lattice $V=\{0,\ldots,L-1\}^d$. 

\subsubsection*{3-Dimensional lattice models}
We begin with the 3D case, where the goal is to construct a transformation of depth $\mathcal{O}(N^{\frac{1}{3}})$ that allows us to switch between $\mathrm{3D}$ boustrophedon encodings.
Analogous to the two-dimensional case, we partition the three-dimensional grid into three-dimensional subgrids with side lenghts scaling with the maximum interaction distance $\delta=\mathcal{O}(1)$, such that the boustrophedon encoding restricts to a Hamiltonian path through each subgrid of a given partition.

\begin{definition}[Subgrid ordering]\label{def:suborder}
Let the set of qubits $V = \{0,\dots,N-1\}$ be partitioned into contiguous disjoint subsets indexed by $J$:
\[
\mathcal{B}=\{B_j\}_{j\in J}, \qquad B_j\subseteq V \quad \text{for all } j\in J.
\]
Furthermore, assume all subsets are equipped with an ordering $m_{B_{j}}$.
We then define a canonical ordering on the full grid as
\[
m_{\mathcal{B}}(i) = m_{B_{j}}(i-|B_{1:j-1}|)+|B_{1:j-1}|
\]
for $i\in B_{j}$, where $|B_{1:j-1}|=\sum_{k=1}^{j-1}|B_{k}|$.

An important special case occurs when the qubits are associated with an underlying grid, making it natural to index them via coordinate tuples $\mathbf{x}=(x_{1}, \ldots, x_{d}) \in \{1,\dots,L\}^d$.
In this case, we restrict each subset $B_{\mathbf{j}}$ to be a geometrically local subgrid defined by the Cartesian product
\[
B_\mathbf{j} = C_{j_{1}}^{(1)}\times \dots \times C_{j_{d}}^{(d)}.
\]
Here, $C_{1}^{(i)}, \dots, C_{k_i}^{(i)}$ denotes a partition of $\{1,\dots,L\}$ into intervals of consecutive indices in dimension $i$ (specified by start values and widths). The index $j_i$ specifies the interval in dimension $i$ containing $x_i$, i.e., $x_i \in C_{j_i}^{(i)}$ for each $i$.
Assuming a given ordering $m(\mathbf{j})$ over the subgrids $B_\mathbf{j}$, the total ordering on $V$ induced by the subgrid orderings and the inter-subgrid ordering is expressed as
\begin{align*}
m_{\mathcal{B}}(\mathbf{x}) = &m _{B_\mathbf{j'}}(x_{1}-|C_{1:j'_{1}-1}|, \ldots, x_{d}-|C_{1:j'_{d}-1}|)\\
&+|B_{1:m(\mathbf{j'})-1}|
\end{align*}
for $\mathbf{x}\in B_{\mathbf{j'}}$, where $|C_{1:j'_{l}-1}|=\sum_{k=1}^{j'_{l}-1}|C_{k}^{(l)}|$.
\end{definition}




\color{black}

Note that the 2D boustrophedon encoding in Eq.~\eqref{eq:meander_order} also follows the general form of Definition~\ref{def:suborder}.
An other important example are 3D boustrophedon encodings which we can explicitly define as follows.

\begin{definition}[3D boustrophedon encoding]\label{def:3D_boustrophedon}
A three-dimensional boustrophedon encoding is specified by two interval partitions
\[
\begin{aligned}
    \mathcal{C}_c &= \{C_1^{(c)},\ldots,C_{K_c}^{(c)}\}\\
    \mathcal{C}_r &= \{C_1^{(r)},\ldots,C_{K_r}^{(r)}\},
\end{aligned}
\]
of the column and row index sets, respectively. Each interval
$C_i^{(x)}$, for $x\in\{c,r\}$, consists of consecutive indices,
\[
    C_i^{(x)}
    =
    \{a_i^{(x)},a_i^{(x)}+1,\ldots,a_i^{(x)}+w_i^{(x)}-1\},
\]
where $w_i^{(x)}$ denotes the width of the interval and $a_i$ the start index.
Each pair of column and row interval then describes a subgrid
\[
B_{i,j} = C^{(c)}_{i} \times C^{(r)}_{j} \times [L]
\]
where we order the subgrids according to a snake pattern $m^\mathrm{2D}_{S}(i,j)$.
The 3D boustrophedon related to the partition is
\begin{equation}
\label{eq:3D_boustrophedon}
m_{\mathcal{C}_r,\mathcal{C}_c}(c,r,p)= m^{\mathrm{3D}}_{S}(c_{i'}, r_{j'}, p) + \sum_{(i,j)\in J}|B_{i,j}|
\end{equation}
where $i'$ and $j'$ are the indices such that $(c,r,p)\in B_{i',j'}$, and accordingly 
\[
\begin{aligned}
c_{i'} &= c-\sum_{i=1}^{i'-1}|C^{(c)}_{i}|\\
&\text{and}\\
r_{j'} &= r-\sum_{j=1}^{j'-1}|C^{(r)}_{j}|.
\end{aligned}
\]
\end{definition}
Following this definition, the ordering within the subgrid $B_{i',j'}$ under a 3D boustrophedon encoding is given by
\[
m^{\mathrm{3D}}_S(c_{i'}, r_{j'}, p) = \rho_{p+r_{j'}}(c_{i'}) + w^{(c)}_{i'}\rho_p(r_{j'})+w^{(c)}_{i'}w^{(r)}_{j'}p
\]
with $c_{i'} \in \{0, \ldots, w^{(c)}_{i'}-1\}$ and $r_{j'}\in \{0, \ldots, w^{(r)}_{j'}-1\}$.
Additionally, the starting index of the fermionic mode in the subgrid $B_{i',j'}$ equals the combined sum of all preceding subgrids.
To formalize this, let $J$ denote the set of all index pairs $(i,j)$ satisfying $m^\mathrm{2D}_{S}(i,j)<m^\mathrm{2D}_{S}(i',j')$ and $|B_{i,j}|=Lw^{(c)}_{i}w^{(r)}_{j}$.
With this definition, the three-dimensional boustrophedon encoding follows a sequence of $S$-patterns through a set of ordered subgrids defined by the partition.

\begin{figure*}
    \centering
    \includegraphics{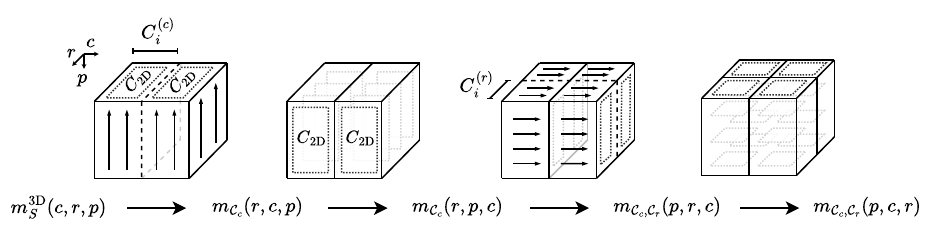}
    \caption{Schematic for switching from the 3D snake pattern $m^{\mathrm{3D}}(c,r,p)$ to the 3D boustrophedon encoding $m_{\mathcal{C}_c, \mathcal{C}_r}(p, c, r)$ as described in Lemma~\ref{lem:3DBousSwitch}. The arrows in the blocks indicate the direction of the CNOT ladders and the faces with dotted boundaries where $C_{\mathrm{2D}}$ is applied.}
    \label{fig:3D_boustrophedon}
\end{figure*}

Next, as in the two-dimensional case outlined in Corollary~\ref{cor:bous2D}, we extend our approach to three dimensions to enable switching between any two 3D boustrophedon encodings defined by different interval partitions.

\begin{lemma}\label{lem:3DBousSwitch}
    We can switch between three-dimensional boustrophedon encodings using $\mathcal{O}(N)$ gates in $\mathcal{O}(N^{\frac{1}{3}})$ depth.
\end{lemma}
\begin{proof}

    We now discuss how to switch from the $m^{\mathrm{3D}}_{S}(c,r,p)$ ordering to any three-dimensional boustrophedon ordering $m_{\mathcal{C}_r,\mathcal{C}_c}(c,r,p)$. 
    This transformation can then be applied twice (once forward and once in reverse), using $m^{\mathrm{3D}}_{S}(c,r,p)$ as an intermediate step, to switch between any two boustrophedon 3D encodings defined by different partitions.
    
    The first step transitions from $m^{\mathrm{3D}}_{S}(c,r,p)$ to $m_{\mathcal{C}_c}(r,p,c)$ for a partition of column indices into intervals $\mathcal{C}_c$. 
    Importantly, for any two points $(c, r, p)$ and $(c', r', p')$ belonging to different subgrids $B_{i}=C_i^{(c)}\times L \times L$ and $B_{i'}=C_{i'}^{(c)}\times L \times L$, respectively, we find that $m^{\mathrm{3D}}_{S}(c,r,p)<m^{\mathrm{3D}}_{S}(c',r',p')$ and $m_{\mathcal{C}_c}(c,r,p)<m_{\mathcal{C}_c}(c',r',p')$ for $i<i'$.
    By Lemma~\ref{lem:FermPerm},  this order-preservation implies that we can transform each subgrid independently, without inter-subgrid interactions. 
    To this end, we perform the dimensional hierarchy exchange outlined in Theorem~\ref{thm:3DimTransposition} for each subgrid $B_{i}$ in parallel, switching within each subgrid first from $m^{\mathrm{3D}}_{S}(c,r,p)$ to $m^{\mathrm{3D}}_{S}(r,c,p)$ (see the first schematic in Fig.~\ref{fig:3D_boustrophedon}). 
    Next, the plane and column hierarchies are swapped via parallel $C_\mathrm{2D}$ circuits to obtain $m^{\mathrm{3D}}_{S}(r,p,c)$ in each subgrid in accordance with Lemma~\ref{lem:3D-lower-dimension-exchange} (see the second schematic in Fig.~\ref{fig:3D_boustrophedon}), yielding $m_{\mathcal{C}_c}(r,p,c)$ on the total grid.
    
    With the rows now occupying the highest position in the dimensional hierarchy within each subgrid $B_{i}$, we can repeat the previous strategy for the row dimension.
    Therefore we partition the rows into the intervals $\mathcal{C}_r$, and use the approach from Theorem~\ref{thm:3DimTransposition} inside the newly defined subgrids $B_{i,j}=C_i^{(c)}\times C_j^{(r)} \times L$ (see the third and fourth schematic in Fig.~\ref{fig:3D_boustrophedon}). 
    Consequently, the ordering within each of these subgrids transforms to $m^{\mathrm{3D}}_{S}(p,c,r)$. 
    Finally, we pull the planes down in the dimensional hierarchy of each subgrid to realize the 3D boustrophedon encoding $m_{\mathcal{C}_c,\mathcal{C}_r}(c,r,p)$. 
    Note that partitioning the columns first and the rows second enforces an inner-subgrid order that first iterates over row partitions prior to column partitions.
\end{proof}

Having established an efficient method to switch between 3D boustrophedon encodings, we now demonstrate how to preserve the locality of fermionic interactions embedded on a 3D qubit lattice.

\begin{theorem}\label{thm:3Dlocality}
    It is possible to implement all pairwise fermionic interactions between geometrically local modes $\|x_i-x_j\|\le \delta$ on the nearest-neighbor connected qubit lattice $G_\mathrm{3D}$ using $\mathcal{O}(N)$ gates in $\mathcal{O}(N^{\frac{1}{3}})$ depth.
\end{theorem}
\begin{proof}

The proof is analogous to that of Theorem~\ref{thm:LocalImplementation}, wherein different boustrophedon encodings are specified on overlapping subgrids with side length $2\delta$ to ensure that all interactions are fully contained in at least one subgrid. 
To this end, we define eight sets of subgrids $\mathcal{G}_{b_{c}, b_{r}, b_{p}}$ in accordance with Definition~\ref{def:subgrid} and use the one-dimensional partition intervals of size $2\delta$
\[
\begin{aligned}
    \mathcal{G}_{0}&=\{I_0, I_{2\delta}, \ldots\},\\
    \mathcal{G}_{1}&=\{I_{-\delta}, I_{\delta}, \ldots\}
\end{aligned}
\]
to define four different 3D boustrophedon encodings as introduced in Definition~\ref{def:3D_boustrophedon}.
This enables the implementation of all pairwise interactions between geometrically local modes with the following procedure.
For each $b_{c}\in \{0,1\}$ and each $b_{r}\in \{0,1\}$:
\begin{enumerate}
    \item Switch encoding form $m^{\mathrm{3D}}(c,r,p)$ to $m_{\mathcal{G}_{b_{r}}, \mathcal{G}_{b_{c}}}(c,r,p)$ using Lemma~\ref{lem:3DBousSwitch}.
    \item Apply an FSN within each subgrid in $\mathcal{G}_{b_{c},b_{r},0}$ for which this encoding follows a Hamiltonian path.
    \item Apply an FSN within each shifted subgrid in $\mathcal{G}_{b_{c},b_{r},1}$, for which this encoding likewise follows a Hamiltonian path.
    \item Switch back to $m^{\mathrm{3D}}(c,r,p)$.
\end{enumerate}
By Lemma~\ref{lem:localPartition}, every interaction is covered in at least one of the subgrids.
Because, all steps apart from the encoding switches execute in constant depth, the overall asymptotic gate-count and depth scaling are governed by the cost of switching between 3D boustrophedon encodings via Lemma~\ref{lem:3DBousSwitch}.
\end{proof}

\subsubsection*{$d$-Dimensional boustrophedon encodings}
\label{app:d_dimensional_boustrophedon}
The methods introduced for 2D and 3D can be extended to local lattice models of arbitrary dimension number $d$.
To this end we assume a qubit grid of dimensions $d$ with lengths $L$ such that
\begin{equation}
    N = L \times L \times \ldots \times L = L^{d},
\end{equation} 
for which the $d$-dimensional $S$-pattern is given in Definition~\ref{def:dD_S_pattern}.
Similar to Definition~\ref{def:3D_boustrophedon}, for the three-dimensional boustrophedon encoding, we can define a $d$-dimensional boustrophedon encoding.

\begin{definition}[$d$D boustrophedon encoding]\label{def:dD_boustrophedon}
A $d$-dimensional boustrophedon encoding is specified by $d-1$ interval partitions

\[
\mathcal{C}_{x_{i}} = \{C_{1_{i}}, \ldots, C_{K_{i}}\}
\text{ for } x_{i}\in D/x_{1},
\]

where interval $C_{j_{i}}$ consists of $w^{(x_{j})}_{i}$ consecutive indices.
Each combination of intervals indexed by $\mathbf{j}=\{j_{2}, \ldots, j_{d}\}$ describes a subgrid
\[
B_{\mathbf{j}}= [L] \times C_{j_{2}} \times \ldots \times C_{j_{d}}.
\]
Among these subsets we choose an ordering denoted by $m^{d-1}_{S}(j_{2}, \ldots, j_{d})$.

Then the corresponding boustrophedon encoding is defined by
\[
\begin{gathered}
m_{\mathcal{C}_{x_{2}},\ldots, \mathcal{C}_{x_{d}}}(x_{1}, \ldots, x_{d})\\
=\\
m^{d}_{S}(x_{j_{1}}, x_{j_{2}} \ldots, x_{j_{d}})+|B_{1:m^{d-1}_{S}(\mathbf{j})}|,
\end{gathered}
\]
where $x_{j_{i}} = x_{i}-\sum_{k=1}^{j-1}|C_{k_{i}}|$ for $x_{j_{i}} \in C_{j_{i}}$.
With this definition, $m_{\mathcal{C}_{x_{2}},\ldots, \mathcal{C}_{x_{d}}}(x_{1}, \ldots, x_{d})$ follows a $d$-dimensional $S$-pattern ordering through each subgrid $B_{\mathbf{j}}$ and is therefore a Hamiltonian path for each subgrid.
\end{definition}

Next we show how to efficiently switch between $d$-dimensional boustrophedon encodings with different interval partitions, which is an extension of the three-dimensional method in Lemma~\ref{lem:3DBousSwitch}.

\begin{lemma}\label{lem:dDBousSwitch}
    It is possible to switch between $d$-dimensional boustrophedon encodings using $\mathcal{O}(N)$ gates in $\mathcal{O}(N^{\frac{1}{d}})$ depth.
\end{lemma}
\begin{proof}
    We now discuss how to switch from the ordering $m^{\mathrm{d}}_{S}(x_{1},\ldots, x_{d})$ to an arbitrary $d$-dimensional boustrophedon ordering $m_{\mathcal{C}_{x_{2}}, \ldots, \mathcal{C}_{x_{d}}}(x_{1},\ldots, x_{d})$. 
    This transformation can then be applied twice, using $m^{\mathrm{d}}_{S}(x_{1},\ldots, x_{d})$ as an intermediate step, to switch between any two $d$D-boustrophedon encodings defined by different partitions.

    The first step transitions from $m^{\mathrm{d}}_{S}(x_{1},\ldots, x_{d})$ to $m_{\mathcal{C}_{x_{d}}}(x_{d},x_{1},\ldots, x_{d-1})$ by partitioning the indices of the highest hierarchy-dimension into intervals $\mathcal{C}_{x_{d}}$. 
    Importantly, for any two points $(x_{1}, \ldots, x_{d})$ and $(x'_{1}, \ldots, x'_{d})$ in distinct subgrids $B_{j}= L \times \ldots \times L \times C_{j_{d}}$ and $B_{j'}= L \times \ldots \times L \times C_{j'_{d}}$ respectively, we have that $m^{\mathrm{d}}_{S}(x_{1}, \ldots, x_{d})<m^{\mathrm{d}}_{S}(x'_{1}, \ldots, x'_{d})$ and $m_{\mathcal{C}_{x_{d}}}(x_{d},x_{1},\ldots, x_{d-1})<m_{\mathcal{C}_{x_{d}}}(x'_{d},x'_{1},\ldots, x'_{d-1})$ for $j<j'$.
    By Lemma~\ref{lem:FermPerm}, this order preservation implies that each subgrid can be transformed independently, without requiring an inter-subgrid interactions.
    To this end, we perform the dimensional hierarchy change from Theorem~\ref{thm:dDimTransposition} for each subgrid $B_{j}$ in parallel, switching $m^{d}_{S}(x_{1}, \ldots, x_{d})$ to $m^{d}_{S}(x_{d}, x_{1}, \ldots, x_{d-1})$ locally within each subgrid.

    Consequently, $x_{d-1}$ is now the highest dimension in hierarchy, allowing us to repeat the previous strategy within each subgrid $B_{j}$. 
    That is, each subgrid is partitioned by $\mathcal{C}_{x_{d-1}}$ to create the new subgrids $B_{j_{d-1}, j_{d}}= L \times \ldots \times L \times C_{j_{d-1}} \times C_{j_{d}}$.
    Within each new subgrid, dimension $d-1$ is shifted to the lowest position in the hierarchy, yielding the encoding $m_{\mathcal{C}_{x_{d-1}}, \mathcal{C}_{x_{d}}}(x_{d-1},x_{d}, x_{1}\ldots, x_{d-2})$.

    Repeating this process sequentially for dimensions $d-2$ down to $1$ results in the target encoding $m_{\mathcal{C}_{x_{2}}, \ldots, \mathcal{C}_{x_{d}}}(x_{1}, \ldots, x_{d-2})$. 
    Note that it is not necessary to partition the first dimension, as long as it ends up being lowest in the dimensional hierarchy.
    Importantly, this method applies Theorem~\ref{thm:dDimTransposition} $d$ times, resulting in $\mathcal{O}(dN)$ gates in $\mathcal{O}(dN^{\frac{1}{d}})$ depth.
\end{proof}

At last we demonstrate that the use of different $d$-dimensional boustrophedon encodings allow us to preserve the locality of fermionic interactions embedded on a $d$-dimensional qubit lattice.

\begin{theorem}\label{thm:dDlocality}
    It is possible to implement all pairwise fermionic interactions between geometrically local modes $\|x_i-x_j\|\le \delta$ on the nearest-neighbor connected qubit lattice $G_\mathrm{dD}$ using $\mathcal{O}(N)$ gates in $\mathcal{O}(N^{\frac{1}{d}})$ depth.
\end{theorem}
\begin{proof}
The proof is an extension of the 3D case in Theorem~\ref{thm:3Dlocality}.
We define $2^{d}$ sets of subgrids $\mathcal{G}_{b_{1}, \dots, b_{d}}$ along Definition~\ref{def:subgrid},
using the one-dimensional partition intervals of size $2\delta$
\[
\begin{aligned}
    \mathcal{G}_{0}&=\{I_0, I_{2\delta}, \ldots\},\\
    \mathcal{G}_{1}&=\{I_{-\delta}, I_{\delta}, \ldots\}
\end{aligned}
\]
to define $2^{d-1}$ different boustrophedon encodings according to Definition~\ref{def:dD_boustrophedon}.

This allows us to implement all pairwise interactions between geometrically local modes with the following procedure.\\
For each $b_{j}\in \{0,1\}$ except $b_{1}$:
\begin{enumerate}
    \item Switch encoding form $m_{S}^{\mathrm{d}}(x_{1}, \dots, x_{d})$ to $m_{\mathcal{G}_{b_{2}}, \ldots, \mathcal{G}_{b_{d}}}(x_{1}, \dots, x_{d})$ using Lemma~\ref{lem:dDBousSwitch}.
    \item Apply an FSN within each subgrid in $\mathcal{G}_{0,b_{2},\ldots, b_{d}}$ for which this encoding follows a Hamiltonian path.
    \item Apply an FSN within each shifted subgrid $\mathcal{G}_{1,b_{2},\ldots, b_{d}}$ for which this encoding likewise follows a Hamiltonian path.
    \item Switch back to $m_{S}^{\mathrm{d}}(x_{1}, \dots, x_{d})$.
\end{enumerate}
By Lemma~\ref{lem:localPartition}, every interaction is covered in at least one of the subgrids.
Because, all steps except the encoding switches can be implemented in constant depth, the asymptotic gate-count and depth scaling are governed by the cost of switching between the three-dimensional boustrophedon encodings via Lemma~\ref{lem:dDBousSwitch}.
\end{proof}

\section{Fermion routing}\label{app:fermion_routing}

\begin{figure}
    \centering
    \includegraphics[width = \columnwidth]{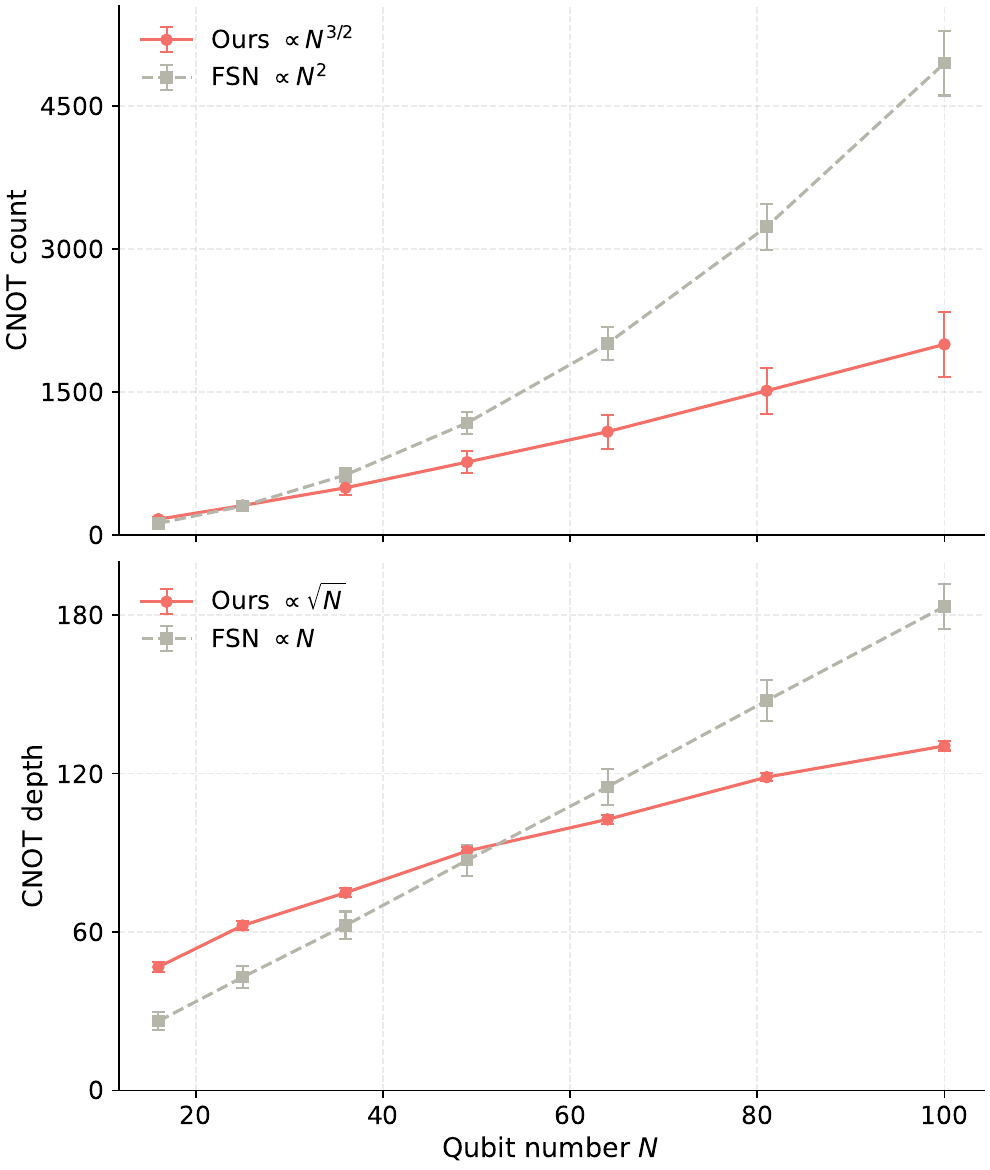}
    \caption{
    Permutation-routing of fermions via FSNs versus our dynamic encoding method on square qubit lattices of size $N=L\times L$. Shown are the CNOT gate count (top) and CNOT depth (bottom). The FSN baseline employs odd-even sort and a geometrically static JW path, whereas our method leverages two-dimensional permutation algorithm for square lattices \cite{annexsteinUnifiedApproachOffline1990}, enabled by dynamic encoding switches. 
    Data points depict the mean and standard deviation over $1000$ uniformly random permutations. For lattices exceeding $5\times5$, our method reduces the CNOT count and for lattices larger than $7\times7$, it simultaneously reduces the circuit depth.
    }
    \label{fig:PermutationImplementation}
\end{figure}

In preceding sections, we considered permutations of the canonical ordering. 
In that context, changing the assignment of physical lattice sites to canonical ranks  modified the physical layout of the JW path, while leaving the essential quantum information stored at each physical site unchanged.

However, supporting arbitrary fermionic computations on qubit lattices also requires the ability to physically route the quantum information stored at each lattice site, thereby enabling fermionic interactions between physically distant sites.
Here, we therefore permute the fermionic information among the physical qubits while keeping the canonical ordering, and consequently the physical JW path, fixed.

In this section we consider permutations of the lattice-site set
\[
G=[L]^d,
\qquad
[L]:=\{0,\ldots,L-1\}.
\]
More precisely, a routing task is defined via a permutation, i.e., a bijection
\[
\pi\colon G\to G,
\]
where $\pi(\mathbf v)$ denotes the destination of the quantum information initially stored at site $\mathbf v$.

Fermionic swap networks provide a natural mechanism for implementing such fermion routing permutations. 
In this setting, FSWAP gates alter the assignment of fermionic modes to physical qubitsvia
\[
f_{\mathrm{swap}}^{p,q}\,
c_p^\dagger\,
\bigl(f_{\mathrm{swap}}^{p,q}\bigr)^\dagger
=
c_q^\dagger,
\qquad
f_{\mathrm{swap}}^{p,q}\,
c_p\,
\bigl(f_{\mathrm{swap}}^{p,q}\bigr)^\dagger
=
c_q,
\]
as discussed, for example, in
Ref.~\cite{kivlichanQuantumSimulationElectronic2018}.
The canonical ranks assigned to the physical qubits, however, remain unchanged under the application of these gates.
As discussed in the main text, FSWAPs operate strictly along the canonical ordering, making them intrinsically one-dimensional. 

In this section, we combine fermionic swap networks with our dynamic-encoding primitives to construct fermion routing permutations for $d$-dimensional lattices.
A central ingredient in our construction is a result of
Annexstein et al.~\cite{annexsteinUnifiedApproachOffline1990}, which derives a routing schedule on the product of routing graphs $A\times B$ from routing algorithms on the individual factor graphs $A$ and $B$.
In our context, a \emph{routing graph} $A=(V,E)$ specifies the allowed qubit information exchanges: vertices correspond to qubit locations and edges to pairs of locations between which a SWAP can be applied.
A routing step consists of parallel SWAPs along a matching in $A$, and an arbitrary permutation is routed by a sequence of such steps.
As we will primarily consider nearest-neighbor routing graphs, we often specify only the vertex set $V$ as points in $\mathbb{Z}_{\geq 0}^d$, with the corresponding nearest-neighbor edges understood implicitly, typically points which are unit-distance appart.

In the context of routing, we state the following crucial result.

\begin{theorem}[Product-routing bound, adapted from Annexstein and Baumslag~{\cite[Theorem~1]{annexsteinUnifiedApproachOffline1990}}]
\label{thm:productRouting}
Let $A$ and $B$ be graphs equipped with permutation-routing algorithms whose worst-case routing depths are $T(A)$ and $T(B)$, respectively. 
Then the product graph $A \times B$ admits a permutation-routing algorithm with worst-case routing depth at most
\[
    T(A) + T(B) + \min\{T(A),T(B)\}.
\]
\end{theorem}

\subsubsection*{2D lattices}
Our goal is to construct a fermionic implementation of an arbitrary routing permutation
\[
\pi\colon G\to G.
\]
In this section we focus on the two-dimensional lattice, defined as the Cartesian product of two one-dimensional arrays,
\[
G=R\times C,
\qquad
R=C=\{0,\ldots,L-1\}.
\]
Here, $R$ and $C$ label the row and column coordinates,
respectively.
Our construction proceeds in two steps. 
First, we decompose the routing permutation on the product graph (2D grid graph) into permutations acting on its one-dimensional factor graphs (path graphs). 
Next, we demonstrate that these factor-wise permutations can be implemented efficiently using fermionic swap networks in combination with our dynamic-encoding primitives.

The routing schedule underlying Theorem~\ref{thm:productRouting}, as constructed in Ref.~\cite{annexsteinUnifiedApproachOffline1990},
decomposes an arbitrary routing permutation $\pi$ into
\[
\pi
=
\pi_R' \circ \pi_C \circ \pi_R.
\]
Here, $\pi_R$ and $\pi_R'$ act independently within each row, whereas $\pi_C$ acts independently within each column. More precisely, there exist permutations
\[
\rho_i,\rho_i',\kappa_i\colon [L]\to[L],
\qquad i,j\in[L],
\]
such that
\[
\pi_R(i,j)=(i,\rho_i(j)),
\qquad
\pi_R'(i,j)=(i,\rho_i'(j)),
\]
and
\[
\pi_C(i,j)=(\kappa_j(i),j).
\]
Thus, the three stages consist of a parallel routing operation along all rows, followed by one along all columns, and finally another along all rows.

Each of the constituent routing permutations
$\pi_R$, $\pi_R'$, and $\pi_C$ can be implemented using
odd-even transposition sort. 
For $\pi_R$ and $\pi_R'$, the sorting algorithm is applied independently and in parallel across every row, where each row $i$ realizes its corresponding 1D permutation $\rho_i$ or $\rho_i'$. 
Likewise, for $\pi_C$, the algorithm operates independently and in parallel across every column $j$ to realize  $\kappa_j$.

Since odd-even transposition sort relies exclusively on nearest-neighbor swaps, i.e., between sites satisfying
\[
\lVert \mathbf v_1-\mathbf v_2\rVert_1=1,
\]
it can be executed natively on the respective nearest-neighbor connected paths.
Moreover, because all rows (or all columns) are vertex-disjoint, they can be processed simultaneously. 
The worst-case depth is therefore linear in the path length $L$, yielding
\[
T(R)=T(C)=\mathcal{O}(L).
\]

For a fermionic implementation of this routing strategy, we replace nearest-neighbor SWAP gates with nearest-neighbor FSWAP gates. 
Here, however, adjacency must hold in two distinct senses:
the participating qubits must be physicaly adjecent on the hardware, and their corresponding sites must be adjacent in the canonical ordering, i.e.,
\[
\bigl|m(\mathbf v_1)-m(\mathbf v_2)\bigr|=1.
\]
Whenever both conditions are satisfied, the corresponding two-qubit SWAP gate translates directly into a two-qubit FSWAP gate.

We consider an initial canonical ordering defined by a  two-dimensional $S$-pattern
\begin{equation*}
    m_S^{\mathrm{2D}}(r, c) =  \rho_r (c) + Lr.
\end{equation*}
Since the JW path follows a Hamiltonian path along all rows, the nearest-neighbour SWAP gates in $\pi_R$ immediately translate to FSWAP gates.

Afterwards, using Theorem~\ref{thm:C2D}, we switch to a $Z$ pattern encoding
\begin{equation*}
    m_{Z}^{\mathrm{2D}}(r,c) = \rho_{1}(r)+Lc,
\end{equation*}
which enables translating the SWAP operations in $\pi_C$ to FSWAPS.
Because this transformation costs only $\mathcal{O}(L)$ depth, it does not increase the asymptotic depth scaling of the routing permutation.
Finally, we revert to the $m_S^{\mathrm{2D}}(r, c)$ ordering and execute the last routing step.

Remarkably, the practical advantages of this approach immediately extend beyond asymptotic scaling: 
as shown in Fig.~\ref{fig:PermutationImplementation}, this method outperforms the corresponding one-dimensional FSN-based approach even at modest system sizes. 
Note that we omit a comparison with Ref.~\cite{constantinidesLowdepthFermionRouting2025}, as their construction establishes only an $\mathcal{O}\left(L\log^2(L)\right)$  depth-bound for qubit lattices and has not yet been fully optimized for this setting.
Since the publication of the first version of this manuscript, Li et al.~\cite{li2026asymptoticallyoptimaldepthfermionic} have independently developed a different 2D routing protocol that reaches the same asymptotic scaling, but with different prefactors.
\subsubsection*{Higher dimensional lattices}

Consider the family of lattices
\[
G_k\coloneqq [L]^k,
\qquad
k\in\{1,\ldots,d\},
\]
so that
\[
G_k=G_{k-1}\times[L]
\qquad
\text{for }k\geq 2.
\]

We first introduce the notion of a fiberwise permutation.

\begin{definition}
A permutation $\Phi_j\colon G_d\to G_d$ is said to act fiberwise along the $j$-th coordinate if it leaves all coordinates other than the $j$-th invariant.
More precisely, letting
\[
\mathbf v_{\setminus j}
\coloneqq
(v_1,\ldots,v_{j-1},v_{j+1},\ldots,v_d),
\]
there exists a family of 1D permutations
\[
\sigma_{\mathbf v_{\setminus j}}\colon[L]\to[L],
\qquad
\mathbf v_{\setminus j}\in[L]^{d-1},
\]
such that
\[
\Phi_j(v_1,\ldots,v_d)
=
\bigl(
v_1,\ldots,v_{j-1},
\sigma_{\mathbf v_{\setminus j}}(v_j),
v_{j+1},\ldots,v_d
\bigr).
\]
Thus, $\Phi_j$ independently permutes the sites along every coordinate line parallel to the $j$-th axis.    
\end{definition}

For completeness, we start by presenting an adaptation of the construction by Annexstein et al.~\cite{annexsteinUnifiedApproachOffline1990}, which decomposes an arbitrary permutation into a sequence of fiberwise permutations.
Subsequently, as in the two-dimensional case, we show that each of these fiberwise permutations can be implemented natively using dimensional hierarchy changes and FSNs.

\begin{lemma}
\label{lem:routingDecomposition}
Let $\pi\colon G_d\to G_d$
be an arbitrary routing permutation. There exist fiberwise permutations
\[
\Phi_j,\Phi_j'\colon G_d\to G_d,
\qquad
j=2,\ldots,d,
\]
together with a central fiberwise permutation $\Phi_1\colon G_d\to G_d$, 
such that
\[
\pi
=
\Phi_d'
\circ
\Phi_{d-1}'
\circ\cdots\circ
\Phi_2'
\circ
\Phi_1
\circ
\Phi_2
\circ\cdots\circ
\Phi_{d-1}
\circ
\Phi_d.
\]
Thus, an arbitrary routing permutation on $G_d$ can be decomposed into $2d-1$ fiberwise routing stages.
\end{lemma}
\begin{proof}
Applying Theorem~\ref{thm:productRouting} to the decomposition
\[
G_d=G_{d-1}\times[L],
\]
we obtain permutations $\Phi_d$, $\Phi_d'$, and $\Psi_{d-1}$ such that
\[
\pi 
=
\Phi_d'
\circ
\Psi_{d-1}
\circ
\Phi_d.
\]
The permutations $\Phi_d$ and $\Phi_d'$ act fiberwise along the $d$-th coordinate. 
In particular, writing a vertex of $G_d$ as $(\mathbf u,t)$, with
\[
\mathbf u\in G_{d-1},
\qquad
t\in[L],
\]
there exist permutations
\[
\sigma_{\mathbf u},\sigma_{\mathbf u}'\colon[L]\to[L]
\]
such that
\[
\Phi_d(\mathbf u,t)
=
\bigl(\mathbf u,\sigma_{\mathbf u}(t)\bigr),
\qquad
\Phi_d'(\mathbf u,t)
=
\bigl(\mathbf u,\sigma_{\mathbf u}'(t)\bigr).
\]

The middle stage $\Psi_{d-1}$ preserves the $d$-th coordinate and acts independently within each $(d-1)$-dimensional slice.
Hence, there exists a family of routing permutations
\[
\tau_t\colon G_{d-1}\to G_{d-1},
\qquad
t\in[L],
\]
such that
\[
\Psi_{d-1}(\mathbf u,t)
=
\bigl(\tau_t(\mathbf u),t\bigr).
\]
The $L$ permutations $\tau_t$ act on disjoint slices and may
therefore be implemented in parallel.

We now apply the same product-graph decomposition recursively to each slice permutation $\tau_t$.
For example, applying Theorem~\ref{thm:productRouting} to
\[
G_{d-1}=G_{d-2}\times[L]
\]
gives, for every $t\in[L]$,
\[
\tau_t
=
\phi_{d-1,t}'
\circ
\psi_{d-2,t}
\circ
\phi_{d-1,t},
\]
where $\phi_{d-1,t}$ and $\phi_{d-1,t}'$ act fiberwise along the $(d-1)$-st coordinate within the slice $G_{d-1}\times\{t\}$, while $\psi_{d-2,t}$ acts independently within its $(d-2)$-dimensional subslices.

Although the permutations $\phi_{d-1,t}$ may depend on $t$, their parallel composition defines a single permutation
\[
\Phi_{d-1}\colon G_d\to G_d
\]
by
\[
\Phi_{d-1}(\mathbf u,t)
\coloneqq
\bigl(\phi_{d-1,t}(\mathbf u),t\bigr).
\]
Likewise, we define
\[
\Phi_{d-1}'(\mathbf u,t)
\coloneqq
\bigl(\phi_{d-1,t}'(\mathbf u),t\bigr)
\]
and
\[
\Psi_{d-2}(\mathbf u,t)
\coloneqq
\bigl(\psi_{d-2,t}(\mathbf u),t\bigr).
\]
Because the slices $G_{d-1}\times\{t\}$ are pairwise disjoint, the corresponding slice permutations have disjoint support and may be executed simultaneously.
Moreover, $\Phi_{d-1}$ and $\Phi_{d-1}'$ preserve every coordinate except the $(d-1)$-st one.
They therefore constitute fiberwise routing stages along the
$(d-1)$-st coordinate on the full lattice $G_d$.
The middle stage consequently admits the decomposition
\[
\Psi_{d-1}
=
\Phi_{d-1}'
\circ
\Psi_{d-2}
\circ
\Phi_{d-1}.
\]

Repeating this argument down to $G_1=[L]$ yields fiberwise permutations
\[
\Phi_j,\Phi_j'\colon G_d\to G_d,
\qquad
j=2,\ldots,d,
\]
together with a central fiberwise permutation
\[
\Phi_1\colon G_d\to G_d,
\]
such that
\[
\pi
=
\Phi_d'
\circ
\Phi_{d-1}'
\circ\cdots\circ
\Phi_2'
\circ
\Phi_1
\circ
\Phi_2
\circ\cdots\circ
\Phi_{d-1}
\circ
\Phi_d.
\]
Here, each $\Phi_j$ and $\Phi_j'$ acts independently along all coordinate lines parallel to the $j$-th axis. 
\end{proof}

We can now prove the main result of this section.

\begin{theorem}[$d$-dimensional fermion routing]
\label{thm:dDimensionalFermionicRouting}
Let $G=[L]^d$ be a $d$-dimensional nearest-neighbor-connected qubit lattice. Any routing permutation
\[
\pi\colon G\to G
\]
admits a fermionic implementation with circuit depth
\[
\mathcal{O}\left(d^2N^{1/d}\right).
\]
\end{theorem}

\begin{proof}
To obtain a fermionic implementation of $\pi$, we decompose it into nearest-neighbor FSWAP gates. 
Here, nearest-neighbor adjacency is required both with respect to the physical lattice connectivity and with respect to the canonical JW ordering.

By Lemma~\ref{lem:routingDecomposition}, the routing permutation admits a decomposition into $2d-1$ fiberwise routing stages,
\[
\pi
=
\Phi_d'
\circ
\Phi_{d-1}'
\circ\cdots\circ
\Phi_2'
\circ
\Phi_1
\circ
\Phi_2
\circ\cdots\circ
\Phi_{d-1}
\circ
\Phi_d.
\]
We implement each such stage using odd-even transposition sort on every corresponding coordinate line.
Since these lines are pairwise disjoint, all sorting procedures within a given stage can be executed in parallel. 
Each line contains $L$ sites, so each stage results in a nearest-neighbor SWAP circuit with depth
\[
\mathcal{O}(L)
=
\mathcal{O}\left(N^{1/d}\right).
\]

We initialize the JW encoding using a $d$-dimensional $S$-pattern with an arbitrary dimensional hierarchy
$\boldsymbol{\sigma}$. 
Before executing a fiberwise routing stage along dimension $j$, we use the dynamic-encoding transformation of
Theorem~\ref{thm:dDimTransposition} to move that dimension to the lowest level of the hierarchy. 
This re-encoding requires depth
\[
\mathcal{O}\left(dN^{1/d}\right).
\]
In the resulting encoding, every coordinate line parallel to the $j$-th axis forms a contiguous segment of the JW path.
Consequently, each nearest-neighbor SWAP in the odd-even transposition network acts on qubits that are adjacent both physically and in the canonical ordering, and may therefore be replaced directly by an
FSWAP gate.

Thus, each fiberwise stage, including the required change of encoding, has depth
\[
\mathcal{O}(N^{1/d})+\mathcal{O}(dN^{1/d})
=
\mathcal{O}(dN^{1/d}).
\]
Since the decomposition contains $2d-1=\mathcal{O}(d)$ such stages, the total circuit depth is
\[
T(\pi)
=
\mathcal{O}\left(d^2N^{1/d}\right).
\]
\end{proof}

\subsubsection*{All-to-All Connected Systems}
The best previously established asymptotic depth for ancilla-free fermion routing on all-to-all-connected qubit architectures is
\[
    \mathcal{O}\!\left(\log^2 N\right),
\]
as shown in Ref.~\cite{constantinidesLowdepthFermionRouting2025}.
This scaling, however, remains asymptotically worse than that achieved by ancilla-assisted fermion routing, for which a depth of
\[
    \mathcal{O}\!\left(\log N\right)
\]
is known~\cite{maskaraFastSimulationFermions2025}. 
Our main result in the all-to-all setting is an improvement in the asymptotic depth for ancilla-free fermion routing.
\begin{corollary}
\label{cor:RecursiveRouting}
For every fixed $a\in\mathbb{N}_{>0}$, there exists a fermion-routing algorithm that implements any permutation $\pi:[N]\to[N]$ in circuit depth
\[
    \mathcal O\left(\log^{1+\frac{1}{a}} N\right).
\]
\end{corollary}
We defer the proof of Corollary~\ref{cor:RecursiveRouting} until later.
Before turning to the technical details, however, we first discuss the motivation behind the construction and the significance of this result.

The main purpose of the construction is to show that there is no fixed exponent $m>1$ that constitutes an asymptotic barrier for ancilla free fermion routing: for any such $m$, one can obtain a routing algorithm with circuit depth $\mathcal O(\log^m N)$ by choosing $a$ sufficiently large.
%
Its immediate practical relevance is more limited, however. In particular, the value of $a$ that optimizes the resulting depth bound grows only as
\[
    a_{\mathrm{optimal}}=\mathcal O(\log\log N),
\]
so that very large system sizes are required before algorithms with large values of $a$ become advantageous. 
%
We next develop the ingredients needed to prove this statement.

In our construction we use that on an all-to-all-connected architecture the physical qubits may be assigned a virtual $d$-dimensional lattice structure. 
Since the hardware connectivity imposes no geometric locality constraints, the effective dimension $d$ may be chosen as a function of the system size $N$.
For a uniform lattice $G=[L]^d$, this requires
\[
N=L^d.
\]

Imposing this condition exactly, however, would restrict the
construction to system sizes that are perfect $d$-th powers. 
We therefore begin by generalizing the permutation-routing framework to bounded vertex occupancies. 
This generalization allows the routing construction to accommodate a wider range of system sizes while preserving its asymptotic depth.

\begin{lemma}
\label{lem:twoSlotRouting}
Let $H=(V,E)$ be a routing graph, and associate with every vertex $v\in V$ a primary slot $(v,0)$ and a secondary slot $(v,1)$.
Assume that every primary slot is occupied both in the initial and in the target configuration, while secondary slots may be occupied or empty.

Suppose that every packet is assigned a unique target slot.
Then the resulting slot permutation can be implemented using at most three permutation-routing problems on $G$, together with one layer of local swaps between the primary and secondary slots of individual vertices.

In particular, if every permutation on $G$ can be routed in depth $T(G)$, and a local swap between the two slots of a vertex can be performed in constant depth, then the complete slot permutation can be implemented in depth
\[
3T(H)+\mathcal O(1)
=
\mathcal O\left(T(H)\right).
\]
\end{lemma}

\begin{proof}
We classify the packets according to the layer of their initial and target slots. Let
\[
P\to S
\]
denote the packets whose initial slot is primary and whose target slot is secondary, and let
\[
S\to P
\]
denote those whose initial slot is secondary and whose target slot is primary.

Since every primary slot is occupied both initially and in the target configuration, the primary layer contains exactly $|V|$ packets at both times. 
Consequently, the number of packets leaving the primary layer equals the number entering it, and hence
\[
|P\to S|=|S\to P|.
\]
We therefore choose an arbitrary bijection between the packets in $P\to S$ and those in $S\to P$.

Consider a paired set of packets $p\in P\to S$ and $q\in S\to P$, and let $v_p\in V$ be the vertex whose primary slot initially contains $p$. 
We first permute the secondary slots so that $q$ is
moved to the secondary slot $(v_p,1)$. 
Since distinct packets in $P\to S$ initially occupy distinct primary slots, these prescribed secondary slots are pairwise distinct. 
The remaining occupied secondary slots and vacancies may be assigned arbitrarily to the remaining secondary positions, thereby extending this assignment to a permutation of the entire secondary layer. 
This gives the first permutation-routing problem on $G$.

After this permutation, every packet in $S\to P$ is located directly below its paired packet in $P\to S$. 
We now simultaneously swap the primary and secondary slots at all such vertices. 
After these local swaps, every packet occupies the same layer as its prescribed target slot.

It remains only to route packets within their respective layers.
Because every primary target slot is occupied exactly once, the
current primary packets and their prescribed primary target slots define a permutation of $V$. 
Hence, they can be routed by one permutation-routing problem on $G$.

Similarly, the packets in the secondary layer have pairwise distinct prescribed secondary target slots. 
By treating empty secondary slots as vacancies and assigning the vacancies to the remaining target positions, this routing can be completed to a permutation of $V$.
It therefore constitutes one further permutation-routing problem on $G$.

Thus, the original slot permutation is realized by a preliminary secondary-layer permutation, a layer of local primary-secondary swaps, and independent final permutations of the primary and secondary layers. 
In total, this requires at most three permutation-routing problems on $G$.
\end{proof}

We will now discuss how such a two-slot embedding approach can be utilized in our setting.

\begin{definition}[Condensed lattice]
\label{def:condensedLattice}
Let $L\geq 2$ and $d\geq 1$ be real numbers such that
\[
N\coloneqq L^d
\]
is an integer. 
Let
\[
\ell\coloneqq\lfloor L\rfloor,
\qquad
k\coloneqq\lfloor d\rfloor,
\qquad
b\coloneqq\left\lfloor \frac{N}{\ell^k} \right\rfloor.
\]
We then define the condensed lattice via the vertice set
\[
\widetilde G\coloneqq[\ell]^k\times[b].
\]
\end{definition}
The condensed lattice satisfies
\[
    0 \leq N-|\widetilde G| < \ell^k,
\]
so it may contain fewer than $N$ vertices, but still satisfies
\[
    |\widetilde G| > \frac{N}{2}.
\]
Assigning two slots to each vertex therefore is sufficient to accommodate all $N$ qubits.

The condensed-lattice construction allows the characteristic side length $L$ to scale more flexibly with the system size.
Accordingly, let
\[
L:\mathbb N\to\mathbb R_{>2},
\qquad N\mapsto L(N),
\]
and define the corresponding effective dimension by
\[
    d(N)\coloneqq \frac{\log N}{\log L(N)}.
\]
By construction,
\[
    N=L(N)^{d(N)}.
\]
To ensure that all side lengths of the condensed lattice remain
$\mathcal O(L)$, as required for the desired permutation-depth scaling, we impose an additional growth condition on $L(N)$.

\begin{lemma}
\label{lem:condensedLastDimension}
Consider the condensed lattice from Definition~\ref{def:condensedLattice}:
\[
\widetilde G\coloneqq[\ell]^k\times[b].
\]
Let the characteristic side length $L(N)>2$ be a function of the system size.
If
\[
d(N)=\mathcal O\left(L(N)\right)
\]
as $L(N)\to\infty$, then
\[
b=b(N)=\mathcal O(L(N)).
\]
\end{lemma}

\begin{proof}
Writing
\[
\alpha\coloneqq d-k\in[0,1),
\]
we have
\[
\frac{N}{\ell^k}
=
\frac{L^{k+\alpha}}{\ell^k}
=
L^\alpha
\left(\frac{L}{\ell}\right)^k.
\]
Since $\ell=\lfloor L\rfloor$, we have $0\leq L-\ell<1$, and hence
\[
\frac{L}{\ell}
=
1+\frac{L-\ell}{\ell}
\leq
1+\frac{1}{\ell}.
\]
Using $1+x\leq e^x$, we obtain
\[
\left(\frac{L}{\ell}\right)^k
\leq
\left(1+\frac{1}{\ell}\right)^k
\leq
\exp\left(\frac{k}{\ell}\right).
\]
Moreover,
$k\leq d=\mathcal O(L)$ and $\ell > L/2$, so
\[
\frac{k}{\ell}=\mathcal O(1).
\]
Consequently,
\[
\left(\frac{L}{\ell}\right)^k
=
\mathcal O(1).
\]
Since $\alpha\in[0,1)$, we also have
\[
L^\alpha\leq L.
\]
It follows that
\[
\frac{N}{\ell^k}
=
\mathcal O(L).
\]
Finally,
\[
b
=
\left\lfloor\frac{N}{\ell^k}\right\rfloor
\leq
\frac{N}{\ell^k},
\]
and therefore
\[
b=\mathcal O(L).
\]
\end{proof}
We proceed by making the two-slot embedding on the condensed lattice rigorous.
\begin{definition}[Two-slot embedding on the condensed lattice]
\label{def:condensedLatticeFermions}
Let $\widetilde G$ be the condensed lattice from
Definition~\ref{def:condensedLattice}, and let
\[
n\coloneqq |\widetilde G|.
\]

Each vertex $\mathbf v\in\widetilde G$ carries two virtual slots,
\[
\mathbf v^{(0)}\coloneqq(\mathbf v,0),
\qquad
\mathbf v^{(1)}\coloneqq(\mathbf v,1),
\]
called its primary and secondary slot, respectively. 
Every primary slot is occupied, whereas a secondary slot may be either occupied or empty.
\end{definition}
As in our previous lattice constructions, each lattice site is associated with a specific rank in the JW ordering.
We then populate the virtual lattice by assigning qubits to the sites according to this ranking.
\begin{definition}[JW ordering on the condensed lattice]
Let
\[
B\subseteq\widetilde G
\]
denote the set of vertices whose secondary slot is occupied. The set of occupied slots is therefore
\[
\mathcal S_B
\coloneqq
\bigl(\widetilde G\times\{0\}\bigr)
\cup
\bigl(B\times\{1\}\bigr).
\]
In particular,
\[
|\mathcal S_B|=n+|B|=N.
\]
We equip the vertices of $\widetilde G$ with an $S$-pattern ordering
\begin{gather*}
  \mu\colon\widetilde G\to[n],\\[0.2em]  
  \mu(\mathbf v) = m_S^k(\mathbf v).
\end{gather*}
For \(\mathbf v\in\widetilde G\), define
\[
p_B(\mathbf v)
\coloneqq
\left|
\left\{
\mathbf u\in B:
\mu(\mathbf u)<\mu(\mathbf v)
\right\}
\right|,
\]
which counts the occupied secondary slots associated with vertices preceding $\mathbf v$ in the $S$-pattern.

The induced canonical JW ordering of the occupied slots is the bijection
\[
\widetilde m_B\colon\mathcal S_B\to[N]
\]
defined by
\[
\widetilde m_B(\mathbf v,0)
\coloneqq
\mu(\mathbf v)+p_B(\mathbf v)
\]
and, whenever \(\mathbf v\in B\),
\[
\widetilde m_B(\mathbf v,1)
\coloneqq
\mu(\mathbf v)+p_B(\mathbf v)+1.
\]
\end{definition}
An immediate consequence of this definition is that every occupied secondary slot immediately follows the primary slot of the same vertex in the JW ordering:
\[
    \widetilde m_B(\mathbf v,1)
    =
    \widetilde m_B(\mathbf v,0)+1.
\]

Consequently, two vertices that are adjacent in the $S$-pattern ordering, i.e., $\mathbf v,\mathbf v'\in\widetilde G$ satisfying
\[
    \lvert \mu(\mathbf v)-\mu(\mathbf v')\rvert = 1,
\]
together contain at most four qubits forming a contiguous block in the JW ordering.
Any reordering of the qubits within such a vertex pair can therefore be implemented using nearest-neighbor FSWAPs.

We can now formalize the building blocks of later general permutations.
Let
\begin{align*}
\mathbf v_{\mathrm{high}}
&=
\operatorname*{arg\,max}_{\mathbf u\in\{\mathbf v,\mathbf v'\}}
\mu(\mathbf u),\\
\mathbf v_{\mathrm{low}}
&=
\operatorname*{arg\,min}_{\mathbf u\in\{\mathbf v,\mathbf v'\}}
\mu(\mathbf u).
\end{align*}

\begin{definition}[Primary-primary swap]
For a \emph{primary-primary swap} between
$(\mathbf v_{\mathrm{high}},0)$ and
$(\mathbf v_{\mathrm{low}},0)$, we distinguish whether the secondary slot of $\mathbf v_{\mathrm{high}}$ is occupied.

If $\mathbf v_{\mathrm{high}}\in B$, three relevant slots are consecutive in the JW ordering:
\[
    m\left((\mathbf v_{\mathrm{high}},0)\right)
    =
    m\left((\mathbf v_{\mathrm{high}},1)\right)+1
    =
    m\left((\mathbf v_{\mathrm{low}},0)\right)+2.
\]
Writing the qubits occupying these slots as $(A,B,C)$, respectively, the desired primary-primary swap acts as
\[
    (A,B,C)\longmapsto(C,B,A),
\]
i.e., it exchanges the two primary qubits while leaving the secondary qubit unchanged.
This transformation can be realized by $\mathcal O(1)$ FSWAPs.

If $\mathbf v_{\mathrm{high}}\notin B$, the two primary slots are consecutive in the JW ordering, so a single FSWAP between $(\mathbf v_{\mathrm{high}},0)$ and $(\mathbf v_{\mathrm{low}},0)$ suffices.
\end{definition}

\begin{definition}[Secondary-secondary swap]
For a \emph{secondary-secondary swap} between $(\mathbf v_{\mathrm{high}},1)$ and $(\mathbf v_{\mathrm{low}},1)$, we distinguish whether both secondary slots are occupied.
If both are occupied, the two secondary slots are separated by the primary slot $(\mathbf v_{\mathrm{low}},0)$, and three relevant slots are consecutive in the JW ordering:
\[
    m\left((\mathbf v_{\mathrm{high}},1)\right)
    =
    m\left((\mathbf v_{\mathrm{low}},0)\right)+1
    =
    m\left((\mathbf v_{\mathrm{low}},1)\right)+2.
\]
Denoting the qubits occupying these three slots by $(A,B,C)$,
respectively, the desired secondary-secondary swap again corresponds to
\[
    (A,B,C)\longmapsto(C,B,A),
\]
i.e., the two secondary qubits are exchanged while the intermediate primary qubit remains unchanged.
As in the primary-primary case, this transformation can be implemented using $\mathcal O(1)$ nearest-neighbor FSWAPs.

If exactly one of the two secondary slots is occupied, it suffices to move the corresponding secondary qubit across
$(\mathbf v_{\mathrm{low}},0)$ using a single FSWAP.
We then update $B$ accordingly to record the new location of the occupied secondary slot.
\end{definition}
Thus, each nearest-neighbor vertex swap can be realized using
$\mathcal O(1)$ nearest-neighbor swaps in the canonical ordering, together with corresponding updates of $B$.
Consequently, the physical part of implementing an arbitrary primary-primary or secondary-secondary permutation along a fiber of size $\mathcal O(L)$ reduces to implementing a fermionic permutation of the $\mathcal O(L)$ constituent qubits.
Importantly, we may use any fermionic permutation algorithm to realize the required fiberwise permutation.
This includes, for example, FSN based on odd-even
transposition sort, as well as routing algorithms that exploit the all-to-all connectivity available among the qubits associated with a fiber.

To formally designate such permutations between primary or secondary qubits, we introduce the following definition.
\begin{definition}[Layer-preserving permutation]
\label{def:layerPreservingSlotPermutation}
Let
\[
\mathcal S\coloneqq \widetilde G\times[2]
\]
be the set of primary and secondary slots. 
A permutation
\[
\Pi\colon\mathcal S\to\mathcal S
\]
is called layer preserving if there exist permutations
\[
\pi_0,\pi_1\colon\widetilde G\to\widetilde G
\]
such that
\[
\Pi(\mathbf v,r)
=
\bigl(\pi_r(\mathbf v),r\bigr)
\]
for every $\mathbf v\in\widetilde G$ and $r\in[2]$.
\end{definition}
Thus in a layer-preserving permutation, primary slots are permuted only among primary slots, while secondary slots are permuted only among secondary slots.

Analogously to Definition~\ref{def:dcompression}, we can compress the two occupied slots at every vertex $\mathbf v\in B$ by applying
\[
C_B
\coloneqq
\prod_{\mathbf v\in B}
\operatorname{CNOT}_{(\mathbf v,1),(\mathbf v,0)}.
\]
Since the vertex pairs are disjoint, all CNOT gates in $C_B$ may be applied in parallel.

After this compression, each vertex is represented by its primary slot. 
Hence, the ordering of the representative qubits coincides with
the $S$-pattern ordering
\begin{gather*}
  m_r\colon\widetilde G\to[n],\\[0.2em]  
  m_r(\mathbf v) = \mu(\mathbf v) = m_S^k(\mathbf v).
\end{gather*}
of the condensed lattice.

We may therefore apply the hierarchy-switching construction of
Theorem~\ref{thm:dDimTransposition} directly to the representative lattice. 
Consequently, the dimensional hierarchy of the condensed
lattice can be changed in depth
\[
\mathcal O(\log N).
\]

We can now prove our main result of this section.
\begin{theorem}
\label{thm:condensedFermionicRouting}
Let $L=L(N)>2$ be a function of the system size and assume that
\[
d(N)=\mathcal O\left(L(N)\right)
\]
as $L(N)\to\infty$.
Suppose further that there exists a fermionic permutation algorithm such that, for some $m=\mathcal O(n)$, any permutation
\[
    \pi:[m]\to[m]
\]
can be implemented in circuit depth $\mathcal O(T(n))$.
Then any fermionic permutation of the $N$ qubits 
\[
\pi: [N] \to [N] 
\]
can be implemented on an all-to-all-connected qubit architecture in circuit depth
\[
\mathcal O\left(
d(N)\bigl[\log N+T(L(N))\bigr]
\right).
\]
\end{theorem}

\begin{proof}
Consider the corresponding condensed lattice $\widetilde G$, equipped with the two-slot fermionic embedding introduced above.
The lattice $\widetilde G$ has $\mathcal O(d)$ coordinate directions, each of side length $\mathcal O(L)$, except possibly for the variable side length $b$.
However, by Lemma~\ref{lem:condensedLastDimension}, 
\[
b(N)=\mathcal O(L(N)).
\]
Thus, every coordinate direction of $\widetilde G$ has side length $\mathcal O(L)$.

Applying the routing decomposition to $\widetilde G$, any layer preserving permutation $\Pi$ can be written as
\[
\Pi
=
\Phi'_{k+1}\circ \Phi'_k\circ\cdots\circ\Phi'_2
\circ\Phi_1
\circ\Phi_2\circ\cdots\circ\Phi_k\circ\Phi_{k+1},
\]
where $k=\mathcal O(d)$, each $\Phi_i$ and $\Phi'_i$ is a fiberwise layer-preserving permutation along the $i$-th coordinate direction, and $\Phi_{k+1}$ and $\Phi'_{k+1}$ act along the direction of variable side length $b$.
Since all coordinate directions have side length $\mathcal O(L)$, every fiber appearing in this decomposition has length $\mathcal O(L)$.

It remains to bound the cost of a single fiberwise layer-preserving permutation.
To implement such a permutation along a given coordinate direction, we first move that direction to the lowest level of the hierarchy.
As discussed previously, this requires circuit depth $\mathcal O(\log N)$ on an all-to-all-connected qubit architecture.
Once the desired direction is at the lowest level, the permutation along each fiber can be implemented, using the existing fermionic permutation algorithm.
Because each fiber has length $\mathcal O(L)$, this contributes an additional depth of $\mathcal O\left(T(L)\right)$.
Hence, each fiberwise layer-preserving permutation can be implemented in depth
\[
    \mathcal O(\log N+T(L)).
\]

The routing decomposition contains $\mathcal O(d)$ such permutations.
Therefore, an arbitrary layer-preserving permutation of the condensed lattice can be implemented in depth
\[
\mathcal O\left(
d(N)\bigl[\log N+T(L(N))\bigr]
\right).
\]

Finally, by Lemma~\ref{lem:twoSlotRouting}, any permutation in the two-slot embedding can be realized using three layer-preserving permutations of the condensed lattice together with one layer of swaps between the primary and secondary slots of individual vertices.
This contributes only a constant-factor overhead and therefore does not change the asymptotic depth bound.
\end{proof}

As an immediate consequence, this yields a slightly tighter asymptotic bound for arbitrary fermionic permutations than the one obtained in Ref.~\cite{constantinidesLowdepthFermionRouting2025}.

\begin{corollary}
\label{cor:DirectCondensedLatticeRouting}
An arbitrary fermionic permutation
\[
\pi: [N] \to [N],
\]
can be implemented in circuit depth
\[
\mathcal O\left(\frac{\log^2 N}{\log \log N}\right).
\]
\end{corollary}
\begin{proof}
For $N$ qubits, consider the condensed lattice defined by 
\[
L(N) = \log N.
\]
By definition,
\[
d(N) = \frac{\log N}{\log L} = \frac{\log N}{\log \log N}.
\]
We can invoke Theorem~\ref{thm:condensedFermionicRouting} as
\[
d(N)=\mathcal O(L(N))
\]
for $L(N) \to \infty$.
When using odd-even transposition sort, we have
\[
T(L)=\mathcal O(L),
\]
leading to the stated result.
\end{proof}

Corollary~\ref{cor:DirectCondensedLatticeRouting} establishes an important intermediate result, showing that an appropriate choice of the virtual-lattice dimension, adapted to the scaling of the underlying permutation algorithm,
is sufficient to obtain polylogarithmic circuit depth.
The following result shows that this procedure can be iterated recursively, yielding progressively improved depth bounds.

\begin{corollary}
\label{cor:RecursiveCondensedLatticeRouting}
An arbitrary fermionic permutation
\[
\pi: [N] \to [N],
\]
can be implemented in circuit depth
\[
    \mathcal O\left(\log^{3/2} N\right).
\]
\end{corollary}

\begin{proof}
Choose the condensed-lattice parameter
\[
    L(N)=\exp\bigl(\sqrt{\log N}\bigr).
\]
Then
\[
    d(N)
    =\frac{\log N}{\log L(N)}
    =\sqrt{\log N}.
\]
Moreover, since $d(N)=\mathcal O(L(N))$, the condition of Theorem~\ref{thm:condensedFermionicRouting} is satisfied.

Using the routing algorithm of Corollary~\ref{cor:DirectCondensedLatticeRouting}, we may take the coarser bound
\[
    T(L)=\mathcal O(\log^2 L).
\]
For the above choice of $L$,
\[
    \log^2 L=\log N,
\]
and hence
\[
    T(L)=\mathcal O(\log N).
\]
Substituting this together with $d(N)=\sqrt{\log N}$ gives
\[
\begin{aligned}
    \mathcal O\!\left(
        d(N)\bigl[\log N+T(L)\bigr]
    \right)
    &=
    \mathcal O\!\left(
        \sqrt{\log N}\,[\log N+\log N]
    \right) \\
    &=
    \mathcal O\!\left(\log^{3/2}N\right).
\end{aligned}
\]
\end{proof}
Corollary~\ref{cor:RecursiveRouting} can be proven by continuing this recursion.
\begin{proof}[Proof of Corollary~\ref{cor:RecursiveRouting}]
We proceed by induction on $a$.

For $a=1$, the claim follows from the previously established routing algorithm of depth
\[
    \mathcal O(\log^2 N).
\]

Now assume that, for some fixed $a\in\mathbb N_{>0}$, there exists a fermion-routing algorithm whose depth on $n$ qubits is bounded by
\[
T_a(n)
=
\mathcal O\!\left(
    \log^{\frac{a+1}{a}} n
\right).
\]
We use this algorithm as the input routing algorithm in
Theorem~\ref{thm:condensedFermionicRouting} and choose the side length $L=L(N)$ such that
\[
\log L(N)
=
\log^{\frac{a}{a+1}} N,
\]
or equivalently,
\[
L(N)
=
\exp\left(
\log^{\frac{a}{a+1}} N
\right).
\]
For this choice,
\[
\begin{aligned}
T_a(L(N))
&=
\mathcal O\left(
    \log^{\frac{a+1}{a}} L(N)
\right) \\
&=
\mathcal O(\log N).
\end{aligned}
\]
Moreover,
\begin{equation}
\label{eq:RecursiveDimension}
d(N)
=
\frac{\log N}{\log L(N)}
=
\log^{\frac{1}{a+1}} N.  
\end{equation}
Since
\[
L(N)
=
\exp\!\left(
    \log^{\frac{a}{a+1}} N
\right)
\]
grows faster than any power of $\log N$, we have
\[
    d(N)=\mathcal O(L(N)).
\]
Thus, the assumptions of Theorem~\ref{thm:condensedFermionicRouting} are satisfied.

Applying the theorem therefore yields a routing algorithm of depth
\[
\mathcal O\left(
    d(N)\bigl[\log N+T_a(L(N))\bigr]
\right)
=
\mathcal O\!\left(
    \log^{\frac{a+2}{a+1}}N
\right).
\]
This establishes the claim for $a+1$ and completes the induction.
\end{proof}

An important caveat of Corollary~\ref{cor:RecursiveCondensedLatticeRouting} is that the asymptotic improvement obtained by increasing $a$ is only
meaningful as long as the outermost routing problem remains at least two-dimensional, i.e.,
\[
    d(N)\geq 2.
\]
Using Eq.~\eqref{eq:RecursiveDimension}, this condition becomes
\[
    \log^{\frac{1}{a+1}} N \geq 2,
\]
or equivalently,
\[
    a\leq \frac{\log\log N}{\log 2}-1.
\]
Hence, for a fixed system size $N$, increasing the recursion order beyond this threshold no longer yields a meaningful improvement over the corresponding lower-order routing constructions.

Because the admissible recursion order grows only double-logarithmically with the system size, the present construction is best interpreted as an asymptotic bound on the achievable scaling of fermion routing.
Further optimization of the construction, however, may improve its practical relevance for future quantum-simulation experiments.

\section{Periodic boundaries}
In this section we demonstrate how the methods previously designed for lattices with closed boundaries can be extended to lattices with periodic boundaries without altering the asymptotic resource scaling.
More specifically, we consider the same $d$-dimensional lattice as before, but now incorporating the boundary edges
\begin{equation}
\label{eq:boundary_edges}
    ((x_{1}, \ldots, L_{i}-1, \ldots, x_{d}), (x_{1}, \ldots, 0_{i}, \ldots, x_{d}))
\end{equation}
for each dimension $1 \leq i \leq d$.
In this setting, it is helpful to define a shift operator that enables dynamical relocation of the effective boundary of an $S$-pattern, that is,  the endpoints of the JW path.

Consider a 1D ring of $L$ qubits with nearest-neighbor periodic connectivity and the standard canonical Jordan-Wigner encoding $m(i) = i$, where $L$ typically denoted the number of qubits along an axis of the grid. 
In this case the contiguous elements of the canonical ordering align with all physical edges except for the boundary edge $(L-1, 0)$, effectively imposing an open boundary condition on a physically periodic graph. 
To bridge this boundary edge, we introduce a cyclic shift of the JW encoding.
 
\begin{definition}[Cyclic Shift of a 1D Encoding]
    Let $m: \{0, \ldots, k-1\} \to \{0, \ldots, L-1\}$ be a one-dimensional fermionic encoding on $N$ qubits. 
    For any integer shift $j \in \{0, \ldots, L-1\}$, the \emph{cyclic shift} of $m$ by $j$, denoted $m_j$, is defined as
    \begin{equation}
        m_j(i) \coloneqq \big(m(i) - j\big) \bmod L.
    \end{equation}
\end{definition}
The cyclic shift effectively relocates the unmapped edge from $(L-1, 0)$ to $(j-1, j)$, allowing local interaction terms across the physical boundary at the cost of breaking locality at site $j$. Crucially, this transformation can be implemented efficiently on hardware:

Crucially, this transformation can be efficiently implemented on hardware by utilizing Lemma~\ref{lem:columnSwap}. 
To this end, we partition the qubit register into two subsets, $A = \{i : m(i) \in \{0, \ldots, j-1\}\}$ and $B = \{i : m(i) \in \{j, \ldots, L-1\}\}$, which under the canonical ordering consist of all qubits in set $A$ followed sequentially by those in set $B$. According to Lemma~\ref{lem:columnSwap}, we can transform to a $B$-first ordering by executing $\mathrm{CNOT}$ chains  within each subset, terminating at the qubits with initial ranks $m(j-1)$ and $m(j)$, applying a single $\mathrm{CZ}$ gate between these boundary qubits, and subsequently reversing the $\mathrm{CNOT}$ chains.

In the following, we extend the definition of shifting to $d$-dimensional $S$-patterns and show how these can be constructed efficiently.

\begin{definition}[Shifted $d$-Dimensional $S$-pattern]
    Let $m^{d}_{S}(x_1, \ldots, x_d)$ denote the $d$-dimensional $S$-pattern from Definition~\ref{def:dD_S_pattern} on a lattice with periodic boundaries. 
    For a shift vector $\boldsymbol{\Delta} = (\Delta_{1}, \ldots, \Delta_{d})$ with $\Delta_{i} \in [L]$, the \emph{shifted $S$-pattern} is defined as
    \begin{equation}
        m^{d}_{S}\big((x_{1}-\Delta_{1}) \bmod L, \ldots, (x_{d}-\Delta_{d}) \bmod L\big).
    \end{equation}
\end{definition}

\begin{lemma}
\label{lem:Spattern_shift}
A $d$-dimensional $S$-pattern on a $d$-dimensional periodic qubit lattice can be shifted by any $\mathbf{\Delta}$ with $\mathcal{O}(N)$ nearest-neighbour qubit gates in $\mathcal{O}(N^{1/d})$ depth.
\end{lemma}
\begin{proof}
Starting from the unshifted $d$-dimensional $S$-pattern $m^{d}_{S}(x_{1}, \ldots, x_{d})$, we shift the first dimension by $\Delta_{1}$ by applying the 1D cyclic shift on every coordinate line parallel to the axes of the first dimension.
We then compress the first dimension, noting that the compression operation is invariant to the cyclic shift of the previous step. 
In the compressed basis, we apply 1D cyclic shift to the representatives to shift the second dimension by $\Delta_{2}$ and subsequently compress the second dimension.
This alternating sequence of shifting and compression is repeated until all $d$ dimensions have been shifted. 
Finally, the compression operations for all dimensions are inverted in reverse order.
This procedure yields the $S$-pattern shifted by $\mathbf{\Delta}$ where compressing and shifting of each dimension has the same asymptotic scaling of $\mathcal{O}(N)$ nearest-neighbour qubit gates with $\mathcal{O}(N^{1/d})$ depth.
\end{proof}

Finally, the shifted $S$-pattern can be used to conserve the locality of fermionic interactions in lattices with periodic boundaries. 
This method mirrors to the construction in Theorem~\ref{thm:dDlocality}, however, in addition to shifting the partitions, we also shift the intermediate $S$-pattern, resulting in a shifted boustrophedon encoding.

\begin{theorem}\label{thm:dDlocality_periodic}
    It is possible to implement all pairwise fermionic interactions between geometrically local modes $\|x_i-x_j\|\le \delta$ on the nearest-neighbor connected qubit lattice $G_\mathrm{dD}$ with periodic boundaries using $\mathcal{O}(N)$ gates in $\mathcal{O}(N^{\frac{1}{d}})$ depth.
\end{theorem}
\begin{proof}

We define $2^{d}$ sets of subgrids $\mathcal{G}_{b_{1}, \dots, b_{d}}$ along Definition~\ref{def:subgrid} and the one dimensional partition
\[
\begin{aligned}
    \mathcal{G}_{0}&=\{I_0, I_{2\delta}, \ldots\},\\
    \mathcal{G}_{1}&=\{I_{-\delta}, I_{\delta}, \ldots\},
\end{aligned}
\]
but change the definition of the one-dimensional intervals to
\[
I_j := \{k : j\mod L \le k \le (j+2\delta-1)\mod L\},
\]
to make the partitions and subgrids periodic.
From the partitions we define $2^{d}$ different boustrophedon encodings according to Definition~\ref{def:dD_boustrophedon}.
This allows us to implement all pairwise interactions between geometrically local modes with the following procedure.\\
For each $b_{j}\in \{0,1\}$:
\begin{enumerate}
    \item Shift the encoding $m_{S}^{\mathrm{d}}(x_{1}, \dots, x_{d})$ to 
    \begin{equation*}
        m_{S}^{\mathrm{d}} \big((x_{1}-b_{1}\delta) \bmod L, \dots, (x_{d}-b_{d}\delta) \bmod L\big)
    \end{equation*}
    using Lemma~\ref{lem:Spattern_shift} and then switch to $m_{\mathcal{G}_{b_{1}}, \ldots, \mathcal{G}_{b_{d}}}(x_{1}, \dots, x_{d})$ using Lemma~\ref{lem:dDBousSwitch}.
    \item Apply an FSN within each subgrid in $\mathcal{G}_{b_{1},\ldots, b_{d}}$ for which this encoding follows a Hamiltonian path.
    \item Switch back to $m_{S}^{\mathrm{d}}(x_{1}, \dots, x_{d})$.
\end{enumerate}
By Lemma~\ref{lem:localPartition}, every interaction is covered in at least one of the subgrids.
The asymptotic gate-count and depth scaling are dominated by the cost of switching between the boustrophedon encodings via Lemma~\ref{lem:dDBousSwitch} and the $S$-pattern shifting of Lemma~\ref{lem:Spattern_shift}.
\end{proof}

\section{Example applications}
\label{app:lattice_examples}

\begin{figure}
    \centering
        \begin{quantikz}[column sep=2mm, wire types={q,q}]
        & \gate{X} \vqw{1}
        & \gate{R_Y(2\theta_{j(j+1)})} 
        & \gate{Z}
        & \gate{H}
        & \gate{S}
        \\ 
        & \gate{Y}
        & \gate{R_Z(2\theta_{j(j+1)})}
        & \gate{Z} \vqw{-1}
        & \gate{S^\dagger}
        & \gate{H}
        \\ 
        \end{quantikz}
    \caption{
    Circuit implementation of a fermionic swap gate between qubits corresponding to adjacent spin-orbitals $f^{j,j+1}_\text{swap}$, which additionally implements the hopping evolution $\exp( {i\theta_{j(j+1)}(a^{\dagger}_{j}a_{j+1}+a^{\dagger}_{j+1}a_{j})} )$. 
    The first and second two-qubit  operations correspond to Clifford gates $\exp({i\frac{\pi}{4}(II-XI)(II-IY)})$ and $\exp({i\frac{\pi}{4}(II-ZI)(II-IZ)})$ respectively.
}
    \label{fig:fswap_XY}
\end{figure}

In this section we provide a detailed description of the example applications showcased in the main text.

\subsubsection*{\textsf{FFFT} }

\begin{figure}
    \centering
    \includegraphics[width = \columnwidth]{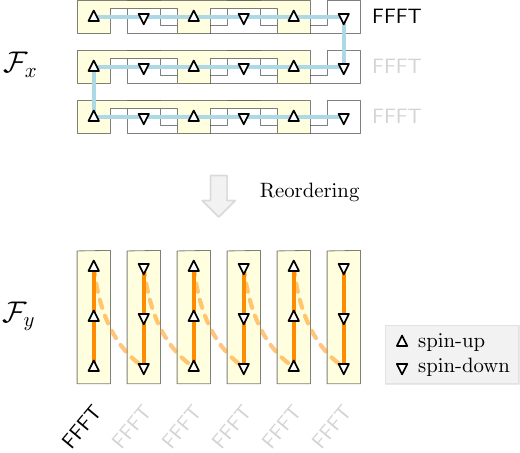}
    \caption{Two-dimensional \textsf{FFFT} for spinful problems.  The 2D \textsf{FFFT} is decomposed into successive one-dimensional transforms, $\mathcal{F}=\mathcal{F}_x\mathcal{F}_y$.
    These act along the $x$ and $y$ dimensions and are efficiently parallelized using $S$ and $Z$ patterns, respectively.
    For spinful systems, we interleave spin-up and spin-down species along the columns. The fermionic fast Fourier transformation is applied separately on each spin species.
}
    \label{fig:FFFTExplanation}
\end{figure}

The currently most efficient known method for implementing a two-dimensional
\textsf{FFFT} under local connectivity constraints is based on
Givens-rotation networks~\cite{kivlichanQuantumSimulationElectronic2018}.
These networks can implement arbitrary orbital basis changes of the form
\[
\hat{U}_{\mathrm{orb}}^\dagger \, a_p^\dagger \, \hat{U}_{\mathrm{orb}}
=
\sum_q U_{q p}\, a_q^\dagger ,
\qquad
\hat{U}_{\mathrm{orb}}^\dagger \, a_p \, \hat{U}_{\mathrm{orb}}
=
\sum_q U_{q p}^*\, a_q,
\]
using a circuit consisting of \(N(N-1)/2\) Givens rotations with depth \(2N\).

In many cases, such Givens-rotation networks can be simplified by applying
separate networks to the spin-up and spin-down orbitals. However, the state
preparation protocols introduced in Ref.~\cite{maskaraFastSimulationFermions2025},
which are particularly promising for combination with our Hubbard simulations,
require spin-up and spin-down orbitals to be adjacent in the canonical ordering.
This requirement arises from the Bogoliubov, or pairing, rotation
\[
U_{pq}(\theta,\phi)
=
\exp\!\left[
\theta
\left(
e^{i\phi} c_p^\dagger c_q^\dagger
-
e^{-i\phi} c_q c_p
\right)
\right],
\]
which acts between modes belonging to different spin sectors.

With these considerations in mind, Fig.~\ref{fig:FFFT} compares the resource requirements of implementing the two-dimensional \textsf{FFFT} using a single Givens-rotation network over the full qubit register with those of implementing it via smaller Givens-rotation networks for one-dimensional \textsf{FFFT}s, together with an intermediate reordering step. 
A more detailed visualization of our apprach is provided in Fig.~\ref{fig:FFFTExplanation}.

\subsubsection*{Spinful lattice models}

The general strategy enabled by the boustrophedon encodings is to map a local fermionic interaction graph to a local qubit interaction graph, and then implement those interactions via specialized FSNs aligned with the path of the boustrophedon.

We focus on implementing a second-order Trotter step due to its simplicity and relevance for near term experiments \cite{hemeryMeasuringLoschmidtAmplitude2024, alamProgrammableDigitalQuantum2025a}.
Such a Trotter step is defined as
\begin{equation*}
\mathcal{S}_2(t)\coloneqq e^{(t/2)H_1}\cdots e^{(t/2)H_\Gamma}\,e^{(t/2)H_\Gamma}\cdots e^{(t/2)H_1}, 
\end{equation*}
for the Hamiltonian
\begin{equation*}
    H=\sum_{\gamma=1}^\Gamma H_\gamma
\end{equation*}
consisting of $\Gamma$ terms $H_1, ..., H_{\Gamma}$.

For our resource estimations, we assume a standard gate decomposition, where a hopping interaction between neighboring orbitals, as well as a FSWAP gate, can be implemented using $2$ CNOTs. Furthermore, a gate decomposition like the one shown in Fig.~\ref{fig:fswap_XY} allows us to perform an FSWAP plus hopping interaction between neighboring orbitals. 
This decomposition requires only $2$ two-qubit Pauli gates, which can also be easily transformed into CNOTs via some Hadamard and $S$ gates.

\begin{figure*}
    \centering
    \includegraphics[width=.8 \linewidth]{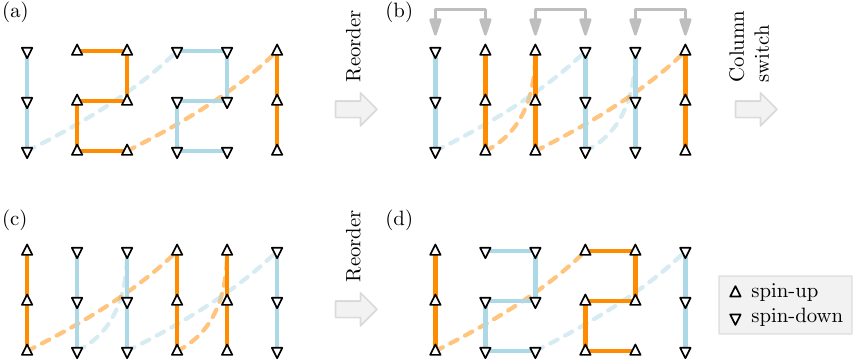}
    \caption{
    (a) For spinful models, we employ two separate Jordan-Wigner (JW) encodings (visualized in blue and orange). Consequently, only modes belonging to the same spin species can be simultaneously local within a single boustrophedon encoding. 
    (b) During the encoding switch (illustrated in Fig.~\ref{fig:3}), a transient $Z$-pattern configuration is generated. At this stage, inter-column qubit SWAP gates are executed—adding only a single parallel SWAP layer (indicated by gray arrows) to the standard reconfiguration sequence—to yield configuration (c). 
    (d) A final reordering step prepares the encoding complementary to (a), enabling access to the remaining inter-species interactions and together spanning the full square grid connectivity.
}
    \label{fig:LocalReconfigurationSwaps}
\end{figure*}

Many physically relevant fermionic systems are spinful, i.e., they have two internal species 
$\sigma \in \{\uparrow, \downarrow\}$. In a JW encoding, each spin-orbital corresponds to a fermionic mode and is mapped to its own qubit, so spin-up and spin-down modes are represented separately. In the standard Hubbard model the kinetic term conserves spin, so particles hop only between sites within the same spin sector ( $\ket{\uparrow}_i \leftrightarrow \ket{\uparrow}_j$ and $\ket{\downarrow}_k \leftrightarrow \ket{\downarrow}_l$). Figure \ref{fig:LocalReconfigurationSwaps} illustrates a simple addition to our encoding switching strategy that allows us to include spinful models in simulation.

\subsubsection*{Nearest-neighbor hopping on a square lattice}

\begin{figure*}
    \centering
    \includegraphics[width = \linewidth]{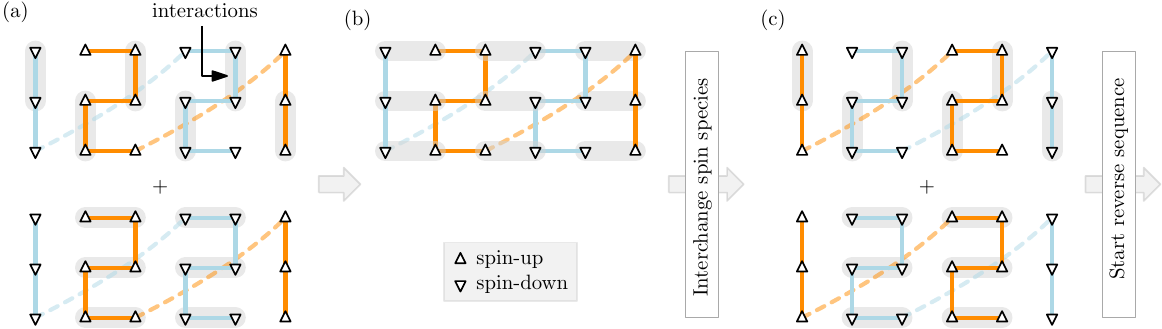}
    \caption{
    Implementation of a second-order Trotter step for the nearest-neighbor (NN) Hubbard model on a grid lattice. 
(a) The step initiates with hopping terms that are local within the initial boustrophedon encoding. 
(b) Next, the inter-species density-density interactions are executed. 
(c) The encoding is then reconfigured, as depicted in Fig.~\ref{fig:LocalReconfigurationSwaps}, to enable the remaining 'orthogonal' set of hopping terms. To complete the second-order Trotter step, the reverse sequence is subsequently executed. To maximize circuit efficiency, the first step of the reverse sequence can be merged with the preceding step shown in panel (c). Analogously, the final hopping terms of a given Trotter step can be fused with the initial hopping terms of the subsequent step.
    }
    \label{fig:FermiHubbard}
\end{figure*}

\begin{figure*}
    \centering
    \includegraphics[width = .8\linewidth]{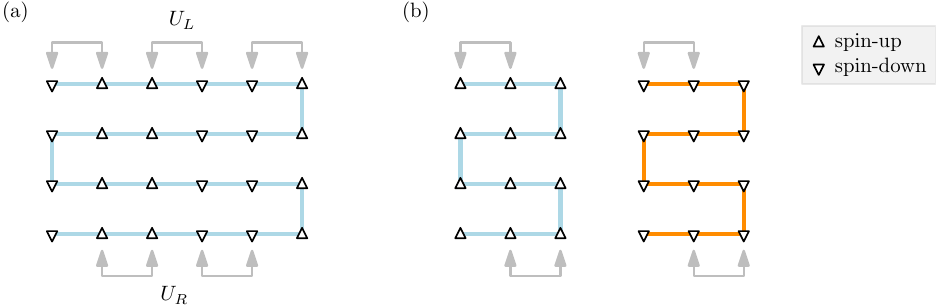}
    \caption{
    Implementation of hopping terms for the nearest-neighbor (NN) Hubbard model on square grid via fermionic swap networks (FSNs). 
    (a) The standard strategy \cite{kivlichanQuantumSimulationElectronic2018} applies fermionic swaps between adjacent columns, initiating from column 0 ($U_L$) or column 1 ($U_R$), respectively. Each $U_L U_R$ operation shifts a new column to the right and left boundaries of the grid, returning all columns to their initial positions after $N_C$ applications (where $N_C$ denotes the number of columns). For modes at the grid boundaries, the $S$-pattern encoding enables the implementation of the upward and downward hopping terms. Repeating this cycle $N_C-1$ times yields all required hopping terms for the NN square lattice. Crucially, this strategy requires only 1D nearest-neighbor qubit connectivity. 
    (b) Alternatively, given ladder connectivity, or a 2D grid in which such a strategy can be embedded with minimal overhead \cite{cadeStrategiesSolvingFermiHubbard2020}, one can execute parallel swap networks for each spin species separately while implementing the inter-species density-density interactions via the transverse connections of the ladder. This approach effectively reduces $N_C$ by half, approximately halving the overall cost of a Trotter step. 
    For a fair comparison when evaluating the second-order Trotter step of Fig.~\ref{fig:HubbardComp}, we assume that the initial and final hopping terms of the depicted FSNs are merged (fused) with the preceding and succeeding steps, respectively.
    }
    \label{fig:FermiHubbardFSN}
\end{figure*}

The Hubbard model on an NN-connected square serves as the simplest model that properly describes the most relevant and essential physics of cuprates \cite{avellaEmeryVsHubbard2013}. Currently, it is the most widely performed simulation of a Hubbard model on digital quantum computers \cite{alamProgrammableDigitalQuantum2025a, alamFermionicDynamicsTrappedion2025, nigmatullinExperimentalDemonstrationBreakeven2025, hemeryMeasuringLoschmidtAmplitude2024}, due to its simplicity.

A textbook example of an NN-connected Hubbard model is the $t-U$ model

\begin{equation*}
H=-t\sum_{\langle i,j\rangle,\sigma}\left(c_{i,\sigma}^{\dagger}c_{j,\sigma}+\mathrm{h.c.}\right)
+U\sum_i n_{i,\uparrow}n_{i,\downarrow}.
\end{equation*}
On a high level, this model captures the competition between kinetic delocalization and on-site repulsion. There is particle–hole symmetry at half filling, making $n=\langle n_{i,\uparrow} + n_{i,\downarrow}\rangle$ a particularly symmetric point.

For the repulsive model ($U > 0$), 
quantum Monte Carlo (QMC) and related methods are sign-problem-free at half filling on bipartite lattices, such as an NN square, but generally develop a fermion sign problem upon doping, making the 2D doped regime challenging to simulate classically.

Fig.~\ref{fig:FermiHubbard} illustrates a second-order Trotter step for implementing an NN lattice model via our dynamic encoding. For comparison, we give a fully FSN-based strategy, representing the current state of the art in Fig.~\ref{fig:FermiHubbardFSN}. We note that for the latter method this NN square model can naturally integrate the vertical hopping terms between parallel column-wise swaps. However, implementing a more complicated interaction graph, like given by the NNN example or the Lieb lattice, would be much less efficient.

To arrive at the gate-counts reported in Table \ref{tab:gate-counts}, consider that, in leading order, the encoding switch shown in Fig.~\ref{fig:LocalReconfigurationSwaps} can be executed with $6.5N$ CNOTs, the density-density interactions with $N$ CNOTs and the hopping terms with $4N$ CNOTs. Then for executing the second-order Trotter step, as shown in Figure~\ref{fig:FermiHubbard}, while absorbing two of the hopping steps into the previous ones, as described in the caption, we arrive at a gate count of $21N$. The corresponding depth is dominated by encoding switches, where each switch has a CNOT depth of $2N_R$, with $N_R$ being the number of rows of the qubit lattice. Interestingly, this is the same depth scaling as one would expect from a single Pauli gadget implementing a vertical hopping term. We would like to explicitly note that we neglect the extra circuit cost incurred by the first Trotter step, where no terms can be absorbed into a previous step, but do the same for the literature methods.

\subsubsection*{Next-nearest-neighbor hopping on a square lattice}

In this setting, NNN hopping corresponds to diagonally adjacent unit cells on the lattice and the standard $t-t'-U$ Hubbard model

\begin{equation*}
\begin{aligned}
H=&-t\sum_{\langle i,j\rangle,\sigma}\left(c^\dagger_{i\sigma}c_{j\sigma}+\mathrm{h.c.}\right)
-t'\sum_{\langle\!\langle i,j\rangle\!\rangle,\sigma}\left(c^\dagger_{i\sigma}c_{j\sigma}+\mathrm{h.c.}\right)\\
&+U\sum_i n_{i\uparrow}n_{i\downarrow}.
\end{aligned}
\end{equation*}
Research has shown that NNN hopping in a square lattice Hubbard model is critical for a superconducting phase \cite{jiangSuperconductivityDopedHubbard2019, simonscollaborationonthemany-electronproblemAbsenceSuperconductivityPure2020, tasakiHubbardModelIntroduction1998}.
Here, particle–hole symmetry is broken for $t'\neq 0$, so electron-doped and hole-doped sides are no longer mirror images. This also means that the sign problem appears even at half filling, making the model significantly harder to simulate classically.

\begin{figure}
    \centering
    \includegraphics[width = \columnwidth]{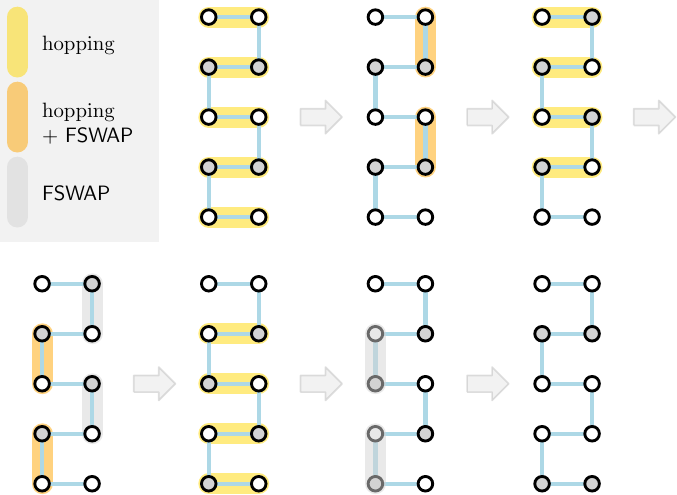}
    \caption{
    Implementation sequence (from left to right) of the next-nearest-neighbor (NNN) Hubbard Hamiltonian on a square lattice, corresponding to a single local snake pattern. The color coding of the fermionic modes (filling color of nodes) serves to trace their transport during fermionic swaps. Throughout this sequence, each mode interacts with $4$ of its $8$ NNN sites. After executing a single encoding switch, the remaining $4$ out of $8$ NNN interactions are implemented.
    }
    \label{fig:FHNNN}
\end{figure}

The implementation of this model can be performed similarly to the previous example, but this time one needs to perform more hopping interactions in both encodings. We provide one possible sequence in Fig.~\ref{fig:FHNNN}. In order to cancel some of the FSWAP gates that do not implement a Hamiltonian interaction, it is beneficial to perform the reverse sequence in the first boustrophedon encoding [corresponding to Fig.~\ref{fig:FermiHubbard} (a)]. Further, we can even define the starting position of nodes to fit this reverse sequence, which allows to fully remove these FSWAP gates.

\subsubsection*{Lieb lattice}

\begin{figure}
    \centering
    \includegraphics[width = \columnwidth]{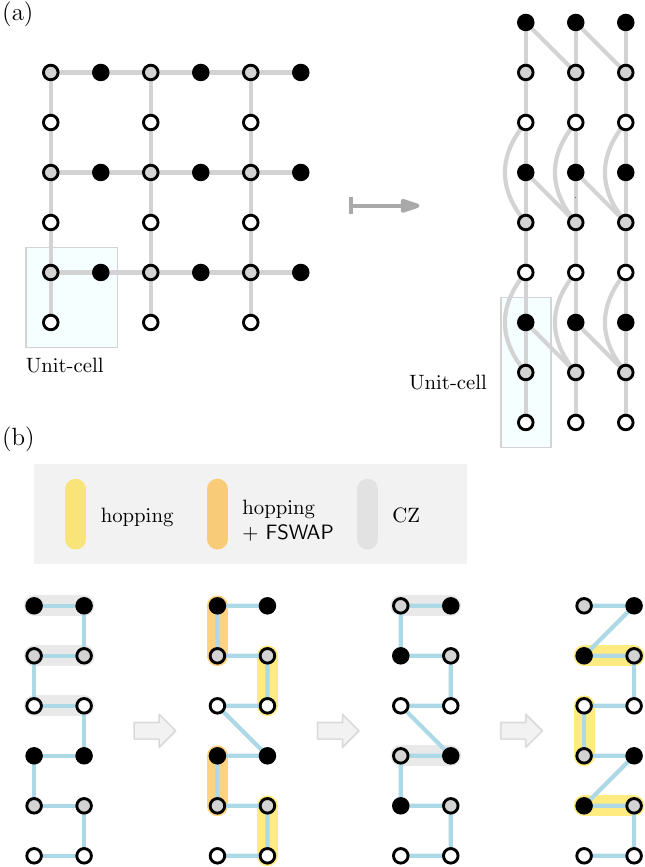}
    \caption{
    Description of the Lieb lattice simulation. 
    (a) Tight embedding of the Lieb lattice onto a 2D square qubit lattice. Note that to incorporate spin, each column (of the embedded graph) would be represented by two consecutive columns of a qubit array.
    Here, the color coding white, grey, and black acts as a guideline to easier follow information flow.
    (b) Gate sequence implementing hopping terms along a single snake path. The layout takes advantage of $\mathrm{CZ}$ gates performing logical swaps, as depicted in Fig.~\ref{fig:1}. By executing the reverse sequence within the initial boustrophedon encoding, the modes are systematically returned to their original positions.
    }
    \label{fig:FHLieb}
\end{figure}

The Lieb lattice serves as our minimal example of a three-band model, motivated by its role as the natural geometry for the simulation of the Emery model, a canonical model of cuprate superconductors.

However, the standard Emery model also includes additional NNN interactions~\cite{avellaEmeryVsHubbard2013}, corresponding to direct oxygen–oxygen hopping, which we neglect in our minimal simulation.

Figure~\ref{fig:FHLieb} illustrates the implementation of a Trotter step. Unlike previous cases, we must embed the Lieb lattice into a square grid, which we achieve by deforming the unit cell to enable dense packing. This does increase the range of the fermionic interactions somewhat and requires a $3L \times 2L$ qubit grid, where the $3L$ factor also affects the depth of our encoding switches.

\subsubsection*{Kagome lattice}

The kagome lattice contains triangles (odd cycles), so bipartite-lattice simplifications do not apply even when only considering NN interactions. As the lattice also possesses 3 sites per primitive unit cell, any Hubbard model has three bands.

\begin{figure}
    \centering
    \includegraphics[width = \columnwidth]{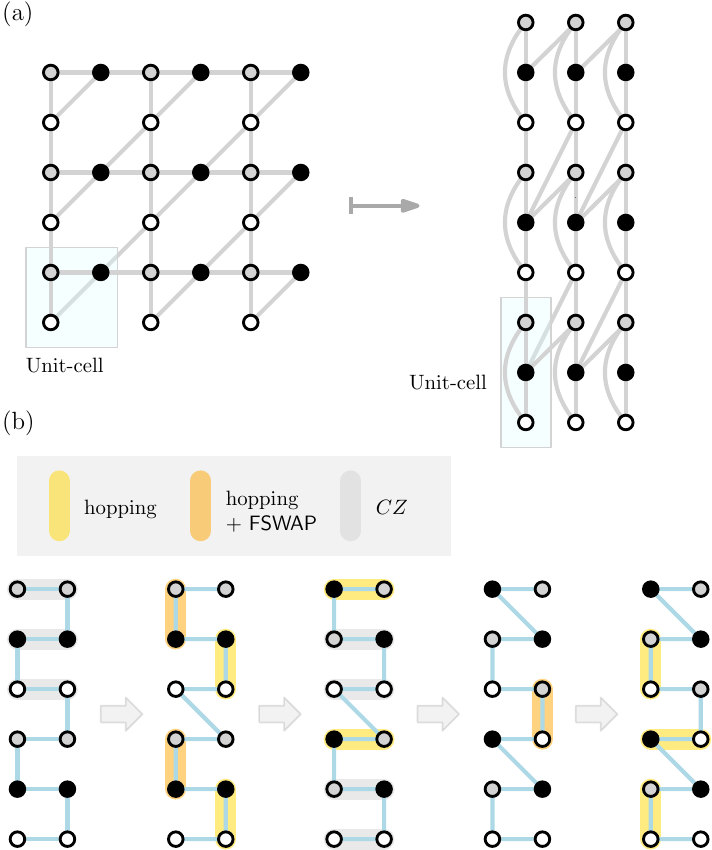}
    \caption{
    Steps for simulating a kagome lattice. 
    (a) Tight embedding of the kagome lattice into a 2D square lattice. 
    Here, the color coding white, grey, and black acts as a guideline to easier follow information flow.
    (b) Circuit implementing hopping terms along a single snake path under a boustrophedon encoding. 
    By applying the reverse sequence within the initial encoding, nodes are systematically returned to their original positions.
    }
    \label{fig:FHKagome}
\end{figure}

Similar to before, Fig.~\ref{fig:FHKagome} shows a possible sequence of gates that perform the necessary hopping interactions inside one snake-pattern of one boustrophedon encoding. The embedding of the kagome lattice into the square lattice is performed similar to the Lieb lattice.

\end{document}